\documentclass[a4paper,10pt]{article}
\usepackage{amssymb}
\usepackage{amsthm}
\usepackage{amsmath,empheq}
\usepackage{tikz}
\usepackage{pgfplots}

\newtheorem{tm}{Theorem}
\newtheorem{lm}[tm]{Lemma}
\newtheorem{pr}[tm]{Proposition}

\theoremstyle{definition}
\newtheorem{df}[tm]{Definition}
\newtheorem{as}[tm]{Assumption}
\newtheorem*{ack}{Acknowledgements}

\theoremstyle{remark}
\newtheorem{rk}[tm]{Remark}

\newcommand{\dd}{\mathrm{d}}

\DeclareMathOperator{\wsc}{\hspace{1pt}\rightharpoonup\hspace{-11pt}{^\ast}\hspace{5pt}}

\numberwithin{tm}{section}
\numberwithin{equation}{section}
\numberwithin{figure}{section}

\begin{document}
\title{On self-similar solutions to a kinetic equation arising in weak turbulence theory for the nonlinear Schr\"odinger equation}
\author{A.H.M.~Kierkels\hspace{12pt} J.J.~L.~Vel\'azquez\\\small Institute for Applied Mathematics, University of Bonn\\\small Endenicher Allee 60, 53115 Bonn, Germany\\\small kierkels@iam.uni-bonn.de\hspace{12pt}velazquez@iam.uni-bonn.de}
{\linespread{1}
\maketitle
\begin{abstract}
We construct a family of self-similar solutions with fat tails to a quad\-ratic kinetic equation. This equation describes the long time behaviour of weak solutions with finite mass to the weak turbulence equation associated to the nonlinear Schr\"odinger equation. The solutions that we construct have finite mass, but infinite energy. In {\em J.~Stat.~Phys.}~{\bf159}(3):668-712, self-\linebreak[2]similar solutions with finite mass and energy were constructed. Here we prove upper and lower exponential bounds on the tails of these solutions.\\


\end{abstract}}

\section{Introduction}
The theory of weak turbulence is a physical theory which describes the transfer of energy between different wavelengths in a large class of wave systems. This theory can be applied to homogeneous problems that can be approximated to leading order by a system of linear waves that interact by means of weak nonlinearities. The basic mathematical model in the theory of weak turbulence is a kinetic equation that describes the transfer of energy between different wavelengths. Contrary to the starting wave equations, the kinetic equations arising in weak turbulence theory exhibit irreversible behaviour. Examples of applications of the theory of weak turbulence to specific physical systems can be found in \cite{DJR06}, \cite{DNPZ92}, \cite{H62,H63}, \cite{N11}, \cite{P29}, \cite{Z67} and \cite{ZF67}.\\

One of the most extensively studied systems in the setting of weak turbulence theory is the one associated to the nonlinear Schr\"odinger equation
\begin{equation}\label{Schroedinger}
iu_t=-\Delta u+\varepsilon|u|^2u,
\end{equation}
with $\varepsilon>0$ small (cf.~\cite{DNPZ92}, \cite{N11}, \cite{ZLF92}). Denoting $F(t,k)=|\hat{u}(t,k)|^2$, where $\hat{u}$ is the space Fourier transform of the solution of \eqref{Schroedinger}, then restricting to isotropic solutions one obtains up to rescaling the following equation for $f(t,\omega):=F(t,k)$ with $\omega:=|k|^2$:
\begin{equation}\label{eq:S1E2}
\partial_t f_1=\frac12\iint_{[0,\infty)^2}W\big[(f_1+f_2)f_3f_4-(f_3+f_4)f_1f_2\big]\dd\omega_3\dd\omega_4,
\end{equation}
where $f_i=f(t,\omega_i)$ for each $i\in\{1,2,3,4\}$, where $\omega_2=(\omega_3+\omega_4-\omega_1)_+$, and where $W=\min_i\{\sqrt{\omega_i}\}/\sqrt{\omega_1}$. The mathematical theory of this equation has been studied in detail in \cite{EV15} where several properties of the solutions of \eqref{eq:S1E2} were obtained. As it is more convenient to study the evolution of the mass density function $g(t,\omega)=\sqrt{\omega}f(t,\omega)$, we reformulate \eqref{eq:S1E2} as
\begin{multline}\label{eq:S1E3}
\partial_t g_1=\frac12\iint_{[0,\infty)^2}\tilde{W}\left[\left(\frac{g_1}{\sqrt{\omega_1}}+\frac{g_2}{\sqrt{\omega_2}}\right)\frac{g_3g_4}{\sqrt{\omega_3\omega_4}}\right.\\-\left.\left(\frac{g_3}{\sqrt{\omega_3}}+\frac{g_4}{\sqrt{\omega_4}}\right)\frac{g_1g_2}{\sqrt{\omega_1\omega_2}}\right]\dd\omega_3\dd\omega_4,
\end{multline}
where now $g_i=g(t,\omega_i)$ for each $i\in\{1,2,3,4\}$, where $\omega_2=(\omega_3+\omega_4-\omega_1)_+$ and where $\tilde{W}=\min_i\{\sqrt{\omega_i}\}$. Formally this equation has two conserved quantities, namely the mass, which is the integral of $g$, and the energy, which is the first moment of $g$.

For a class of weak solutions to \eqref{eq:S1E3} with finite mass, it was proved in \cite{EV15} that all solutions converge, as $t\rightarrow\infty$, to a Dirac mass supported at a well-defined point $a\geq0$, which depends only on the support of the initial distribution. It turns out that unless the initial distribution is contained in a periodic lattice, there holds $a=0$. In this last case, it is possible to formally derive an equation that describes the behaviour of the fraction of mass that is not supported near the origin, which we denote by $G$. Formally this equation reads as
\begin{multline}\label{eq:strongform}
\partial_tG(\omega)=\frac12\int_0^\omega\frac{G(\omega-\xi)G(\xi)\dd \xi}{\sqrt{(\omega-\xi)\xi}}-\frac{G(\omega)}{\sqrt{\omega}}\int_0^\infty\frac{G(\xi)\dd \xi}{\sqrt{\xi}}\\
-\frac12\frac{G(\omega)}{\sqrt{\omega}}\int_0^\omega\left[\frac{G(\omega-\xi)}{\sqrt{\omega-\xi}}+\frac{G(\xi)}{\sqrt{\xi}}\right]\dd\xi+\int_0^\infty\frac{G(\omega+\xi)}{\sqrt{\omega+\xi}}\left[\frac{G(\omega)}{\sqrt{\omega}}+\frac{G(\xi)}{\sqrt{\xi}}\right]\dd\xi.
\end{multline}
in which one may recognize a coagulation-fragmentation equation with nonlinear fragmentation. Note that many terms in \eqref{eq:strongform} are singular and the meaning this equation has to be precised. A more elaborate discussion on the sense in which $G$ describes the asymptotic behaviour of $g$ can be found in \cite{EV15} and \cite{KV15}. It is known that solutions to \eqref{eq:S1E3} can contain Dirac masses at the origin. If that is the case, then \eqref{eq:strongform} can be obtained by only considering those collisions which are mediated by an interaction with the condensate.

Notice that if we assume that $g=M\delta_0+G$, then the energy of $g$ is contained in $G$. Therefore the analysis of the long time behaviour of $G$ is relevant, even though the mass of $G$ becomes negligible compared to the mass supported at the origin as $t\rightarrow\infty$. As conjectured in \cite{EV15}, we expect a self-similar distribution of the energy among the different wavelengths as $t\rightarrow\infty$, provided that $g$ has initially finite energy. In \cite{KV15} we have proved the existence of a family of self-similar solutions $G$ to \eqref{eq:strongform} with finite energy. These solutions are the natural candidates for describing the long time asymptotics of solutions $g$ to \eqref{eq:S1E3} with finite mass and energy. Note however, that stability of these self-similar solutions is an open problem, even at the linearised level.

Of course, the assumption of finite energy is not really needed to have self-similar behaviour for the solutions $g$ of \eqref{eq:S1E3}. Long time self-similarity can be expected if the initial data has a power law tail as $\omega\rightarrow\infty$. This gives a natural scaling law relating the energy density and $\omega$. Actually, long time self-similarity can only be expected if either the initial distribution has a power law tail, or if the energy is finite. This is because otherwise the behaviour of the solutions is not stable for large values of $\omega$ under the evolution equations \eqref{eq:S1E3} and \eqref{eq:strongform}. This is reminiscent of the situation for the coagulation equation with constant kernel, where in order to have self-similarity the power law behaviour for the initial data is needed (cf.~\cite{MP04}).\\

Let us briefly discuss the expected self-similar behaviour of a solution $g$ to \eqref{eq:S1E3} if we assume the initial distribution to behave like $\omega^{-\rho}$ for large values of $\omega$, where $\rho>0$. Given that for $\rho\leq\frac12$ it is not clear whether the collision terms in \eqref{eq:S1E3} can be given a meaning, we restrict ourselves to the case $\rho>\frac12$. A particularly relevant exponent is $\rho=\frac23$, which is the so called Kolmogorov-Zakharov exponent for \eqref{eq:S1E3}. The interpretation of this exponent is the existence of a constant flux of particles from large values of $\omega$ to smaller ones in the space of frequencies (cf.~\cite{Z67}).

If we suppose that $\rho>1$, then $g$ has finite mass and the heuristic derivation of \eqref{eq:strongform} is valid (cf.~\cite{EV15}, \cite{KV15}). As discussed before, we thus expect the self-similar solutions $G$ of \eqref{eq:strongform} to describe the long time asymptotics of solutions $g$ to \eqref{eq:S1E3}. One of the main results of this paper will be the proof of existence of self-similar solutions to \eqref{eq:strongform} with tail behaviour $\omega^{-\rho}$ for $1<\rho<2$. If $\rho>2$, then the solutions have finite energy. Existence of self-similar solutions with finite energy has been proved in \cite{KV15}. In this paper we prove that these solutions decay exponentially as $\omega\rightarrow\infty$.

The case $\frac12<\rho<1$ is different, since the mass of $g$ is infinite. We therefore expect the amount of mass located at the origin to grow without limit. Dimensional analysis suggests that the long time asymptotics of solutions to \eqref{eq:S1E3} then cannot be approximated by solutions to a simpler quadratic equation, similar to \eqref{eq:strongform}, where all the collisions are mediated by interaction with one particle placed at the origin. More precisely, if we suppose that $g=M\delta_0+G$, where $M$ is the amount of mass located at the origin, then there are terms that are cubic in $G$ that cannot be ignored, and we expect self-similar solutions of \eqref{eq:S1E3} to be of the form 
\begin{equation}\nonumber
M(t)\delta_0(\omega)+\frac1{t^{\frac\rho{2\rho-1}}}\Phi\left(\frac\omega{t^{\frac1{2\rho-1}}}\right).
\end{equation}
We further note that dimensional analysis alone is insufficient to determine the exact scaling law for $M$. However, it suggests that $M\sim t^\alpha$ for $\alpha\leq\frac{1-\rho}{2\rho-1}$.\\

Seemingly the first paper to consider the asymptotics of solutions of \eqref{eq:S1E3} in connection with solutions to the nonlinear Schr\"odinger equation is \cite{P92}. In particular, that paper describes the scaling properties of solutions $g$ to \eqref{eq:S1E3} in the cases where either the energy is finite, or where $g$ behaves for large frequencies according to the Kolmogorov-Zakharov power law. These two cases correspond to assuming that the spectral distribution $F(t,k)=|\hat{u}(t,k)|^2$ has either finite energy or decays according to the Kolmogorov-Zakharov exponent for large $|k|$. However, in the case of infinite energy there is no particular reason for the exponent of the power law to be restricted to this one. Hence, it makes sense to study solutions $g$ to \eqref{eq:S1E3} where $g$ initially has arbitrary power law behaviour at infinity, at least from a mathematical point of view.

This paper is a continuation of the study of self-similar solutions to \eqref{eq:strongform}, which was initiated in \cite{KV15}. We refer to that paper for a more extensive discussion of the connection of \eqref{eq:S1E2} to particle models, as well as other equations in mathematical physics such as the Nordheim equation.

The structure of the paper is as follows. In Section \ref{sec:statementofresult} we introduce our notation, and we give the statements of the main results. Section \ref{sec:existence} contains the proof of existence of self-similar profiles, while in Section \ref{sec:regularity} regularity is proven. In Section \ref{sec:fattails} we then demonstrate unique power law behaviour of the self-similar profiles in the case of infinite energy. Lastly, in the case of finite energy we prove a pointwise exponential upper bound and an exponential lower bound in averaged sense in Section \ref{sec:exponentialtail}.

\section{Notation and results}\label{sec:statementofresult}
We start with some definitions and notations that we use throughout the paper.
\begin{df}\label{df:measurespace}
We write $\mathcal{M}([0,\infty))$, $\mathcal{M}_+([0,\infty))$, and $\mathcal{M}_+([0,\infty])$ for the spaces of signed, nonnegative, and finite nonnegative Radon measures respectively.
\end{df}
\begin{rk}
Note that the notation for measure spaces as introduced in Definition \ref{df:measurespace} differs from the one in \cite{KV15}. In that paper $\mathcal{M}_+([0,\infty))$ was used to denote the space of finite nonnegative Radon measures $\mu$ on $[0,\infty]$ for which $\mu(\{\infty\})=0$.
\end{rk}
\begin{rk}
For an integral with respect to a measure $\mu$ we will always use the notation $\mu(x)\dd x$, even if $\mu$ is not absolutely continuous with respect to Lebesgue measure.
\end{rk}
\begin{df}
Given $I\subset[-\infty,\infty]$, we write $C(I)=C^0(I)$ for the set of functions that are continuous on $I$. Given further $k\in\mathbb{N}$, we write $C^k(I)$ for the subset of these functions in $C(I)$ for which the derivatives of order up to $k$ exist and are in $C(I)$, and $C_c^k(I)$ [$C_c(I)$] for the set of functions in $C^k(I)$ [$C(I)$] supported in a compact $K\subset I$. Given finally $k\in\mathbb{N}_0=\{0,1,\ldots\}$ and $\alpha\in(0,1)$, we write $C^{k,\alpha}(I)$ for the set of functions in $C^k(I)$ for which the $k$-th derivative is $\alpha$-H\"older continuous on any compact $K\subset I$.
\end{df}
\begin{rk}
Given $\varphi\in C(I)$, we write $\|\varphi\|_\infty=\|\varphi\|_{C(I)}$.
\end{rk}
\begin{rk}
Note that if $f\in C^k([0,\infty])$, then not only are the functions $f^{(\ell)}$, with $\ell=0,1,\ldots,k$, bounded on the interval $[0,\infty]$, but also $\lim_{x\rightarrow\infty}f(x)$ exists, and $\lim_{x\rightarrow\infty}f^{(\ell)}(x)=0$ for $\ell=1,\ldots,k$.
\end{rk}
\begin{df}
Given $k\in\mathbb{N}_0$, we denote by $\mathcal{B}_k$ the subspace of those functions $\vartheta\in C^k([0,\infty])$ for which $\|(1+x)\vartheta(x)\|_{C^k([0,\infty])}=:\|\vartheta\|_{\mathcal{B}_k}<\infty$.
\end{df}
\begin{rk}
Since $\mathcal{B}_0$ is a separable Banach space, the unit ball in the dual space $\mathcal{B}_0'$ endowed with the weak-$*$ topology is metrizable (cf.~\cite[Thm.~3.28]{B11}). Consequently, the properties of the weak-$*$ topology restricted to the unit ball in $\mathcal{B}_0'$ can be characterized by means convergence of sequences. We recall that a sequence $\{\mu_n\}$ in $\mathcal{B}_0'$ converges to $\mu\in\mathcal{B}_0'$ with respect to the weak-$*$ topology (denoted $\mu_n\wsc\mu$ in $\mathcal{B}_0'$) if and only if
\begin{equation}\nonumber
\int_{[0,\infty)}\vartheta(x)\mu_n(x)\dd x\rightarrow\int_{[0,\infty)}\vartheta(x)\mu(x)\dd x\text{ for all }\vartheta\in\mathcal{B}_0.
\end{equation}
\end{rk}
\begin{rk}
We use the notations $a\vee b=\max\{a,b\}$ and $a\wedge b=\min\{a,b\}$.
\end{rk}
A robust characterization of the power law behaviour of measures $\mu$ near infinity will be achieved by means of the functionals
\begin{equation}\label{eq:functional}
R^{\rho-2}\int_{[0,\infty)}\left(1\wedge\tfrac{R}{x}\right)|\mu(x)|\dd x.
\end{equation}
More precisely, we will make extensive use of the following normed spaces.
\begin{df}
Given $\rho\in(1,2]$, we define $\mathcal{X}_\rho$ to be the subset of those nonnegative Radon measures $\mu\in\mathcal{M}_+([0,\infty))$ for which
\begin{equation}\label{eq:rhonorm}
\sup_{R>0}\left\{R^{\rho-2}\int_{[0,\infty)}\left(1\wedge\tfrac{R}{x}\right)|\mu(x)|\dd x\right\}=:\|\mu\|_\rho<\infty.
\end{equation}
\end{df}
\begin{rk}
Even though the space $\mathcal{X}_\rho$ only contains nonnegative Radon measures, the norm $\|\cdot\|_\rho$ is defined for arbitrary signed Radon measures by \eqref{eq:rhonorm}.

Also, since $\|\mu\|_2=\int_{[0,\infty)}\mu(x)\dd x$ we can identify any $\mu\in\mathcal{X}_2$ with a unique element in $\mathcal{M}_+([0,\infty])\cap\{\mu(\{\infty\})=0\}$, and we will henceforth use the abbreviated notation $\mathcal{X}_2=\mathcal{M}_+([0,\infty])\cap\{\mu(\{\infty\})=0\}$.

Note lastly that if $\rho\in(1,2)$, then for all $\mu\in\mathcal{X}_\rho$ there holds $\mu(\{0\})=0$, since $0\leq\mu(\{0\})\leq R^{\rho-2}\int_{[0,\infty]}(1\wedge\frac{R}{x})\mu(x)\dd x\times R^{2-\rho}\leq\|\mu\|_\rho R^{2-\rho}$ for all $R>0$.
\end{rk}
\begin{lm}\label{lm:inclusionandclosedness}
Given $\rho\in(1,2)$, there holds $\mathcal{X}_\rho\subset\mathcal{B}_0'$, and $\{\|\mu\|_\rho\leq 1\}\cap\mathcal{X}_\rho$ is weakly-$*$ closed in $\mathcal{B}_0'$.
\end{lm}
\begin{proof}
Since for $\vartheta\in\mathcal{B}_0$ with $\|\vartheta\|_{\mathcal{B}_0}=1$ there holds $|\vartheta(x)|\leq\frac1{1+x}\leq1\wedge\frac1x$ for $x>0$, we find for $\mu\in\mathcal{X}_\rho$, which are nonnegative, that
\begin{equation}\label{eq:inlusionestimate}
\|\mu\|_{\mathcal{B}_0'}\overset{\rm def}{=}\sup_{\|\vartheta\|_{\mathcal{B}_0}=1}\int_{[0,\infty)}\vartheta(x)\mu(x)\dd x
\leq\int_{[0,\infty)}\left(1\wedge\tfrac1x\right)\mu(x)\dd x\leq\|\mu\|_\rho,
\end{equation}
which proves the inclusion. Given further a sequence $\{\mu_n\}$ in $\{\|\mu\|_\rho\leq 1\}\cap\mathcal{X}_\rho$ such that $\mu_n\wsc\mu$ in $\mathcal{B}_0'$, then clearly $\mu\geq0$. Furthermore, for all $R>0$ there holds by definition of weak-$*$ convergence that
\begin{equation}\nonumber
\zeta_n(R)=R^{2-\rho}\int_{[0,\infty)}\left(1\wedge\tfrac{R}{x}\right)\mu_n(x)\dd x\rightarrow R^{2-\rho}\int_{[0,\infty)}\left(1\wedge\tfrac{R}{x}\right)\mu(x)\dd x,
\end{equation}
and since for all $R>0$ the sequence $\{\zeta_n(R)\}$ is bounded by one, we conclude that $\|\mu\|_\rho\leq1$.
\end{proof}
\begin{df}
By the weak-$*$ topology on $\mathcal{X}_\rho$ we mean the restriction of the weak-$*$ topology of $\mathcal{B}_0'$ to $\mathcal{X}_\rho$.
\end{df}
\begin{lm}
Given $\rho\in(1,2)$, the unit ball in $\mathcal{X}_\rho$ is weakly-$*$ compact.
\end{lm}
\begin{proof}
Using \eqref{eq:inlusionestimate}, there holds $\{\|\mu\|_\rho\leq1\}\cap\mathcal{X}_\rho\subset\{\|\mu\|_{\mathcal{B}_0'}\leq1\}$, so by weak-$*$ closedness it suffices to check that $\{\|\mu\|_{\mathcal{B}_0'}\leq1\}$ is weakly-$*$ compact, which follows by Banach-Alaoglu (cf.~\cite[Thm.~3.16]{B11}).
\end{proof}
\begin{rk}\label{rk:isomorphism}
Notice that for any function $\varphi\in C([0,\infty])$ there exists a unique $\vartheta\in\mathcal{B}_0$ such that $\varphi(x)=(1+x)\vartheta(x)$, and vice versa. Therefore
\begin{multline}\nonumber
\|\mu\|_{\mathcal{B}_0'}=\sup_{\|\vartheta\|_{\mathcal{B}_0}=1}\int_{[0,\infty)}\vartheta(x)\mu(x)\dd x\\
=\sup_{\|\varphi\|_{C([0,\infty])}=1}\int_{[0,\infty)}\varphi(x)\tfrac{\mu(x)}{1+x}\dd x=\left\|\tfrac{\mu(x)}{1+x}\right\|_{(C([0,\infty]))'},
\end{multline}
and $\mathcal{B}_0'$ and $\mathcal{M}([0,\infty];\frac{\dd x}{1+x})$ are isomorphic.
\end{rk}
\begin{df}\label{df:defofY}
For $\rho\in(1,2)$, we define for any $R_0>0$ the subset $\mathcal{Y}_\rho=\mathcal{Y}_\rho(R_0)$ to contain those elements $\mu\in\mathcal{X}_\rho$ that satisfy both $\|\mu\|_\rho\leq1$ and
\begin{equation}\label{eq:lowerboundinY}
\int_{[0,\infty)}\left(1\wedge\tfrac{R}{x}\right)\mu(x)\dd x\geq R^{2-\rho}\lambda_\rho\big(\tfrac{R}{R_0}\big)\text{ for all }R>0,
\end{equation}
with $\lambda_\rho(x)=(1-|x|^{-(2-\rho)/2})_+$.
\end{df}
\begin{rk}
For any $R_0>0$, the set $\mathcal{Y}_\rho(R_0)$ is a nonempty ($(2-\rho)(\rho-1)x^{1-\rho}\linebreak[0]\dd x\in\mathcal{Y}_\rho$), convex and weakly-$*$ compact subset of the unit sphere $\{\|\mu\|_\rho=1\}$.
\end{rk}
We now state the notion of weak solution to \eqref{eq:strongform}, which is analogous to the one that was introduced in \cite{KV15}.
\begin{df}\label{df:notionofsolution}
A function $G\in C([0,\infty):\mathcal{X}_2)$ that for all $t\in[0,\infty)$ and all $\varphi\in C^1([0,\infty):C_c^1([0,\infty)))$ satisfies
\begin{multline}\label{eq:notionofsolution}
\int_{[0,\infty)}\varphi(t,x)G(t,x)\dd x-\int_{[0,\infty)}\varphi(0,x)G(0,x)\dd x\\
=\int_0^t\left[\int_{[0,\infty)}\varphi_s(x)G(x)\dd x+\frac12\iint_{[0,\infty)^2}\frac{G(x)G(y)}{\sqrt{xy}}\mathcal{D}_2[\varphi](x,y)\dd x\dd y\right]\dd s,
\end{multline}
where $\mathcal{D}_2$ for $\varphi\in C([0,\infty))$ is defined by
\begin{equation}\nonumber
\mathcal{D}_2[\varphi](x,y)=\varphi(x+y)+\varphi(|x-y|)-2\varphi(x\vee y),
\end{equation}
will be called a {\em weak solution} to \eqref{eq:strongform}.
\end{df}
\begin{rk}
The use of the space $\mathcal{X}_2=\mathcal{M}_+([0,\infty])\cap\{\mu(\{\infty\})=0\}$ might seem artificial. We only impose the restriction to $\{\mu(\{\infty\})=0\}$ to avoid trivial nonuniqueness due to the fact that \eqref{eq:notionofsolution} does not give any information about $G(\cdot,\{\infty\})$, which could be an arbitrary function since we are using test functions that are compactly supported in $[0,\infty)$.
\end{rk}
\begin{rk}
We frequently use the following notation for the second difference:
\begin{equation}\nonumber
\Delta_y^2f(x)=f(x+y)+f(x-y)-2f(x)
\end{equation}
Also, for notational convenience we introduce
\begin{multline}\label{eq:defcalDstar}
\mathcal{D}_2^*[\vartheta](x,y)=\mathcal{D}_2[\varphi](x,y)\text{ with }\varphi(z)=z\vartheta(z)\\
=(x+y)\vartheta(x+y)+|x-y|\vartheta(|x-y|)-2(x\vee y)\vartheta(x\vee y).
\end{multline}
\end{rk}
\begin{lm}
For $f\in C^2(\mathbb{R})$ and $x\in\mathbb{R}$ there hold
\begin{align}
\label{eq:rewriteofsecdif}\Delta_y^2f(x)&=\int_\mathbb{R}(|y|-|w-x|)_+f''(w)\dd w&\text{for }y&\in\mathbb{R},\\
\label{eq:rewriteofsecdifprime}\partial_y\left[\Delta_y^2f(x)\right]
&=\int_{x-y}^{x+y}f''(w)\dd w&\text{for }y&\geq0.
\end{align}
\end{lm}
\begin{proof}
By the fundamental theorem of calculus we observe that
\begin{multline}\nonumber
\Delta_y^2f(x)=\int_x^{x+|y|}f'(z)\dd z-\int_{x-|y|}^xf'(z)\dd z\\
=\int_x^{x+|y|}\int_x^zf''(w)\dd w\dd z+\int_{x-|y|}^x\int_z^xf''(w)\dd w\dd z.
\end{multline}
Applying Fubini to the right hand side and rearranging terms, we find \eqref{eq:rewriteofsecdif}. The proof of \eqref{eq:rewriteofsecdifprime} is similar.
\end{proof}
\begin{rk}
For any $f,g\in C([0,\infty))$, we write $f(x)\sim g(x)$ as $x\rightarrow\infty$ if there holds $\lim_{x\rightarrow\infty}\frac{f(x)}{g(x)}=1$.
\end{rk}
\subsection{Statement of main results}
In this section we state the main results of this paper. The first result gives a sufficient condition for existence of a self-similar solution.
\begin{pr}\label{pr:existenceofselfsimilarsolutions}
Given $\rho\in(1,2]$, if $\Phi_\rho\in L^1(0,\infty)$ is a nonnegative function that for all $\varphi\in C_c^1([0,\infty))$ satisfies
\begin{multline}\label{eq:selfsimilarprofileL1}
\frac1\rho\int_{(0,\infty)}\left(x\varphi'(x)-(\rho-1)(\varphi(x)-\varphi(0))\right)\Phi_\rho(x)\dd x\\
=\iint_{\{x>y>0\}}\frac{\Phi_\rho(x)\Phi_\rho(y)}{\sqrt{xy}}\Delta_y^2\varphi(x)\dd x\dd y,
\end{multline}
then the function $G\in C([0,\infty):\mathcal{X}_2)$ that is given by
\begin{equation}\label{eq:selfsimilarsolution}\nonumber
G(t,x)\dd x=\left(M-\frac{\|\Phi_\rho\|_{L^1(0,\infty)}}{(t+t_0)^{(\rho-1)/\rho}}\right)\delta_0(x)\dd x+\Phi_\rho\left(\frac{x}{(t+t_0)^{1/\rho}}\right)\frac{\dd x}{t+t_0},
\end{equation}
with $t_0\in(0,\infty)$ and $M\in[0,\infty)$ such that $Mt_0^{(\rho-1)/\rho}\geq\|\Phi_\rho\|_{L^1(0,\infty)}$, is a weak solution to \eqref{eq:strongform} in the sense of Definition \ref{df:notionofsolution}.
\end{pr}
\begin{proof}
Mutatis mutandis identical to the proof of \cite[Prop.~4.1]{KV15}, where the case $\rho=2$ was considered.
\end{proof}
The rest of this paper is devoted to proving of the following results.
\begin{tm}[Existence]\label{tm:existence}
Given $\rho\in(1,2]$, there exists at least one $\Phi_\rho\in\mathcal{X}_2$ that for all $\varphi\in C_c^1([0,\infty))$ satisfies
\begin{multline}\label{eq:selfsimilarprofileM+}
\frac1\rho\int_{[0,\infty)}\left(x\varphi'(x)-(\rho-1)(\varphi(x)-\varphi(0))\right)\Phi_\rho(x)\dd x\\
=\frac12\iint_{[0,\infty)^2}\frac{\Phi_\rho(x)\Phi_\rho(y)}{\sqrt{xy}}\mathcal{D}_2[\varphi](x,y)\dd x\dd y.
\end{multline}
\end{tm}
\begin{pr}[Regularity]\label{pr:regularity}
Given $\rho\in(1,2]$, if $\Phi_\rho\in\mathcal{X}_2$ satisfies \eqref{eq:selfsimilarprofileM+} for all $\varphi\in C_c^1([0,\infty))$, then $\Phi_\rho$ is absolutely continuous with respect to Lebesgue measure, and its Radon-Nykodim derivative is smooth on $(0,\infty)$ and satisfies
\begin{multline}\label{eq:reg_3}
-\tfrac1\rho x\Phi_\rho'(x)-\Phi_\rho(x)=\int_0^{x/2}\frac{\Phi_\rho(y)}{\sqrt{y}}\left[\frac{\Phi_\rho(x+y)}{\sqrt{x+y}}+\frac{\Phi_\rho(x-y)}{\sqrt{x-y}}-2\frac{\Phi_\rho(x)}{\sqrt{x}}\right]\dd y\\
+\int_{x/2}^\infty\frac{\Phi_\rho(y)\Phi_\rho(x+y)}{\sqrt{y(x+y)}}\dd y-2\frac{\Phi_\rho(x)}{\sqrt{x}}\int_{x/2}^x\frac{\Phi_\rho(y)}{\sqrt{y}}\dd y.
\end{multline}
Actually, $\Phi_\rho$ thus satisfies \eqref{eq:selfsimilarprofileL1} for all $\varphi\in C_c^1([0,\infty))$, and furthermore, $\Phi_\rho$ is either strictly positive or identically zero on $(0,\infty)$.
\end{pr}

\begin{pr}\label{pr:rescaling}
Given $\rho\in(1,2]$, if $\Phi_\rho\in\mathcal{X}_2$ satisfies \eqref{eq:selfsimilarprofileM+} for all $\varphi\in C_c^1([0,\infty))$, then $x\Phi_\rho(x)\in\mathcal{X}_\rho$. Furthermore, for any $c>0$ the rescaled measure $\Phi_*(x)\dd x=\Phi_\rho(cx)\dd x$ also satisfies \eqref{eq:selfsimilarprofileM+} for all $\varphi\in C_c^1([0,\infty))$, and there holds $\|x\Phi_*(x)\|_\rho=c^{-\rho}\|x\Phi_\rho(x)\|_\rho$.
\end{pr}
\begin{rk}\label{rk:restriction}
The statements in Theorem \ref{tm:existence} and Proposition \ref{pr:rescaling} have already been proven for the case $\rho=2$ in \cite[Sec.~4]{KV15}. In the proofs in this paper we will thus restrict ourselves to the case $\rho\in(1,2)$.
\end{rk}
\begin{tm}[Power law asymptotics]\label{tm:fattails}
Given $\rho\in(1,2)$, if $\Phi_\rho\in\mathcal{X}_2$ satisfies \eqref{eq:selfsimilarprofileM+} for all $\varphi\in C_c^1([0,\infty))$, and if furthermore $\|x\Phi_\rho(x)\|_\rho=1$, then
\begin{equation}\nonumber\label{eq:rhodecaybehaviour}
\Phi_\rho(r)\sim(2-\rho)(\rho-1)r^{-\rho}\text{ as }r\rightarrow\infty.
\end{equation}
\end{tm}
\begin{tm}[Exponential bounds]\label{tm:exponentialtails}
If $\Phi_2\in\mathcal{X}_2$ satisfies \eqref{eq:selfsimilarprofileM+} with $\rho=2$ for all $\varphi\in C_c^1([0,\infty))$, and if $\Phi_2$ is not identically zero on $(0,\infty)$, then there exist constants $a,c\in(0,1)$ such that
\begin{equation}\nonumber
\Phi_2(r)\leq\frac{e^{-ar}}c\text{ for all }r\geq1\text{, and }\int_{(R,R+1)}\Phi_2(x)\dd x\geq ce^{-\frac{R}{a}}\text{ for all }R\geq0.
\end{equation}
\end{tm}

\section{Existence of self-similar profiles}\label{sec:existence}
The proof of Theorem \ref{tm:existence} for the case $\rho=2$ was already given in \cite{KV15}, and the obtained profiles $\Phi_2$ turned out to have finite energy. Due to this finiteness of the energy the existence result for self-similar solutions to \eqref{eq:strongform} in \cite{KV15} can be seen as the analogue to the existence result for self-similar solutions with finite mass to the coagulation equation obtained in \cite{EMR05} and \cite{FL05}.

For the coagulation equation, self-similar solutions with infinite mass, i.e.~with fat tailed behaviour at infinity, have been obtained in \cite{NV13} for locally bounded kernels, and in \cite{NTV15} for a class of singular kernels, which in particular includes the classical Smoluchowski kernel for Brownian coagulation.

In this paper we construct self-similar solutions with fat tailed behaviour at infinity to \eqref{eq:strongform}, adapting the methods of \cite{NV13}. The main idea in the construction made in that paper, is that for fat tailed solutions the linear terms in the equation for the self-similar profile are dominant for large values of $x$. The effect of the nonlinear collision kernels can be seen as a nonlocal diffusive effect for large particles, which gives a lower order correction. Due to the fact that in coagulation equations the size of the particles is always increasing, the resulting diffusive effect is directed towards larger values. Conversely, in our case the collision kernel can transport particles to both larger or smaller values, and the resulting nonlocal diffusive effect is no longer directed. However, the operator describing this diffusive effect is more symmetric than in the case of coagulation. This has two main consequences. Firstly, the natural test functions required to study the transport of particles are those that are either convex or concave, while in the case of the coagulation equation the natural test functions were the monotone ones. Secondly, due to the symmetry of our collision kernel the singular terms in \eqref{eq:strongform} have a weaker effect, and many of the technicalities that had to be introduced in \cite{NTV15} can be avoided.

On the proof presented in the following, we would like to mention that large parts of our construction also work in the case $\rho=2$. However, technicalities aside, it is not a priori clear that this construction yields a nontrivial solution where not all the mass is concentrated in the origin.\\

We now restrict ourselves to $\rho\in(1,2)$. Introducing as in \cite{KV15} the notations $\Psi_\rho(x)=x\Phi_\rho(x)$ and $\vartheta(x)=\frac1x(\varphi(x)-\varphi(0))$, we can rewrite \eqref{eq:selfsimilarprofileM+} as
\begin{multline}\label{eq:selfsimilarpsiprofile}
\frac1\rho\int_{[0,\infty)}\left(x\vartheta'(x)+(2-\rho)\vartheta(x)\right)\Psi_\rho(x)\dd x\\
=\frac12\iint_{[0,\infty)^2}\frac{\Psi_\rho(x)\Psi_\rho(y)}{(xy)^{3/2}}\mathcal{D}_2^*[\vartheta](x,y)\dd x\dd y,
\end{multline}
where we recall the notation \eqref{eq:defcalDstar}. Now, we would like to prove existence of a solution to \eqref{eq:selfsimilarpsiprofile} by showing that there exists a stationary solution to
\begin{multline}\label{eq:selfsimilarpsiprofiledynamiceqn}
\int_{[0,\infty)}\vartheta(t,x)\Psi_\rho(t,x)\dd x-\int_{[0,\infty)}\vartheta(0,x)\Psi_\rho(0,x)\dd x\\
=\int_0^t\Bigg[\int_{[0,\infty)}\left(\vartheta_s(s,x)-\tfrac1\rho\left(x\vartheta_x(s,x)+(2-\rho)\vartheta(s,x)\right)\right)\Psi_\rho(s,x)\dd x\\
+\frac12\iint_{[0,\infty)^2}\frac{\Psi_\rho(s,x)\Psi_\rho(s,y)}{(xy)^{3/2}}\mathcal{D}_2^*[\vartheta(s,\cdot)](x,y)\dd x\dd y\Bigg]\dd s.
\end{multline}
In order to avoid technical difficulties due to the singularity of the kernel, we will consider a regularized version of \eqref{eq:selfsimilarpsiprofiledynamiceqn}. We then prove existence of stationary solutions to that equation by a Schauder type fixed point theorem, and finally show by a compactness result that by removing the regularization we obtain a solution to \eqref{eq:selfsimilarpsiprofile}.
\begin{as}\label{assumption}
Let $\rho\in(1,2)$ and $\varepsilon>0$ be fixed arbitrarily, let $a\in(0,\frac\varepsilon2)$ be arbitrary, and let $\phi\in C_c^\infty((-1,1))$ be a fixed even function such that $\phi\geq0$ and $\|\phi\|_{L^1(\mathbb{R})}=1$. For any $b>0$ we define $\phi_{b}(x)=\frac1{b}\phi(\frac{x}{b})$ for all $x\in\mathbb{R}$.
\end{as}
\begin{pr}\label{pr:wellposednessforregularizedevolution}
Under Assumption \ref{assumption}, there exist $R_0>0$ and a weakly-$*$ continuous semigroup $(S_a(t))_{t\geq0}$ on $\mathcal{Y}_\rho=\mathcal{Y}_\rho(R_0)$ such that if given $\Psi_0\in\mathcal{Y}_\rho$, then $\Psi_a(t,\cdot)= S_a(t)\Psi_0\in\mathcal{Y}_\rho$ satisfies
\begin{multline}\label{eq:evolutionofPsia}
\int_{[0,\infty)}\vartheta(t,x)\Psi_a(t,x)\dd x-\int_{[0,\infty)}\vartheta(0,x)\Psi_0(x)\dd x\\
=\int_0^t\left[\int_{[0,\infty)}\left(\vartheta_s(s,x)-\tfrac1\rho\left(x\vartheta_x(s,x)+(2-\rho)\vartheta(s,x)\right)\right)\Psi_a(s,x)\dd x\right.\\
\left.+\iint_{\{x>y>0\}}\frac{\Psi_a(s,x)(\phi_a\ast\Psi_a(s,\cdot))(y)}{((x+\varepsilon)(y+\varepsilon))^{3/2}}\mathcal{D}_2^*[\vartheta(s,\cdot)](x,y)\dd x\dd y\right]\dd s
\end{multline}
for all $t\geq0$ and all $\vartheta\in C^1([0,\infty):\mathcal{B}_1)$.
\end{pr}

\subsection{Construction of the semigroup}
To prove existence of an evolution semigroup for \eqref{eq:evolutionofPsia}, it is useful to consider a reformulation where the transport term is removed. Introducing the variables
\begin{align}\label{eq:firstchangeofvariables}
H_a(s,X)&=\Psi_a(s,x),&\psi(s,X)&=e^{-s/\rho}\vartheta(s,x),&X&=xe^{s/\rho},
\end{align}
we can write \eqref{eq:evolutionofPsia} as
\begin{multline}\label{eq:approximatemildsolution}
\int_{[0,\infty)}\psi(t,X)H_a(t,X)\dd X-\int_{[0,\infty)}\psi(0,X)H_a(0,X)\dd X\\
=\int_0^t\Bigg[\int_{[0,\infty)}\left(\psi_s(s,X)+\tfrac{\rho-1}{\rho}\psi(s,X)\right)H_a(s,X)\dd X\\
\Bigg.+\iint_{\{X>Y>0\}}\frac{e^{s/\rho}H_a(s,X)(\phi_{ae^{s/\rho}}\ast H_a(s,\cdot))(Y)}{((X+\varepsilon e^{s/\rho})(Y+\varepsilon e^{s/\rho}))^{3/2}}\\\times
\mathcal{D}_2^*[\psi(s,\cdot)](X,Y)\dd X\dd Y\Bigg]\dd s.
\end{multline}
To construct the evolution semigroup for \eqref{eq:evolutionofPsia} we thus construct a solution to \eqref{eq:approximatemildsolution}. To this end we prove existence and uniqueness for suitable mild solutions, which turn out to be weak solutions in the sense of \eqref{eq:approximatemildsolution}.
\begin{lm}\label{lm:localexistenceofapproximatemildsolutions}
Under Assumption \ref{assumption}, then given $H_0\in\mathcal{X}_\rho$, there exist $T>0$, depending on $\varepsilon$ and $\|H_0\|_\rho$, and a unique function $H_a\in C([0,T]:\mathcal{X}_\rho)$ that is a fixed point for the operator $\mathcal{T}_a$, from $C([0,T]:\mathcal{X}_\rho)$ to itself, defined by
\begin{multline}\label{eq:definitionofcontractiveoperator}
\mathcal{T}_a[H](t,X)=H_0(X)e^{-\int_0^tA_a(s)[H(s,\cdot)](X)\dd s}\\+\int_0^te^{-\int_s^tA_a(\sigma)[H(\sigma,\cdot)](X)\dd \sigma}B_a(s)[H(s,\cdot)](X)\dd s,
\end{multline}
where $A_a(s):\mathcal{X}_\rho\rightarrow C([0,\infty])$, for $s\in[0,T]$, is given by
\begin{equation}\label{eq:definitionofA}
A_a(s)[H](X)=\frac{2Xe^{s/\rho}}{(X+\varepsilon e^{s/\rho})^{3/2}}\int_0^X\frac{(\phi_{ae^{s/\rho}}\ast H)(Y)}{(Y+\varepsilon e^{s/\rho})^{3/2}}\dd Y-\frac{\rho-1}{\rho},
\end{equation}
and where $B_a(s):\mathcal{X}_\rho\rightarrow\mathcal{X}_\rho$, again for $s\in[0,T]$, is such that for all $\psi\in\mathcal{B}_0$ we have
\begin{multline}\nonumber
\int_{[0,\infty)}\psi(X)B_a(s)[H](X)\dd X=\iint_{\{X>Y>0\}}\frac{e^{s/\rho} H(X)(\phi_{ae^{s/\rho}}\ast H)(Y)}{((X+\varepsilon e^{s/\rho})(Y+\varepsilon e^{s/\rho}))^{3/2}}\\\times\left[(X+Y)\psi(X+Y)+(X-Y)\psi(X-Y)\right]\dd X\dd Y.
\end{multline}
Moreover, for initial data in $\{\|\mu\|_\rho\leq E_0\}\cap\mathcal{X}_\rho$, the constant $T>0$ depends only on $\varepsilon$ and $E_0$.
\end{lm}
\begin{lm}\label{lm:weaksolution}
The fixed point $H_a\in C([0,T]:\mathcal{X}_\rho)$, obtained in Lemma \ref{lm:localexistenceofapproximatemildsolutions}, satisfies \eqref{eq:approximatemildsolution} for all $t\in[0,T]$ and $\psi\in C^1([0,T]:\mathcal{B}_0)$.
\end{lm}
\begin{proof}[Proof of Lemma \ref{lm:localexistenceofapproximatemildsolutions}]
To check that $\mathcal{T}_a$ is well-defined from $C([0,T]:\mathcal{X}_\rho)$ to itself, it suffices to check that $B_a(s)$ maps $\mathcal{X}_\rho$ into itself. To that end we note that
\begin{multline}\nonumber
\|B_a(s)[H]\|_\rho=\sup_{R>0}R^{\rho-2}\int_{[0,\infty)}\left(1\wedge\tfrac{R}{X}\right)B_a(s)[H](X)\dd X\\
\leq\frac2\varepsilon\int_{(0,\infty)}\left(\sup_{R>0}R^{\rho-2}\int_{(Y,\infty)}\left(1\wedge\tfrac{R}{X}\right)H(X)\dd X\right)\frac{(\phi_{ae^{s/\rho}}\ast H)(Y)}{Y+\varepsilon e^{s/\rho}}\dd Y\\
\leq\frac2\varepsilon\|H\|_\rho\int_{[0,\infty)}\left(\int_{(0,\infty)}\frac{\phi_{ae^{s/\rho}}(Y-Z)}{Y+\varepsilon e^{s/\rho}}\dd Y\right)H(Z)\dd Z,
\end{multline}
so using further that $|Y-Z|\leq ae^{s/\rho}<\frac12\varepsilon e^{s/\rho}$ for all $Y-Z\in{\rm supp}(\phi_{ae^{s/\rho}})$, we have
\begin{multline}\label{eq:estimateonB}
\|B_a(s)[H]\|_\rho\leq\frac2\varepsilon\|H\|_\rho\int_{[0,\infty)}\left(\int_{(0,\infty)}\frac{\phi_{ae^{s/\rho}}(Y-Z)}{Z+(\varepsilon-a)e^{s/\rho}}\dd Y\right)H(Z)\dd Z\\
\leq\frac2\varepsilon\|H\|_\rho\int_{[0,\infty)}\tfrac{1}{\frac12\varepsilon e^{s/\rho}}\left(1\wedge\tfrac{\frac12\varepsilon e^{s/\rho}}{Z}\right)H(Z)\dd Z
\leq\tfrac{2^\rho}{\varepsilon^\rho}e^{-s(\rho-1)/\rho}\|H\|_\rho^2.
\end{multline}
Using this estimate and exploiting the nonnegativity of the first term on the right hand side of \eqref{eq:definitionofA}, we find for any $t\in[0,T]$ that
\begin{multline}\nonumber
\|\mathcal{T}_a[H](t,\cdot)\|_\rho\leq e^{t(\rho-1)/\rho}\|H_0\|_\rho+\int_0^t e^{(t-s)(\rho-1)/\rho}\|B_a(s)[H(s,\cdot)]\|_\rho\dd s\\
\leq e^{t(\rho-1)/\rho}\left(\|H_0\|_\rho+\frac{2^\rho}{\varepsilon^\rho}\int_0^t e^{-2s(\rho-1)/\rho}\dd s\times \sup_{s\in[0,t]}\|H(s,\cdot)\|_\rho^2\right),
\end{multline}
implying that $\mathcal{T}_a$ maps the subset
\begin{equation}\nonumber
\mathcal{S}:=\left\{H\in C([0,T]:\mathcal{X}_\rho)\,\Big|\,\sup_{t\in[0,T]}\|H(t,\cdot)\|_\rho=:\|H\|_{T,\rho}\leq2\|H_0\|_\rho\right\}
\end{equation}
into itself, provided that $T>0$ is small enough. Note that for $\varepsilon>0$ fixed, $T>0$ can be chosen uniformly for $H_0\in\{\|\mu\|_\rho\leq E_0\}\cap\mathcal{X}_\rho$.

To check that the operator is actually strongly contractive on $\mathcal{S}$ for sufficiently small $T>0$, and thereby proving the lemma, we now first observe for $H_1^*,H_2^*\in\mathcal{X}_\rho$ and $\sigma\in[0,T]$ that
\begin{multline}\nonumber
\left\|A_a(\sigma)[H_1^*](\cdot)-A_a(\sigma)[H_2^*](\cdot)\right\|_\infty\\
\leq\sup_{X>0}\frac{2Xe^{\sigma/\rho}}{(X+\varepsilon e^{\sigma/\rho})^{3/2}}\left|\int_0^X\frac{(\phi_{ae^{\sigma/\rho}}\ast(H_1^*-H_2^*))(Y)}{(Y+\varepsilon e^{\sigma/\rho})^{3/2}}\dd Y\right|\\
\leq\frac2{\varepsilon}\int_{[0,\infty)}\left(\int_0^\infty\frac{\phi_{ae^{\sigma/\rho}}(Y-Z)}{Y+\varepsilon e^{\sigma/\rho}}\dd Y\right)|H_1^*-H_2^*|(Z)\dd Z\\
\leq\tfrac{2^\rho}{\varepsilon^\rho}e^{-\sigma(\rho-1)/\rho}\|H_1^*-H_2^*\|_\rho
\end{multline}
hence for $H_1,H_2\in\mathcal{S}$ and $0\leq s\leq t\leq T$ we have
\begin{multline}\label{eq:estimateondifferenceinA}
\left\|e^{-\int_s^tA_a(\sigma)[H_1(\sigma,\cdot)](\cdot)\dd \sigma}-e^{-\int_s^tA_a(\sigma)[H_2(\sigma,\cdot)](\cdot)\dd \sigma}\right\|_\infty\\
\leq e^{(t-s)(\rho-1)/\rho}\int_s^t\left\|A_a(\sigma)[H_1(\sigma,\cdot)](\cdot)-A_a(\sigma)[H_2(\sigma,\cdot)](\cdot)\right\|_\infty\dd \sigma\\
\leq\tfrac{2^\rho}{\varepsilon^\rho}(t-s)e^{(t-s)(\rho-1)/\rho}\|H_1-H_2\|_{T,\rho}.
\end{multline}
Again for $H_1^*,H_2^*\in\mathcal{X}_\rho$ we next note that
\begin{multline}\nonumber
\int_{(0,\infty)}\left(1\wedge\tfrac{R}{X}\right)|B_a(s)[H_1^*]-B_a(s)[H_2^*]|(X)\dd X\\
\leq\iint_{\{X>Y>0\}}\frac{e^{s/\rho}|H_1^*(X)(\phi_{ae^{s/\rho}}\ast H_1^*)(Y)-H_2^*(X)(\phi_{ae^{s/\rho}}\ast H_2^*)(Y)|}{((X+\varepsilon e^{s/\rho})(Y+\varepsilon e^{s/\rho}))^{3/2}}\\\times\left((X+Y)\wedge R+(X-Y)\wedge R\right)\dd X\dd Y,
\end{multline}
so analogous arguments as used to obtain \eqref{eq:estimateonB} give us that
\begin{multline}\label{eq:estimateondifferenceinB}
\|B_a(s)[H_1^*]-B_a(s)[H_2^*]\|_\rho\\
\leq\tfrac{2^\rho}{\varepsilon^\rho}e^{-s(\rho-1)/\rho}\left(\|H_1^*\|_\rho+\|H_2^*\|_\rho\right)\|H_1^*-H_2^*\|_\rho.
\end{multline}
Combining finally \eqref{eq:estimateonB}, \eqref{eq:estimateondifferenceinA} and \eqref{eq:estimateondifferenceinB}, we find for $H_1,H_2\in\mathcal{S}$ the estimate
\begin{equation}\nonumber
\|\mathcal{T}_a[H_1]-\mathcal{T}_a[H_2]\|_{T,\rho}\leq K(T)\|H_1-H_2\|_{T,\rho},
\end{equation}
with
\begin{equation}\nonumber
K(T)=\tfrac{2^\rho}{\varepsilon^2}Te^{T(\rho-1)/\rho}\|H_0\|_\rho\left(1+4\left(\tfrac{2^\rho}{\varepsilon^\rho}T\|H_0\|_\rho+e^{-T(\rho-1)/\rho}\right)\right)\xrightarrow{T\rightarrow0}0,
\end{equation}
and noting again that for $\varepsilon>0$ fixed we can again choose $T>0$ uniformly for $H_0\in\{\|\mu\|_\rho\leq E_0\}\cap\mathcal{X}_\rho$, the proof is complete.
\end{proof}
\begin{proof}[Proof of Lemma \ref{lm:weaksolution}]
By construction there holds $H_a=\mathcal{T}_a[H_a]$, so multiplying this identity by $\varphi\in C^1([0,T]:\mathcal{B}_0)$ and integrating with respect to $X$, we obtain for all $t\in[0,T]$ that
\begin{multline}\label{eq:mildsolution_1}
\int_{[0,\infty)}\varphi(t,X)H_a(t,X)\dd X=\int_{[0,\infty)}\varphi(t,X)H_0(X)e^{-\int_0^tA_a(s)[H_a(s,\cdot)](X)\dd s}\dd X\\+\int_0^t\int_{[0,\infty)}\varphi(t,X)e^{-\int_s^tA_a(\sigma)[H_a(\sigma,\cdot)](X)\dd \sigma}B_a(s)[H_a(s,\cdot)](X)\dd X\dd s.
\end{multline}
If now $\psi\in C^1([0,T]:\mathcal{B}_0)$ is arbitrary, then taking the time derivative of \eqref{eq:mildsolution_1} with $\varphi$ replaced by $\psi$, we get
\begin{multline}\label{eq:mildsolution_2}
\partial_t\left[\int_{[0,\infty)}\psi(t,X)H_a(t,X)\dd X\right]=\int_{[0,\infty)}\psi(t,X)B_a(t)[H_a(t,\cdot)](X)\dd X\\
+\int_{[0,\infty)}\big(\psi_t(t,X)-\psi(t,X)A_a(t)[H_a(t,\cdot)](X)\big)H_a(t,X)\dd X,
\end{multline}
where the last term on the right hand side is obtained by combining two terms, using the identity obtained from \eqref{eq:mildsolution_1} with $\varphi(t,X)=\psi(t,X)A_a(t)[H_a(t,\cdot)](X)$. Integrating \eqref{eq:mildsolution_2}, we then obtain \eqref{eq:approximatemildsolution}.
\end{proof}
We are now able to show local in time existence of solutions to \eqref{eq:evolutionofPsia} by construction, as well as an estimate of the norm $\|\cdot\|_\rho$ for these solutions.
\begin{pr}\label{pr:globalnormbound}
Under Assumption \ref{assumption} and supposing for $\Psi_0\in\mathcal{X}_\rho$ that $T>0$ and $H_a\in C([0,T]:\mathcal{X}_\rho)$ are as obtained in Lemma \ref{lm:localexistenceofapproximatemildsolutions} with $H_0=\Psi_0$, then the function $\Psi_a\in C([0,T]:\mathcal{X}_\rho)$, defined via
\begin{equation}\label{eq:changeofvariables}
\Psi_a(t,x)=H_a(t,X)\quad\text{and}\quad x=Xe^{-t/\rho},
\end{equation}
satisfies \eqref{eq:evolutionofPsia} for all $\vartheta\in C^1([0,T]:\mathcal{B}_1)$ and all $t\in[0,T]$, and for all $t\in[0,T]$ there holds $\|\Psi_a(t,\cdot)\|_\rho\leq\|\Psi_0\|_\rho$.
\end{pr}
\begin{proof}
To check that $\Psi_a$ satisfies \eqref{eq:evolutionofPsia} is an elementary computation [use \eqref{eq:firstchangeofvariables}], so we restrict ourselves to proving the estimate of the norm. We observe that
\begin{equation}\nonumber
\|\Psi_a(t,\cdot)\|_\rho=e^{-t(\rho-1)/\rho}\|H_a(t,\cdot)\|_\rho,
\end{equation}
so it suffices to check that
\begin{equation}\label{eq:normestimate}
\|H_a(t,\cdot)\|_\rho\leq e^{t(\rho-1)/\rho}\|H_0\|_\rho\text{ for all }t\in[0,T].
\end{equation}
By Lemma \ref{lm:weaksolution}, now $H_a$ satisfies \eqref{eq:approximatemildsolution} for all $t\in[0,T]$ and $\psi\in C^1([0,T]:\mathcal{B}_0)$, and we note for any $R>0$ that $\psi(X)=1\wedge\frac{R}{X}$ satisfies $\psi\in\mathcal{B}_0$. Moreover, the mapping $X\mapsto X\psi(X)$ is concave, so $\mathcal{D}_2^*[\psi]\leq0$, and there thus holds
\begin{multline}\nonumber
\int_{[0,\infty)}\left(1\wedge\tfrac{R}{X}\right)H_a(t,X)\dd X
\leq\int_{[0,\infty)}\left(1\wedge\tfrac{R}{X}\right)H_0(X)\dd X\\
+\frac{\rho-1}{\rho}\int_0^t\int_{[0,\infty)}\left(1\wedge\tfrac{R}{X}\right)H_a(s,X)\dd X\dd s.
\end{multline}
By Gronwall's lemma and multiplying by $R^{\rho-2}$ we then get
\begin{equation}\nonumber
R^{\rho-2}\int_{[0,\infty)}\left(1\wedge\tfrac{R}{X}\right)H_a(t,X)\dd X
\leq e^{t(\rho-1)/\rho}R^{\rho-2}\int_{[0,\infty)}\left(1\wedge\tfrac{R}{X}\right)H_0(X)\dd X,
\end{equation}
and \eqref{eq:normestimate} follows by taking the supremum over all $R>0$.
\end{proof}
We are now able to construct a family of operators on the unit ball of $\mathcal{X}_\rho$, which will turn out to be the semigroup required in Proposition \ref{pr:wellposednessforregularizedevolution}.
\begin{df}\label{df:definitionofsemigroup}
Under Assumption \ref{assumption}, we define the family $(S_a(t))_{t\geq0}$ of operators from the unit ball in $\mathcal{X}_\rho$ into itself as follows. Let $T>0$ be as obtained in Lemma \ref{lm:localexistenceofapproximatemildsolutions} such that for all $\Psi_0\in\{\|\mu\|_\rho\leq1\}\cap\mathcal{X}_\rho$ there exists a unique function $H_a\in C([0,T]:\mathcal{X}_\rho)$ that is a fixed point for the operator $\mathcal{T}_a$, given by \eqref{eq:definitionofcontractiveoperator} with $H_0=\Psi_0$. For $t\geq0$ we then set $S_a(t)\Psi_0=\Psi_a(t,\cdot)$ for $t\in[0,T]$, where $\Psi_a\in C([0,T]:\mathcal{X}_\rho)$ is defined via \eqref{eq:changeofvariables}, and then
\begin{equation}\label{eq:defofsemigroup}
S_a(t)\Psi_0= S_a(t-nT)\left(S_a(T)\right)^n\Psi_0\text{ for $t\in(nT,(n+1)T]$, $n\in\mathbb{N}$,}
\end{equation}
which is possible since $S(T)\Psi_*$ is in the unit ball for all $\Psi_*\in\{\|\mu\|_\rho\leq1\}\cap\mathcal{X}_\rho$ (cf.~Proposition \ref{pr:globalnormbound}).
\end{df}
\begin{pr}
Under Assumption \ref{assumption}, the family of operators $(S_a(t))_{t\geq0}$, as defined in Definition \ref{df:definitionofsemigroup}, has the semigroup property, i.e.
\begin{equation}\label{eq:semigroupproperty}
S_a(t_1+t_2)=S_a(t_1)S_a(t_2)\text{ for all }t_1,t_2\geq0.
\end{equation}
Moreover, given $\Psi_0\in\{\|\mu\|_\rho\leq1\}\cap\mathcal{X}_\rho$, then the function defined as $\Psi_a(t,x)=S_a(t)\Psi_0(x)$ satisfies \eqref{eq:evolutionofPsia} for all $t\geq0$ and all $\vartheta\in C^1([0,\infty):\mathcal{B}_1)$.
\end{pr}
\begin{proof}
For any $\Psi_0\in\{\|\mu\|_\rho\leq1\}\cap\mathcal{X}_\rho$, let $H_a\in C([0,T]:\mathcal{X}_\rho)$ be the unique fixed point to $\mathcal{T}_a$ with $H_0=\Psi_0$. Using then \eqref{eq:changeofvariables}, we find by careful computation that
\begin{multline}\label{eq:Arewrite}
A_a(s)[H_a(s,\cdot)](X)
=\frac{2Xe^{s/\rho}}{(X+\varepsilon e^{s/\rho})^{3/2}}\int_0^X\frac{(\phi_{ae^{s/\rho}}\ast H_a(s,\cdot))(Y)}{(Y+\varepsilon e^{s/\rho})^{3/2}}\dd Y-\frac{\rho-1}{\rho}\\
=\frac{2Xe^{-s/\rho}}{(Xe^{-s/\rho}+\varepsilon)^{3/2}}\int_0^{Xe^{-s/\rho}}\frac{(\phi_a\ast \Psi_a(s,\cdot))(y)}{(y+\varepsilon)^{3/2}}\dd y-\frac{\rho-1}{\rho}\\
=A_a(0)[\Psi_a(s,\cdot)](Xe^{-s/\rho}),
\end{multline}
and similarly we can check that
\begin{equation}\label{eq:Brewrite}
B_a(s)[H_a(s,\cdot)](X)=B_a(0)[\Psi_a(s,\cdot)](Xe^{-s/\rho}).
\end{equation}
Using now \eqref{eq:Arewrite} and \eqref{eq:Brewrite}, it follows by the definition of $H_a$ as the fixed point of $\mathcal{T}_a$ that for $t\in[0,T]$ and $X\geq0$ there holds
\begin{multline}\nonumber
H_a(t,X)=\Psi_0(X)e^{-\int_0^tA_a(0)[\Psi_a(s,\cdot)](Xe^{-s/\rho})\dd s}\\+\int_0^te^{-\int_s^tA_a(0)[\Psi_a(\sigma,\cdot)](Xe^{-\sigma/\rho})\dd \sigma}B_a(0)[\Psi_a(s,\cdot)](Xe^{-s/\rho})\dd s,
\end{multline}
hence by again \eqref{eq:changeofvariables} for $t\in[0,T]$ and $x\geq0$ there holds
\begin{multline}\nonumber
\Psi_a(t,x)=\Psi_0(xe^{t/\rho})e^{-\int_0^tA_a(0)[\Psi_a(s,\cdot)](xe^{(t-s)/\rho})\dd s}\\+\int_0^te^{-\int_s^tA_a(0)[\Psi_a(\sigma,\cdot)](xe^{(t-\sigma)/\rho})\dd \sigma}B_a(0)[\Psi_a(s,\cdot)](xe^{(t-s)/\rho})\dd s.
\end{multline}
For any $t_1,t_2\geq0$ with $t_1+t_2\leq T$ we then use the following decomposition
\begin{multline}\nonumber
\Psi_0(\cdot)e^{\int_0^{t_1+t_2}[\cdots]\dd s}+\int_0^{t_1+t_2}e^{\int_s^{t_1+t_2}[\cdots]\dd \sigma}[\cdots]\dd s\\
=\left(\Psi_0(\cdot)e^{\int_0^{t_2}[\cdots]\dd s}+\int_0^{t_2}e^{\int_s^{t_2}[\cdots]\dd \sigma}[\cdots]\dd s\right)e^{\int_{t_2}^{t_1+t_2}[\cdots]\dd s}\\
+\int_{t_2}^{t_1+t_2}e^{\int_s^{t_1+t_2}[\cdots]\dd \sigma}[\cdots]\dd s
\end{multline}
and after performing the changes of variables $s\rightarrow t_2+s$ and $\sigma\rightarrow t_2+\sigma$ in the integrals on the right hand side we obtain
\begin{multline}\nonumber
\Psi_a(t_1+t_2,x)=\Psi_a(t_2,xe^{t_1/\rho})e^{-\int_0^{t_1}A_a(0)[\Psi_a(t_2+s,\cdot)](xe^{(t_1-s)/\rho})\dd s}\\
+\int_0^{t_1}e^{-\int_{s}^{t_1}A_a(0)[\Psi_a(t_2+\sigma,\cdot)](xe^{(t_1-\sigma)/\rho})\dd \sigma}B_a(0)[\Psi_a(t_2+s,\cdot)](xe^{(t_1-s)/\rho})\dd s.
\end{multline}
We now see that $H_*(s,xe^{s/\rho})=\Psi_*(s,x):=\Psi_a(t_2+s,x)$ is a fixed point for the operator $\mathcal{T}_a$ with $H_0=\Psi_a(t_2,\cdot)$, and by the short time uniqueness of fixed points, obtained in Lemma \ref{lm:localexistenceofapproximatemildsolutions}, we thus find that
\begin{equation}\nonumber
S_a(t_1+t_2)\Psi_0=\Psi_a(t_1+t_2,\cdot)= S_a(t_1)\Psi_a(t_2,\cdot)= S_a(t_1)S_a(t_2)\Psi_0,
\end{equation}
which proves the semigroup property for $t_1,t_2\geq0$ with $t_1+t_2\leq T$.

Next we use the local semigroup property as derived above to observe for $t_1,t_2\in[0,T]$ with $t_1+t_2>T$ that
\begin{equation}\nonumber
S_a(t_1+t_2-T)S_a(T)=S_a(t_1+t_2-T)S_a(T-t_2)S_a(t_2)=S_a(t_1)S_a(t_2),
\end{equation}
so since the left hand side equals $S_a(t_1+t_2)$ by definition [cf.~\eqref{eq:defofsemigroup}], we have
\begin{equation}\label{eq:localsemigroupproperty}
S_a(t_1+t_2)=S_a(t_1)S_a(t_2)=S_a(t_2)S_a(t_1)\text{ for all }t_1,t_2\in[0,T].
\end{equation}
Using lastly \eqref{eq:defofsemigroup} and \eqref{eq:localsemigroupproperty} for arbitrary $t_1,t_2\geq0$, and writing $n_i$ for the integer part of $\frac{t_i}{T}$, we find that
\begin{multline}\label{eq:number}
S_a(t_1)S_a(t_2)
=S_a(t_1-n_1T)(S_a(T))^{n_1}S_a(t_2-n_2T)(S_a(T))^{n_2}\\
=S_a(t_1-n_1T)S_a(t_2-n_2T)(S_a(T))^{n_1+n_2}\\
=S_a(t_1+t_2-(n_1+n_2)T)(S_a(T))^{n_1+n_2}.
\end{multline}
If $t_1+t_2<(n_1+n_2+1)T$, then the right hand side of \eqref{eq:number} equals $S_a(t_1+t_2)$ by definition. On the other hand, if $t_1+t_2\geq(n_1+n_2+1)T$, then by again \eqref{eq:localsemigroupproperty} we have
\begin{multline}\nonumber
S_a(t_1+t_2-(n_1+n_2)T)(S_a(T))^{n_1+n_2}\\
=S_a(t_1+t_2-(n_1+n_2+1)T)(S_a(T))^{n_1+n_2+1},
\end{multline}
and here the right hand side equals $S_a(t_1+t_2)$ by definition. We therefore conclude that \eqref{eq:semigroupproperty} holds.
\end{proof}

\subsection{Two useful lemmas}
In this subsection we give two lemmas that will be useful for obtaining the lower bound in our proof of existence of a set $\mathcal{Y}_\rho=\mathcal{Y}_\rho(R_0)$ that is invariant under the previously defined evolution. These results will also be used in the final section of this paper.
\begin{lm}\label{lm:fractionalheatequation}
For any $\alpha\in(0,2)$ the fundamental solution $u^\alpha$ to the integro-differential equation
\begin{equation}\label{eq:ide}
u_t(t,x)=\int_{\mathbb{R}_+}y^{-\alpha-1}\Delta_y^2[u(t,\cdot)](x)\dd y,
\end{equation}
i.e.~the solution to \eqref{eq:ide} with initial datum $u(0,\cdot)=\delta_0$, is given by $u^\alpha(t,x)=t^{-1/\alpha}v_\alpha(xt^{-1/\alpha})$, where $v_\alpha\in C^\infty(\mathbb{R})$ is the probability density function that has characteristic function $\exp(-c_\alpha|k|^\alpha)$, with $c_\alpha=-2\Gamma(-\alpha)\cos(\frac{\alpha\pi}{2})$ if $\alpha\neq1$ and $c_1=\pi$. In particular, $v_\alpha$ is positive, symmetric, nonincreasing on $\mathbb{R}_+$, and it satisfies $\lim_{|z|\rightarrow\infty}|z|^{\alpha+1}v_\alpha(z)=1$.
\end{lm}
\begin{proof}
Taking the Fourier transform of \eqref{eq:ide} gives us
\begin{equation}\nonumber
\hat{u}_t(t,k)=-c_\alpha|k|^\alpha\hat{u}(t,k),
\end{equation}
hence $u^\alpha$ is the inverse Fourier transform of $\exp(-c_\alpha|k|^\alpha t)$:
\begin{multline}\nonumber
u^\alpha(t,x)=\frac1{2\pi}\int_\mathbb{R}e^{ikx}e^{-c_\alpha|k|^\alpha t}\dd k=\frac1{t^{1/\alpha}}v_\alpha\left(\frac{x}{t^{1/\alpha}}\right)\\
\text{with }v_\alpha(z)=\frac1{2\pi}\int_\mathbb{R}e^{ikz}e^{-c_\alpha|k|^\alpha}\dd k
\end{multline}
Smoothness and symmetry of $v_\alpha$ are immediate, while for the remaining properties of $v_\alpha$ we note that $\exp(-c_\alpha|k|^\alpha)$ is the characteristic function of a symmetric stable probability distribution (cf. \cite[Thm.~5.7.3]{L70}). Now, \cite[Thm.~5.10.1]{L70} states that all stable distributions are unimodal, so since by symmetry the maximum of $v_\alpha$ is located at zero we have that $V_\alpha(x)=\int_{-\infty}^xv_\alpha(z)\dd z$ is concave for $x\geq0$. Therefore $v_\alpha'\leq0$ on $\mathbb{R}_+$ and it is shown that $v_\alpha$ is nonincreasing on $\mathbb{R}_+$. The asymptotics of $v_\alpha$ follow by a standard contour deformation argument, and strict positivity follows from combining the decay behaviour of $v_\alpha$ with the monotonicity result.
\end{proof}
Since we will frequently use solutions to \eqref{eq:ide} with odd initial data, we give the following lemma.
\begin{lm}\label{lm:oddfunctionsunderide}
For $\alpha\in(0,2)$, let $u^\alpha$ be the fundamental solution to \eqref{eq:ide}, let $u_0\in C(\mathbb{R})\cap L^1(\mathbb{R};|x|^{-\alpha-1}\dd x)$ be odd, and for $t>0$ let $u(t,\cdot):=[u_0\ast u^\alpha(t,\cdot)](\cdot)$. Then the following hold.
\begin{itemize}
\item{For all $t>0$, $u(t,\cdot)$ is odd and smooth.}
\item{{\em Maximum principle.}~If $u_0\geq0\,[\,\leq0\,]$ on $\mathbb{R}_+$, then $u(t,\cdot)\geq0\,[\,\leq0\,]$ on $\mathbb{R}_+$ for all $t>0$.}
\item{If $u_0$ is concave [convex] on $\mathbb{R}_+$, then $u(t,\cdot)$ is concave [convex] on $\mathbb{R}_+$ for all $t>0$, and in particular
\begin{equation}\label{eq:signofseconddifference}
\Delta_y^2[u(t,\cdot)](x)\leq0\,[\,\geq0\,]\text{ for all }x\geq0\text{, }y\in\mathbb{R}\text{ and }t\geq0.
\end{equation}
}
\end{itemize}
\end{lm}
\begin{proof}
For all $t>0$, $u_0(t,\cdot)$ is odd since it is the convolution of an odd and an even function, while smoothness follows from the fact that $u^\alpha(t,\cdot)$ is smooth for all $t>0$. Suppose now that $u_0\geq0\,[\,\leq0\,]$ on $\mathbb{R}_+$. For $x\geq0$ we then find, by the facts that $u_0$ is odd and that $u^\alpha(t,\cdot)$ is even for all $t>0$, that we can write
\begin{equation}\nonumber
u(t,x)=\int_{\mathbb{R}_+} u_0(y)\left(u^\alpha(t,x-y)-u^\alpha(t,x+y)\right)\dd y,
\end{equation}
and it follows that $u(t,\cdot)\geq0\,[\,\leq0\,]$ on $\mathbb{R}_+$ for all $t>0$, since $u^\alpha(t,\cdot)$ is even and monotonically decreasing on $\mathbb{R}_+$ ($u^\alpha(t,x-z)-u^\alpha(t,x+z)\geq0$ for $x,z\geq0$). We next restrict ourselves to the case where $u_0$ is concave on $\mathbb{R}_+$, since the other case is similar. Then, for all $y\in\mathbb{R}$ there holds $\Delta_y^2u_0\leq0$ on $\mathbb{R}_+$. For $|y|\leq x$ this follows immediately from the definition of concavity, while for $|y|>x>0$ we note, using that $u_0$ is odd, that
\begin{equation}\nonumber
\Delta_y^2u_0(x)=\Delta_x^2u_0(|y|)+2\big(u_0(|y|)-u_0(|y|-x)-u_0(x)\big)\leq0.
\end{equation}
Here the first term on the right hand side of the equality is nonpositive by the previous argument, and the remaining terms are nonpositive by
\begin{multline}\nonumber
u_0(|y|-x)+u_0(x)=u_0\left(\tfrac{x}{|y|}\times0+\tfrac{|y|-x}{|y|}|y|\right)+u_0\left(\tfrac{|y|-x}{|y|}\times0+\tfrac{x}{|y|}|y|\right)\\
\geq\left(\tfrac{x}{|y|}+\tfrac{|y|-x}{|y|}\right)u_0(0)+\left(\tfrac{|y|-x}{|y|}+\tfrac{x}{|y|}\right)u_0(|y|)=u_0(|y|),
\end{multline}
where we have used that $u_0(0)=0$ since $u_0$ is odd. Next, since the second difference operator is linear, it commutes with the integral operator on the right hand side of \eqref{eq:ide}, and since further the second difference of an odd function is odd, we find by the maximum principle proven above that \eqref{eq:signofseconddifference} holds. Additionally we then find that $u(t,\cdot)$ is concave on $\mathbb{R}_+$ since for all $x\geq0$ we have $u_{xx}(t,x)=\lim_{y\rightarrow0}\frac1{y^2}\Delta_y^2[u(t,\cdot)](x)\leq0$, where we recall that we may take two derivatives by smoothness of $u^\alpha(t,\cdot)$.
\end{proof}

\subsection{Invariance of $\mathcal{Y}_\rho(R_0)$}
Our goal in this subsection is to show that, for some suitable $R_0>0$, the semigroup introduced in Definition \ref{df:definitionofsemigroup} maps the set $\mathcal{Y}_\rho(R_0)$ (cf.~Definition \ref{df:defofY}) into itself. The proof of invariance of the lower bound [cf.~\eqref{eq:lowerboundinY}] shows strong similarities with the approach in \cite{NV13}. As mentioned before, the main difference here compared with the approach in that paper is that because of the form of the nonlocal diffusion operator in our collision kernel [cf.~\eqref{eq:notionofsolution}] it is convenient to use test functions that are concave, while in \cite{NV13} it was natural to use monotone test functions. A consequence of this is that in order to measure the size of $\Psi_a$ it is now natural to use the functionals given by \eqref{eq:functional}, as opposed to the functionals $\int_0^Rh(x)\dd x$ which were used in \cite{NV13}. 

We first derive the following estimate.
\begin{lm}\label{lm:lowerboundfirstlemma}
Under Assumption \ref{assumption}, let $(S_a(t))_{t\geq0}$ be the semigroup on the unit ball of $\mathcal{X}_\rho$ as defined in Definition \ref{df:definitionofsemigroup}. Let further $\vartheta\in C(\mathbb{R})$ be such that the mapping $z\mapsto z\vartheta(z)$ is odd, bounded, and concave on $\mathbb{R}_+$. Then for all $t\geq0$ and all $\Psi_0\in\{\|\mu\|_\rho\leq1\}\cap\mathcal{X}_\rho$ it holds that
\begin{multline}\label{eq:conservedlowerbound_1}
\int_{[0,\infty)}\vartheta(x)S_a(t)\Psi_0(x)\dd x\\
\geq e^{t(\rho-1)/\rho}\int_{[0,\infty)}\left(\int_\mathbb{R}\frac{y\vartheta(y)}{(\rho t)^{1/\rho}}v_\rho\left(\frac{xe^{-t/\rho}-y}{(\rho t)^{1/\rho}}\right)\dd y\right)\tfrac1x\Psi_0(x)\dd x,
\end{multline}
where $v_\rho$ is the self-similar profile associated to the fundamental solution of \eqref{eq:ide} with $\alpha=\rho$, which was obtained in Lemma \ref{lm:fractionalheatequation}.
\end{lm}
\begin{proof}
Throughout this proof we let $t\geq0$ and $\Psi_0\in\{\|\mu\|_\rho\leq1\}\cap\mathcal{X}_\rho$ be fixed, and for all $s\in[0,t]$ we write $\Psi_a(s,x)=S_a(s)\Psi_0(x)$.

For $\vartheta\in C(\mathbb{R})$ fixed as in the statement of the lemma, we define
\begin{multline}\label{eq:defofu}
u(s,x)=e^{(t-s)(\rho-1)/\rho}[\varphi\ast u^\rho(\rho(t-s),\cdot)](xe^{(s-t)/\rho})\times\tfrac1x,\\
\text{ for $s\in[0,t]$ and $x\in\mathbb{R}$, and with $\varphi(y):=y\vartheta(y)$,}
\end{multline}
where $u^\rho$ is the fundamental solution of \eqref{eq:ide} with $\alpha=\rho$ as obtained in Lemma \ref{lm:fractionalheatequation}. We then note that $u(t,\cdot)=\vartheta$, and that $u(0,x)$ is equal to the integral with respect to $y$ in the right hand side of \eqref{eq:conservedlowerbound_1}. Therefore \eqref{eq:conservedlowerbound_1} can be written as
\begin{equation}\nonumber
\int_{[0,\infty)}u(t,x)\Psi_a(t,x)\dd x\geq\int_{[0,\infty)}u(0,x)\Psi_0(x)\dd x,
\end{equation}
so, using $u$ as a test function in \eqref{eq:evolutionofPsia}, we see that \eqref{eq:conservedlowerbound_1} is equivalent to
\begin{multline}\label{eq:conlowbnd_2b}
\int_0^t\Bigg[\int_{[0,\infty)}\left(u_s(s,x)-\tfrac1\rho\left(x u_x(s,x)+(2-\rho)u(s,x)\right)\right)\Psi_a(s,x)\dd x\\
+\int_{[0,\infty)}\left(\int_0^x\frac{(\phi_a\ast\Psi_a(s,\cdot))(y)}{((x+\varepsilon)(y+\varepsilon))^{3/2}}\mathcal{D}_2^*[u(s,\cdot)](x,y)\dd y\right)\Psi_a(s,x)\dd x\Bigg]\dd s\geq0,
\end{multline}
where we recall that $\mathcal{D}_2^*$ is defined in \eqref{eq:defcalDstar}.

Next, for $x\geq y\geq0$ we note that $(x+\varepsilon)^{-3/2}\leq\frac1x(y+\varepsilon)^{-1/2}$, and, defining $U(s,x):=xu(s,x)$, that $\mathcal{D}_2^*[u(s,\cdot)](x,y)=\Delta_y^2[U(s,\cdot)](x)$. Further, since $U(s,\cdot)$ is odd, bounded and concave on $\mathbb{R}_+$ for all $s\geq0$ (cf.~Lemma \ref{lm:oddfunctionsunderide}), by \eqref{eq:signofseconddifference} it holds for all $x,y\geq0$ that $\Delta_y^2[U(s,\cdot)](x)\leq0$, which together yields the estimate
\begin{multline}\label{eq:keyest_1}
\int_0^x\frac{(\phi_a\ast\Psi_a(s,\cdot))(y)}{((x+\varepsilon)(y+\varepsilon))^{3/2}}\mathcal{D}_2^*[u(s,\cdot)](x,y)\dd y\\
\geq\frac1x\int_{\mathbb{R}_+}\frac{(\phi_a\ast\Psi_a(s,\cdot))(y)}{(y+\varepsilon)^2}\Delta_y^2[U(s,\cdot)](x)\dd y.
\end{multline}
By an integration by parts, and using \eqref{eq:rewriteofsecdifprime} for $\partial_y[\Delta_y^2[U(s,\cdot)](x)]$, we can now write the right hand side of \eqref{eq:keyest_1} as
\begin{equation}\label{eq:keyest_1b}
\frac1x\int_{\mathbb{R}_+}\left(\int_y^\infty\frac{(\phi_a\ast\Psi_a(s,\cdot))(z)}{(z+\varepsilon)^2}\dd z\right)\left(\int_{x-y}^{x+y}U_{ww}(s,w)\dd w\right)\dd y,
\end{equation}
where we note that the integral with respect to $w$ on the right hand side is nonpositive for $x,y\geq0$, since $U(s,\cdot)$ is odd and concave on $\mathbb{R}_+$ (note that $\int_{-a}^aU''(w)\dd w=0$ for $a\geq0$). To find a lower bound for \eqref{eq:keyest_1b} we thus need an upper bound for the integral with respect to $z$. We thereto note for $z\geq y$ that $(z+\varepsilon)^{-2}\leq\frac1{y^2}(1\wedge\frac{y}{z+\varepsilon})$, and by expanding the domain of integration we find
\begin{multline}\label{eq:keyest_2}
\int_y^\infty\frac{(\phi_a\ast\Psi_a(s,\cdot))(z)}{(z+\varepsilon)^2}\dd z
\leq\frac1{y^2}\int_{[0,\infty)}\left(1\wedge\tfrac{y}{z+\varepsilon}\right)(\phi_a\ast\Psi_a(s,\cdot))(z)\dd z\\
=\frac1{y^2}\int_{[0,\infty)}\left(\int_{[0,\infty)}\left(1\wedge\tfrac{y}{z+\varepsilon}\right)\phi_a(z-x)\dd z\right)\Psi_a(s,x)\dd x,
\end{multline}
where the equality holds by Fubini. Now, since for all $z-x\in{\rm supp}(\phi_a)$ there holds $|z-x|\leq a<\frac\varepsilon2$, we have for all $x\geq0$ that
\begin{equation}\nonumber
\int_{[0,\infty)}\left(1\wedge\tfrac{y}{x+(z-x+\varepsilon)}\right)\phi_a(z-x)\dd z
\leq\left(1\wedge\tfrac{y}{x+\frac\varepsilon2}\right)\int_{\mathbb{R}}\phi_a(z-x)\dd z
\leq\left(1\wedge\tfrac{y}{x}\right),
\end{equation}
which, using the definition of the norm, we can use to estimate the right hand side of \eqref{eq:keyest_2} by
\begin{equation}\nonumber
\frac1{y^2}\int_{[0,\infty)}\left(1\wedge\tfrac{y}{x}\right)\Psi_a(s,x)\dd x\leq\tfrac1{y^2}\|\Psi_a(s,\cdot)\|_\rho\,y^{2-\rho}\leq y^{-\rho}.
\end{equation}
Combining then the previous estimates, and recalling the nonpositivity of the integral with respect to $w$ in \eqref{eq:keyest_1b}, we obtain
\begin{multline}\label{eq:keyest_3}
\int_{[0,\infty)}\left(\int_0^x\frac{(\phi_a\ast\Psi_a(s,\cdot))(y)}{((x+\varepsilon)(y+\varepsilon))^{3/2}}\mathcal{D}_2^*[u(s,\cdot)](x,y)\dd y\right)\Psi_a(s,x)\dd x\\
\geq\int_{[0,\infty)}\left(\int_{\mathbb{R}_+}y^{-\rho}\times\int_{x-y}^{x+y}U_{ww}(s,w)\dd w\ \dd y\right)\tfrac1x\Psi_a(s,x)\dd x,
\end{multline}
where, by an integration by parts in the integral with respect to $y$, the right hand side of \eqref{eq:keyest_3} can be rewritten as
\begin{equation}\label{eq:keyest_4}
\int_{[0,\infty)}\left(\rho\int_{\mathbb{R}_+}y^{-\rho-1}\Delta_y^2[U(s,\cdot)](x)\dd y\right)\tfrac1x\Psi_a(s,x)\dd x.
\end{equation}

Also, since $u^\rho$ is the fundamental solution of \eqref{eq:ide}, we note that $U(s,x)=xu(s,x)$ by construction satisfies [cf.~\eqref{eq:defofu}]
\begin{equation}\label{eq:ide2}
U_s(s,x)+\tfrac{\rho-1}{\rho}U(s,x)-\tfrac1\rho xU_x(s,x)=-\rho e^{t-s}\int_{\mathbb{R}_+}y^{-\rho-1}\Delta_y^2[U(s,\cdot)](x)\dd y.
\end{equation}
Checking then that the left hand side of \eqref{eq:ide2} can be rewritten as
\begin{equation}\nonumber
x\left(u_s(s,x)-\tfrac1\rho\left(x u_x(s,x)+(2-\rho)u(s,x)\right)\right),
\end{equation}
we find that the first integral between square brackets on the left hand side of \eqref{eq:conlowbnd_2b} equals
\begin{equation}\label{eq:keyest_5}
\int_{[0,\infty)}\left(-\rho e^{t-s}\int_{\mathbb{R}_+}y^{-\rho-1}\Delta_y^2[U(s,\cdot)](x)\dd y\right)\tfrac1x\Psi_a(s,x)\dd x.
\end{equation}

Concluding, since the first and second integral between square brackets on the left hand side of \eqref{eq:conlowbnd_2b} can be estimated from below by \eqref{eq:keyest_5} and \eqref{eq:keyest_4} respectively, the left hand side of \eqref{eq:conlowbnd_2b} can be bounded from below by
\begin{equation}\nonumber
\int_0^t\left(\int_{[0,\infty)}\left(\rho(1-e^{t-s})\int_{\mathbb{R}_+}y^{-\rho-1}\Delta_y^2[U(s,\cdot)](x)\dd y\right)\tfrac1x\Psi_a(s,x)\dd x\right)\dd s,
\end{equation}
which is nonnegative since both $1-e^{t-s}$ and $\Delta_y^2[U(s,\cdot)](x)$ are nonpositive on the domain of integration, while all other terms are nonnegative. This proves that \eqref{eq:conlowbnd_2b} holds, and since \eqref{eq:conservedlowerbound_1} and \eqref{eq:conlowbnd_2b} are equivalent the proof is complete.
\end{proof}
The following two lemmas will be useful in the actual proof of invariance of \eqref{eq:lowerboundinY} under the evolution \eqref{eq:evolutionofPsia}, where we will use a suitable function $\vartheta$ in \eqref{eq:conservedlowerbound_1}.
\begin{lm}\label{lm:integrationbyparts}
For $\rho\in(1,2)$ and $\Psi\in\mathcal{X}_\rho$, it holds for all odd $\Theta\in C^2([-\infty,\infty])$ that satisfy $\lim_{x\rightarrow\infty}\Theta'(x)x^{2-\rho}=0$ that
\begin{equation}\label{eq:integrationbyparts}
\int_{[0,\infty)}\Theta(x)\cdot\tfrac1x\Psi(x)\dd x=-\int_{[0,\infty)}\Theta''(x)\left(\int_{[0,\infty)}\left(1\wedge\tfrac{x}{z}\right)\Psi(z)\dd z\right)\dd x.
\end{equation}
\end{lm}
\begin{proof}
Observing that
\begin{equation}\nonumber
\int_{[0,\infty)}\left(1\wedge\tfrac{x}{z}\right)\Psi(z)\dd z=\int_0^x\int_y^\infty\tfrac1z\Psi(z)\dd z\dd y,
\end{equation}
it is clear that \eqref{eq:integrationbyparts} follows by integrating by parts twice, provided that the boundary values vanish. Since $\Theta$ is odd with bounded first derivative, we have $\Theta(x)=\Theta'(0)x+o(x^2)$ as $x\rightarrow0$, so
\begin{multline}\nonumber
\left|\Theta(x)\int_x^\infty\tfrac1z\Psi(z)\dd z\right|\leq2|\Theta'(0)|\int_{[0,\infty)}\left(1\wedge\tfrac{x}{z}\right)\Psi(z)\dd z\\
\leq2|\Theta'(0)|\|\Psi\|_\rho\cdot x^{2-\rho}\rightarrow0\text{ as }x\rightarrow0,
\end{multline}
where the second inequality holds by definition of the norm. Notice further that by our choice of $\Theta$, and using again the definition of the norm, we have
\begin{equation}\nonumber
\left|\Theta'(x)\int_{[0,\infty)}\left(1\wedge\tfrac{x}{z}\right)\Psi(z)\dd z\right|\leq\|\Psi\|_\rho\cdot|\Theta'(x)|x^{2-\rho}\rightarrow0\text{ as }x\rightarrow\infty.
\end{equation}
The claim then follows as the remaining boundary values vanish trivially.
\end{proof}
\begin{lm}\label{lm:inbThetadef}
For $\rho\in(1,2)$, let $v_\rho$ be the self-similar profile associated to the fundamental solution of \eqref{eq:ide} with $\alpha=\rho$. Then for all $\theta_1,\theta_2>0$ the function
\begin{equation}\label{eq:Theta}
\Theta(x)=\int_\mathbb{R}y\left(1\wedge\left|\tfrac{\theta_1}{y}\right|\right)v_\rho\left(\tfrac{x-y}{\theta_2}\right)\tfrac{\dd y}{\theta_2}
\end{equation}
is odd, smooth, and satisfies $\lim_{x\rightarrow\infty}\Theta'(x)x^{2-\rho}=0$ and
\begin{equation}\label{eq:signofThetadoubleprime}
-\Theta''(x)=\left(v_\rho\left(\tfrac{x-\theta_1}{\theta_2}\right)-v_\rho\left(\tfrac{x+\theta_1}{\theta_2}\right)\right)\tfrac1{\theta_2}\geq0\text{ for }x\geq0.
\end{equation}
\end{lm}
\begin{proof}
That $\Theta$ is odd follows from the fact that it is the convolution of an odd and an even function, while smoothness follows since $v_\rho\in C^\infty(\mathbb{R})$. We now note that $\partial_x f(x-y)=-\partial_y f(x-y)$, so differentiating \eqref{eq:Theta} we obtain
\begin{equation}\label{eq:Thetaprime}
\Theta'(x)
=-\int_\mathbb{R}y\left(1\wedge\left|\tfrac{\theta_1}{y}\right|\right)\left[v_\rho\left(\tfrac{x-y}{\theta_2}\right)\right]_y\tfrac{\dd y}{\theta_2}\\
=\int_{-\theta_1}^{\theta_1}v_\rho\left(\tfrac{x-y}{\theta_2}\right)\tfrac{\dd y}{\theta_2},
\end{equation}
where the second equality follows by integration by parts. Differentiating \eqref{eq:Thetaprime} once more, we then obtain the equality in \eqref{eq:signofThetadoubleprime}, while the nonnegativity follows from the symmetry and monotonicity properties of $v_\rho$. Finally, by symmetry and using the tail behaviour of $v_\rho$ (cf.~Lemma \ref{lm:fractionalheatequation}), we find that
\begin{equation}\nonumber
\Theta'(x)=\int_{(x-\theta_1)/\theta_2}^{(x+\theta_1)/\theta_2}v_\rho(z)\dd z
\leq\int_{(x-\theta_1)/\theta_2}^\infty v_\rho(z)\dd z\sim\frac1\rho\frac{\theta_2^\rho}{x^\rho}\text{ as }x\rightarrow\infty,
\end{equation}
hence $\Theta'(x)x^{2-\rho}\leq\frac1\rho\theta_2^\rho\cdot x^{2(1-\rho)}\rightarrow0$ as $x\rightarrow\infty$.
\end{proof}
We are now able to prove the following.
\begin{pr}\label{pr:thelowerbound}
Under Assumption \ref{assumption}, there exists some $R_0>0$, independent of $a$ and $\varepsilon$, such that the set $\mathcal{Y}_\rho=\mathcal{Y}_\rho(R_0)$ is invariant under the evolution of the semigroup $(S_a(t))_{t\geq0}$ as defined in Definition \ref{df:definitionofsemigroup}.
\end{pr}
\begin{proof}
Since the semigroup $(S_a(t))_{t\geq0}$ maps the unit ball of $\mathcal{X}_\rho$ into itself, we only need to prove preservation of the lower bound [cf.~\eqref{eq:lowerboundinY}] for some $R_0>0$. Let thus $R_0>0$ be arbitrary for now, and fix any $\Psi_0\in\mathcal{Y}_\rho=\mathcal{Y}_\rho(R_0)$. We will then write $\Psi_a(t,x)=S(t)\Psi_0(x)$ for all $t\geq0$.

Now, for any $R>0$ the function $\vartheta(z)=1\wedge|\frac{R}{z}|$ satisfies the assumptions from Lemma \ref{lm:lowerboundfirstlemma}, so by \eqref{eq:conservedlowerbound_1} and the change of variables $y\rightarrow ye^{-t/\rho}$ we find for all $t\in[0,T]$ and all $R>0$ that
\begin{multline}\label{eq:lbnd_2_0}
\int_{[0,\infty)}\left(1\wedge\tfrac{R}{x}\right)\Psi_a(t,x)\dd x\\
\geq e^{t(\rho-2)/\rho}\int_{[0,\infty)}\left(\int_\mathbb{R}\frac{y\left(1\wedge|\frac{Re^{t/\rho}}{y}|\right)}{(\rho te^t)^{1/\rho}}v_\rho\left(\frac{x-y}{(\rho te^t)^{1/\rho}}\right)\dd y\right)\tfrac1x\Psi_0(x)\dd x,
\end{multline}
where we recall that $v_\rho$ is the self-similar profile associated to the fundamental solution of \eqref{eq:ide} with $\alpha=\rho$ (cf.~Lemma \ref{lm:fractionalheatequation}). Multiplying then \eqref{eq:lbnd_2_0} by $R^{\rho-2}$, and using Lemmas \ref{lm:integrationbyparts} and \ref{lm:inbThetadef}, we obtain
\begin{multline}\label{eq1}
R^{\rho-2}\int_{[0,\infty)}\left(1\wedge\tfrac{R}{x}\right)\Psi_a(t,x)\dd x\\
\geq \left(Re^{t/\rho}\right)^{\rho-2}\int_{[0,\infty)}\left(v_\rho\left(\tfrac{x-Re^{t/\rho}}{(\rho te^t)^{1/\rho}}\right)-v_\rho\left(\tfrac{x+Re^{t/\rho}}{(\rho te^t)^{1/\rho}}\right)\right)\\
\times\left(\int_{[0,\infty)}\left(1\wedge\tfrac{x}{z}\right)\Psi_0(z)\dd z\right)\tfrac{\dd x}{(\rho te^t)^{1/\rho}},
\end{multline}
which, by the inequality in \eqref{eq:signofThetadoubleprime} and using the lower bound on $\Psi_0\in\mathcal{Y}_\rho$ (cf.~Definition \ref{df:defofY}), we can bound from below by
\begin{multline}\label{eq2}
\left(Re^{t/\rho}\right)^{\rho-2}\int_{[0,\infty)}\left(v_\rho\left(\tfrac{x-Re^{t/\rho}}{(\rho te^t)^{1/\rho}}\right)-v_\rho\left(\tfrac{x+Re^{t/\rho}}{(\rho te^t)^{1/\rho}}\right)\right)x^{2-\rho}\lambda_\rho\big(\tfrac{x}{R_0}\big)\tfrac{\dd x}{(\rho te^t)^{1/\rho}}\\
=\left(Re^{t/\rho}\right)^{\rho-2}\int_{\mathbb{R}}v_\rho\left(\tfrac{x-Re^{t/\rho}}{(\rho te^t)^{1/\rho}}\right)x|x|^{1-\rho}\lambda_\rho\big(\tfrac{x}{R_0}\big)\tfrac{\dd x}{(\rho te^t)^{1/\rho}}\\
=\left(\tfrac{Re^{t/\rho}}{R_0}\right)^{\rho-2}u\left(\tfrac{\rho te^t}{R_0^\rho},\tfrac{Re^{t/\rho}}{R_0}\right),
\end{multline}
where $u$ is the solution to \eqref{eq:ide} with $\alpha=\rho$ and $u(0,x)=x|x|^{1-\rho}\lambda_\rho(x)$. In view of the fact that the left hand side of \eqref{eq1} is nonnegative for all $R>0$, and since the right hand side of \eqref{eq2} does not depend on $\Psi_0$ any more, it now only remains to show that, for all $R>R_0$, with $R_0>0$ chosen appropriately, the left hand side of \eqref{eq2} is bounded from below by $\lambda_\rho(\frac{R}{R_0})$ as $t\rightarrow0$, or equivalently that
\begin{equation}\label{eq2b}
\left(\xi e^{t/\rho}\right)^{\rho-2}u\left(\tfrac{\rho te^t}{R_0^\rho},\xi e^{t/\rho}\right)\geq\lambda_\rho(\xi)\text{ as }t\rightarrow0,\quad\text{for all }\xi>1.
\end{equation}
Let now $u^*$ be the solution to \eqref{eq:ide} with $\alpha=\rho$ and
\begin{equation}\nonumber
u^*(0,x)=x|x|^{1-\rho}\left(1-|x|^{-(2-\rho)/2}\right)\in C(\mathbb{R})\cap L^1(\mathbb{R};|x|^{-\rho-1}\dd x).
\end{equation}
We then observe that $u^*(0,\cdot)\leq u(0,\cdot)$ on $\mathbb{R}_+$, so by the maximum principle in Lemma \ref{lm:oddfunctionsunderide} we have that $u^*(\tau,\cdot)\leq u(\tau,\cdot)$ on $\mathbb{R}_+$ for all $\tau>0$, and in particular
\begin{equation}\label{eq:lowerboundeq0}
\left(\xi e^{t/\rho}\right)^{\rho-2}u\left(\tfrac{\rho te^t}{R_0^\rho},\xi e^{t/\rho}\right)\geq\left(\xi e^{t/\rho}\right)^{\rho-2}u^*\left(\tfrac{\rho te^t}{R_0^\rho},\xi e^{t/\rho}\right)
\end{equation}
for $\xi>1$ and $t>0$. Now, for $x\geq0$ and $\tau>0$, we expand the convolution of $u^*(0,\cdot)$ with the fundamental solution to \eqref{eq:ide} to obtain
\begin{multline}\nonumber
x^{\rho-2}u^*(\tau,x)=x^{\rho-2}\int_\mathbb{R}(x-y)\left|x-y\right|^{1-\rho}\left(1-\left|x-y\right|^{-(2-\rho)/2}\right)v_\rho\left(\tfrac{y}{\tau^{1/\rho}}\right)\tfrac{\dd y}{\tau^{1/\rho}}\\
=\int_\mathbb{R}\left(1-\tfrac{y}{x}\right)\left|1-\tfrac{y}{x}\right|^{1-\rho}\left(1-x^{-(2-\rho)/2}\left|1-\tfrac{y}{x}\right|^{-(2-\rho)/2}\right)v_\rho\left(\tfrac{y}{\tau^{1/\rho}}\right)\tfrac{\dd y}{\tau^{1/\rho}},
\end{multline}
which, using the notation
\begin{equation}\nonumber
W(x,\zeta):=\left(1-\zeta\right)\left|1-\zeta\right|^{1-\rho}\left(1-x^{-(2-\rho)/2}\left|1-\zeta\right|^{-(2-\rho)/2}\right),
\end{equation}
and the facts that $\int_\mathbb{R}v_\rho(z)\dd z=1$ and $\int_\mathbb{R}zv_\rho(z)\dd z=0$, yields
\begin{multline}\label{eq:lowerboundeq1}
x^{\rho-2}u^*(\tau,x)=W(x,0)\\+\int_\mathbb{R}\left(W\left(x,\tfrac{y}{x}\right)-W(x,0)-\tfrac{y}{x}W_\zeta(x,0)\right)v_\rho\left(\tfrac{y}{\tau^{1/\rho}}\right)\tfrac{\dd y}{\tau^{1/\rho}},
\end{multline}
Next, we estimate $v_\rho$ by its tail behaviour (cf.~Lemma \ref{lm:fractionalheatequation}), to find that we can bound the absolute value of the second term on the right hand side of \eqref{eq:lowerboundeq1} by
\begin{multline}\nonumber
\int_\mathbb{R}\left|W\left(x,\tfrac{y}{x}\right)-W(x,0)-\tfrac{y}{x}W_\zeta(x,0)\right|\left|\tfrac{y}{\tau^{1/\rho}}\right|^{-1-\rho}\tfrac{\dd y}{\tau^{1/\rho}}\\
\leq\frac{\tau}{x^\rho}\int_\mathbb{R}\left|W(x,\zeta)-W(x,0)-\zeta W_\zeta(x,0)\right|\left|\zeta\right|^{-1-\rho}\dd \zeta=:\frac{\tau}{x^\rho}K(x),
\end{multline}
so noticing then that $K(x)$ can be uniformly bounded for $x\geq1$ by a constant $\kappa>0$, we obtain that
\begin{equation}\label{eq:lowerboundeq2}
x^{\rho-2}u^*(\tau,x)\geq\left(1-x^{-(2-\rho)/2}\right)-\frac{\tau\kappa}{x^\rho}\text{ for all }x\geq1\text{ and all }\tau>0.
\end{equation}
Combining thus \eqref{eq:lowerboundeq0} and \eqref{eq:lowerboundeq2} we find for $\xi>1$ and $t>0$ that
\begin{multline}\label{eq:lowerboundeq3}
\left(\xi e^{t/\rho}\right)^{\rho-2}u\left(\tfrac{\rho te^t}{R_0^\rho},\xi e^{t/\rho}\right)
\geq\left(1-\left(\xi e^{t/\rho}\right)^{-(2-\rho)/2}\right)-\frac{\frac{\rho te^t}{R_0^\rho}\cdot\kappa}{\left(\xi e^{t/\rho}\right)^\rho}\\
=\left(1-\xi^{-(2-\rho)/2}\right)+\left(1-e^{-t(2-\rho)/(2\rho)}\right)\xi^{-(2-\rho)/2}-t\tfrac{\rho\kappa}{R_0^\rho}\xi^{-\rho},
\end{multline}
where we note that the first term on the right hand side equals $\lambda_\rho(\xi)$ for $\xi>1$ (since it is positive). Expanding then finally the exponential in the second term, we can bound the right hand side of \eqref{eq:lowerboundeq3}, for $t\rightarrow0$, from below by
\begin{equation}\nonumber
\lambda_\rho(\xi)+t\left(\tfrac{2-\rho}{4\rho}\xi^{-(2-\rho)/2}-\tfrac{\rho\kappa}{R_0^\rho}\xi^{-\rho}\right),
\end{equation}
where, provided that $R_0>0$ is such that $R_0^\rho\geq\frac{4\rho^2\kappa}{2-\rho}$, the second term is nonnegative for $\xi>1$, which implies that \eqref{eq2b} holds, and thus proves the claim.
\end{proof}

\subsection{Proof of Proposition \ref{pr:wellposednessforregularizedevolution}}
\begin{proof}[Proof of Proposition \ref{pr:wellposednessforregularizedevolution}]
Let $(S_a(t))_{t\geq0}$ be the semigroup on the unit ball in $\mathcal{X}_\rho$ that was defined in Definition \ref{df:definitionofsemigroup}. Following Proposition \ref{pr:thelowerbound} there then exists some $R_0>0$ such that the set $\mathcal{Y}_\rho=\mathcal{Y}_\rho(R_0)$ (cf.~Definition \ref{df:defofY}) is invariant under the evolution of $(S_a(t))_{t\geq0}$. It thus remains to check, for any fixed $t>0$, that the mapping $\Psi_0\mapsto\Psi_a(t,\cdot):=S_a(t)\Psi_0$ is weakly-$*$ continuous, which by continuity of the change of variables \eqref{eq:changeofvariables} is equivalent to checking weak-$*$ continuity of $\Psi_0\mapsto H_a(t,\cdot)$.

Now, for any two measures $\Psi_1,\Psi_2\in\{\|\mu\|_\rho\leq1\}\cap\mathcal{X}_\rho$, we write $\Psi_a^i(s,\cdot):=S_a(s)\Psi_i$, $i\in\{1,2\}$, for all $s\in[0,t]$, and we let $H_a^i\in C([0,t]:\mathcal{X}_\rho)$ be defined via the change of variables \eqref{eq:changeofvariables}, which are functions that satisfy \eqref{eq:approximatemildsolution} for all $\psi\in C^1([0,t]:\mathcal{B}_0)$. What we need to show is that for any $\psi^*\in\mathcal{B}_0$ and any $\delta>0$ small, there exists a weakly-$*$ open set $\mathcal{U}=\mathcal{U}(\psi^*,\delta)$ such that $\Psi_1-\Psi_2\in\mathcal{U}$ implies
\begin{equation}\label{eq:continuityofsemigroup_1}
\left|\int_{[0,\infty)}\psi^*(X)H_a^1(t,X)\dd X-\int_{[0,\infty)}\psi^*(X)H_a^2(t,X)\dd X\right|<\delta.
\end{equation}
By a density argument we may restrict ourselves to $\psi^*\in\mathcal{B}_1$, and by \eqref{eq:approximatemildsolution} we know that
\begin{multline}\label{eq:continuityofsemigroup_2}
\int_{[0,\infty)}\psi^*(X)\left(H_a^1(t,X)-H_a^2(t,X)\right)\dd X\\=\int_{[0,\infty)}\psi(0,X)\left(\Psi_1(X)-\Psi_2(X)\right)\dd X,
\end{multline}
if $\psi\in C^1([0,t]:\mathcal{B}_1)$ satisfies $\psi(t,\cdot)=\psi^*$ on $[0,\infty]$, and, for all $s\in(0,t)$,
\begin{multline}\nonumber
0=\int_{[0,\infty)}\left(\psi_s(s,X)+\tfrac{\rho-1}{\rho}\psi(s,X)\right)\left(H_a^1(s,X)-H_a^2(s,X)\right)\dd X\\
+q(H_a^1,\psi,s)-q(H_a^2,\psi,s)
\end{multline}
with
\begin{multline}\nonumber
q(H_a^i,\psi,s)=\iint_{\{X>Y>0\}}\frac{e^{s/\rho}H_a^i(s,X)(\phi_{ae^{s/\rho}}\ast H_a^i(s,\cdot))(Y)}{((X+\varepsilon e^{s/\rho})(Y+\varepsilon e^{s/\rho}))^{3/2}}\\\times\mathcal{D}_2^*[\psi(s,\cdot)](X,Y)\dd X\dd Y.
\end{multline}
Using the identity $h_1h_1^*-h_2h_2^*=\frac12(h_1-h_2)(h_1^*+h_2^*)+\frac12(h_1+h_2)(h_1^*-h_2^*)$, we now rewrite the difference $q(H_a^1,\psi,s)-q(H_a^2,\psi,s)$ as
\begin{multline}\nonumber
\frac12\iint_{\{X>Y>0\}}\Big[(H_a^1(s,X)-H_a^2(s,X))(\phi_{ae^{s/\rho}}\ast (H_a^1(s,\cdot)+H_a^2(s,\cdot)))(Y)\\
+(H_a^1(s,X)+H_a^2(s,X))(\phi_{ae^{s/\rho}}\ast (H_a^1(s,\cdot)-H_a^2(s,\cdot)))(Y)\Big]\\\times\frac{e^{s/\rho}\mathcal{D}_2^*[\psi(s,\cdot)](X,Y)}{((X+\varepsilon e^{s/\rho})(Y+\varepsilon e^{s/\rho}))^{3/2}}\dd X\dd Y,
\end{multline}
which in turn equals
\begin{multline}\nonumber
\int_{[0,\infty)}\left(\ell(H_a^1,\psi;s,X)+\ell(H_a^2,\psi;s,X)+\ell^*(H_a^1,\psi;s,X)+\ell^*(H_a^2,\psi;s,X)\right)\\\times\left(H_a^1(s,X)-H_a^2(s,X)\right)\dd X,
\end{multline}
with
\begin{equation}\nonumber
\ell(H_a^i,\psi;s,X)=\frac12\int_0^X\frac{e^{s/\rho}(\phi_{ae^{s/\rho}}\ast H_a^i(s,\cdot))(Y)}{((X+\varepsilon e^{s/\rho})(Y+\varepsilon e^{s/\rho}))^{3/2}}\mathcal{D}_2^*[\psi(s,\cdot)](X,Y)\dd Y,
\end{equation}
and
\begin{multline}\nonumber
\ell^*(H_a^i,\psi;s,X)=\frac12\int_{[0,\infty)}\phi_{ae^{s/\rho}}(Y-X)\int_{[Y,\infty]}\frac{e^{s/\rho}H_a^i(s,X')}{((X'+\varepsilon e^{s/\rho})(Y+\varepsilon e^{s/\rho}))^{3/2}}\\\times\mathcal{D}_2^*[\psi(s,\cdot)](X',Y)\dd X'\dd Y.
\end{multline}
We thus obtain the linear backward in time boundary value problem
\begin{equation}\nonumber
\begin{cases}
\psi_s=-\tfrac{\rho-1}{\rho}\psi-\left(\ell(H_a^1,\psi)+\ell(H_a^2,\psi)+\ell^*(H_a^1,\psi)+\ell^*(H_a^2,\psi)\right),\\
\psi(t,\cdot)=\psi^*,
\end{cases}
\end{equation}
which can be uniquely solved in $C^1([0,t],\mathcal{B}_1)$ by a standard fixed point argument, since $\ell(H_a^i,\cdot)$ and $\ell^*(H_a^i,\cdot)$ are bounded linear operators from $\mathcal{B}_1$ to itself. Moreover, by estimate \eqref{eq:normestimate} we find that there exists a constant $C>0$, independent of $\psi^*$, such that $\|\psi(0,\cdot)\|_{\mathcal{B}_1}\leq C\|\psi^*\|_{\mathcal{B}_1}$. Now, by compactness of $\|\cdot\|_{\mathcal{B}_1}$-bounded sets in $\mathcal{B}_0$, we can select finitely many $\omega_1,\ldots,\omega_n\in\mathcal{B}_0$ such that $\min_i\|\psi(0,\cdot)-\omega_i\|_{\mathcal{B}_0}<\tfrac13\delta$. For any $i\in\{1,\ldots,n\}$ we then write the right hand side of \eqref{eq:continuityofsemigroup_2} as
\begin{multline}\nonumber
\int_{[0,\infty)}\left(\psi(0,X)-\omega_i(X)\right)\left(\Psi_1(X)-\Psi_2(X)\right)\dd X,\\+\int_{[0,\infty)}\omega_i(X)\left(\Psi_1(X)-\Psi_2(X)\right)\dd X,
\end{multline}
so defining finally
\begin{equation}\nonumber
\mathcal{U}=\left\{\mu\in\mathcal{M}([0,\infty))\,:\,\|\mu\|_\rho<\infty\text{ and }\max_i\left|\int_{[0,\infty)}\omega_i(X)\mu(X)\dd X\right|<\tfrac13\delta\right\},
\end{equation}
it follows, by choosing $i\in\{1,\ldots,n\}$ such that $\|\psi(0,\cdot)-\omega_i\|_{\mathcal{B}_0}$ is minimal, that if $\Psi_1-\Psi_2\in\mathcal{U}$, then \eqref{eq:continuityofsemigroup_1} holds, which completes the proof.
\end{proof}

\subsection{Proof of Theorem \ref{tm:existence}}
\begin{proof}[Proof of Theorem \ref{tm:existence}]
In view of Remark \ref{rk:restriction} we restrict ourselves to the case $\rho\in(1,2)$. For any $\varepsilon>0$ fixed and $a\in(0,\frac\varepsilon2)$ arbitrary, let $R_0>0$ and $(S_a(t))_{t\geq0}$ be as obtained in Proposition \ref{pr:wellposednessforregularizedevolution}. Then, by a variant of Tychonoff's theorem (cf.~\cite[Thm.~1.2]{EMR05}) there exists some $\Psi_a^\varepsilon\in\mathcal{Y}_\rho$ for which for all $t\geq0$ there holds $S(t)\Psi_a^\varepsilon=\Psi_a^\varepsilon$, and which for all $\vartheta\in\mathcal{B}_1$ satisfies
\begin{multline}\label{eq:Tex_1}
\frac1\rho\int_{[0,\infty)}\left(x\vartheta'(x)+(2-\rho)\vartheta(x)\right)\Psi_a^\varepsilon(x)\dd x\\
=\iint_{\{x>y>0\}}\frac{\Psi_a^\varepsilon(x)(\phi_a\ast\Psi_a^\varepsilon)(y)}{((x+\varepsilon)(y+\varepsilon))^{3/2}}\mathcal{D}_2^*[\vartheta](x,y)\dd x\dd y.
\end{multline}
Further, since $\mathcal{Y}_\rho$ is compact and independent of $a$, there exist $a_n\rightarrow0$ and $\Psi^\varepsilon\in\mathcal{Y}_\rho$ such that $\Psi_{a_n}^\varepsilon\wsc\Psi^\varepsilon$ in $\mathcal{X}_\rho$, and we will see that $\Psi^\varepsilon$ satisfies
\begin{multline}\label{eq:Tex_2}
\frac1\rho\int_{[0,\infty)}\left(x\vartheta'(x)+(2-\rho)\vartheta(x)\right)\Psi^\varepsilon(x)\dd x\\
=\frac12\iint_{[0,\infty)^2}\frac{\Psi^\varepsilon(x)\Psi^\varepsilon(y)}{((x+\varepsilon)(y+\varepsilon))^{3/2}}\mathcal{D}_2^*[\vartheta](x,y)\dd x\dd y,
\end{multline}
for all $\vartheta\in\mathcal{B}_1$. Indeed, writing $a$ for $a_n$ and $a\rightarrow0$ for $a_n\rightarrow0$, for any $\vartheta\in\mathcal{B}_1$ fixed it follows from the definition of weak-$*$ convergence that the left hand side of \eqref{eq:Tex_1} converges to the left hand side of \eqref{eq:Tex_2}. Convergence of the right hand side is more tricky since in the limit there might be a nontrivial contribution along the diagonal $\{x=y\geq0\}$. Expanding the convolution and using Fubini, we first of all rewrite the right hand side of \eqref{eq:Tex_1} as
\begin{multline}\nonumber
\int_{[0,\infty)}\frac{\Psi_a^\varepsilon(x)}{(x+\varepsilon)^{3/2}}\left(\int_0^x\left(\int_{[0,\infty)}\phi_a(y-z)\Psi_a^\varepsilon(z)\dd z\right)\frac{\mathcal{D}_2^*[\vartheta](x,y)}{(y+\varepsilon)^{3/2}}\dd y\right)\dd x\\
=\iint_{[0,\infty)^2}\frac{\Psi_a^\varepsilon(x)\Psi_a^\varepsilon(z)}{((x+\varepsilon)(z+\varepsilon))^{3/2}}\left(\int_0^x\frac{\phi_a(y-z)\mathcal{D}_2^*[\vartheta](x,y)}{((z+\varepsilon)^{-1}(y+\varepsilon))^{3/2}}\dd y\right)\dd x\dd z,
\end{multline}
which by symmetrization equals
\begin{equation}\label{eq:T1_2}
\frac12\iint_{[0,\infty)^2}\frac{\Psi_a^\varepsilon(x)\Psi_a^\varepsilon(z)}{((x+\varepsilon)(z+\varepsilon))^{3/2}}\mathcal{E}_a(x,z)\dd x\dd z,
\end{equation}
with
\begin{equation}\nonumber
\mathcal{E}_a(x,z)
=\int_{-\infty}^x\frac{\phi_a(y-z)\mathcal{D}_2^*[\vartheta](x,y)}{((z+\varepsilon)^{-1}(y+\varepsilon))^{3/2}}\dd y
+\int_{-\infty}^z\frac{\phi_a(y-x)\mathcal{D}_2^*[\vartheta](z,y)}{((x+\varepsilon)^{-1}(y+\varepsilon))^{3/2}}\dd y,
\end{equation}
where the extension of the domains of integration until $-\infty$ is possible if we extend $\mathcal{D}_2^*[\vartheta]$ to a continuous function on $\mathbb{R}^2$ by setting $\mathcal{D}_2^*[\vartheta]=0$ on $\mathbb{R}^2\setminus\mathbb{R}_+^2$. Observing next that
\begin{multline}\nonumber
\int_{-\infty}^x\phi_a(y-z)\dd y+\int_{-\infty}^z\phi_a(y-x)\dd y
=\int_{-\infty}^{x-z}\phi_a(y)\dd y+\int_{-\infty}^{z-x}\phi_a(y)\dd y\\
=\int_{-\infty}^{x-z}\phi_a(y)\dd y+\int_{x-z}^\infty\phi_a(-y)\dd y=1\text{ for all }x,z\in\mathbb{R},
\end{multline}
where the last equality holds since $\phi_a$ is even, we then find by continuity and symmetry of $\mathcal{D}_2^*[\vartheta]$ that for all $\delta>0$ there exists some $a_\delta>0$ such that
\begin{equation}\nonumber
\left|\mathcal{E}_a(x,z)-\mathcal{D}_2^*[\vartheta](x,z)\right|<\delta\text{ for all }a\in(0,a_\delta)\text{ and all }x,z\geq0.
\end{equation}
Using then the fact that $(x+\varepsilon)^{-3/2}\leq\varepsilon^{-3/2}(1+\frac{x}{\varepsilon})^{-1}\leq\varepsilon^{-3/2}(1\wedge\frac{\varepsilon}{x})$ for $x\geq0$, the definition of the norm $\|\cdot\|_\rho$, and the fact that $\|\Psi_a^\varepsilon\|_\rho=1$ for all $\Psi_a^\varepsilon\in\mathcal{Y}_\rho$, we obtain for $a\in(0,a_\delta)$ that
\begin{multline}\nonumber
\left|\frac12\iint_{[0,\infty)^2}\frac{\Psi_a^\varepsilon(x)\Psi_a^\varepsilon(z)}{((x+\varepsilon)(z+\varepsilon))^{3/2}}\left(\mathcal{E}_a(x,z)-\mathcal{D}_2^*[\vartheta](x,z)\right)\dd x\dd z\right|\\
\leq\frac\delta2\left(\varepsilon^{-3/2}\int_{[0,\infty)}\left(1\wedge\tfrac{\varepsilon}{x}\right)\Psi_a^\varepsilon(x)\dd x\right)^2
\leq\tfrac12\varepsilon^{1-2\rho}\times\delta.
\end{multline}
Further, as $(x+\varepsilon)^{-3/2}\leq r^{-1/2}\varepsilon^{-1}(1\wedge\frac{\varepsilon}{x})$ for $x\geq r>0$, we similarly get that
\begin{multline}\nonumber
\left|\frac12\iint_{[0,\infty)^2\setminus[0,r]^2}\frac{\Psi_a^\varepsilon(x)\Psi_a^\varepsilon(z)}{((x+\varepsilon)(z+\varepsilon))^{3/2}}\mathcal{D}_2^*[\vartheta](x,z)\dd x\dd z\right|\\
\leq4\|\vartheta\|_{\mathcal{B}_0}\left(\int_{[0,\infty)}\frac{\Psi_a^\varepsilon(x)}{(x+\varepsilon)^{3/2}}\dd x\right)\left(\int_{(r,\infty)}\frac{\Psi_a^\varepsilon(x)}{(x+\varepsilon)^{3/2}}\dd x\right)\\
\leq4\|\vartheta\|_{\mathcal{B}_0}\varepsilon^{(3-4\rho)/2}\times r^{-1/2},
\end{multline}
so recalling that on compact squares the finite sums of products of single variable functions are dense in the uniform topology in the continuous functions, it follows that in the limit $a\rightarrow0$, \eqref{eq:T1_2} becomes
\begin{equation}\nonumber
\frac12\iint_{[0,r]^2}\frac{\Psi^\varepsilon(x)\Psi^\varepsilon(y)}{((x+\varepsilon)(y+\varepsilon))^{3/2}}\mathcal{D}_2^*[\vartheta](x,y)\dd x\dd y+O(\delta)+O(r^{-1/2})
\end{equation}
as $\delta\rightarrow0$ and $r\rightarrow\infty$, and taking these limits we obtain that $\Psi^\varepsilon$ indeed satisfies \eqref{eq:Tex_2} for all $\vartheta\in\mathcal{B}_1$.\\

In order to be able to take the limit $\varepsilon\rightarrow0$ we will need the following estimate. We show that there is a constant $K>0$, independent of $\varepsilon$, such that
\begin{equation}\label{eq:Texlim_1}
\int_{(0,z]}\frac{\sqrt{x}\Psi^\varepsilon(x)}{(x+\varepsilon)^{3/2}}\dd x\leq K\cdot z^{1-\rho/2}\text{ for all }z>0.
\end{equation}
To that end, we note that $x\vartheta'(x)+(2-\rho)\vartheta(x)=[x\vartheta(x)]_x-(\rho-1)\vartheta(x)$, and choosing $\vartheta\in\mathcal{B}_1$ such that the mapping $z\mapsto z\vartheta(z)$ is nondecreasing and concave we obtain from \eqref{eq:Tex_2} that
\begin{multline}\label{eq:Texlim_2}
\tfrac{\rho-1}\rho\int_{[0,\infty)}\vartheta(x)\Psi^\varepsilon(x)\dd x\\
\geq-\frac12\iint_{[0,\infty)^2}\frac{\Psi^\varepsilon(x)\Psi^\varepsilon(y)}{((x+\varepsilon)(y+\varepsilon))^{3/2}}\mathcal{D}_2^*[\vartheta](x,y)\dd x\dd y,
\end{multline}
where $\mathcal{D}_2^*[\vartheta]\leq0$. Now, by an approximation argument \eqref{eq:Texlim_2} also holds with $\vartheta_r(x)=(1\wedge\frac{r}{x})$, for arbitrary $r>0$. We then use the definition of the norm $\|\cdot\|_\rho$, and the fact that $\|\Psi^\varepsilon\|_\rho=1$, to estimate the left hand side of \eqref{eq:Texlim_2} by $\frac{\rho-1}{\rho}r^{2-\rho}$. Noting furthermore that the second difference of an affine function is zero, we find for $x,y\geq0$ that
\begin{multline}\nonumber
\mathcal{D}_2^*[\vartheta_r](x,y)=\mathcal{D}_2[z\vartheta_r(z)-r](x,y)=-\mathcal{D}_2[(r-z)_+](x,y)\\
=-\big[(r-(x+y))_++(r-|x-y|)_+-2(r-(x\vee y))_+\big]\\
=-\big[(x+y-r)_+\wedge(r-|x-y|)_+\big],
\end{multline}
and it follows from \eqref{eq:Texlim_2} that
\begin{equation}\label{eq:Texlim_3}
\tfrac{\rho-1}{\rho}r^{2-\rho}\geq\frac12\iint_{[0,\infty)^2}\frac{\sqrt{x}\Psi^\varepsilon(x)}{(x+\varepsilon)^{3/2}}\frac{\sqrt{y}\Psi^\varepsilon(y)}{(y+\varepsilon)^{3/2}}\tfrac{(x+y-r)_+\wedge(r-|x-y|)_+}{\sqrt{xy}}\dd x\dd y.
\end{equation}
Set now $\alpha=\frac13(\sqrt7+1)>1$, which solves $1-(\alpha^2-\alpha)=\frac12\alpha^2$, and notice that $(xy)^{-1/2}(r-|x-y|)\geq (\alpha^2r)^{-1}(r-(\alpha^2-\alpha)r)=\frac12$ for $x,y\in(\alpha r,\alpha^2 r]$. We then restrict the domain of integration on the right hand side of \eqref{eq:Texlim_3} to $(\alpha r,\alpha^2 r]^2$ to obtain
\begin{equation}\nonumber
\frac14\left(\int_{(\alpha r,\alpha^2 r]}\frac{\sqrt{x}\Psi^\varepsilon(x)}{(x+\varepsilon)^{3/2}}\dd x\right)^2\leq\tfrac{\rho-1}{\rho}r^{2-\rho},
\end{equation}
hence there holds
\begin{equation}\nonumber
\int_{(\alpha^{-1}r,r]}\frac{\sqrt{x}\Psi^\varepsilon(x)}{(x+\varepsilon)^{3/2}}\dd x\leq\left(2\alpha^{\rho-2}\sqrt{\tfrac{\rho-1}{\rho}}\right)r^{1-\rho/2}\text{ for all }r>0,
\end{equation}
and using the decomposition $(0,z]=\bigcup_{j=0}^\infty(\alpha^{-j-1}z,\alpha^{-j}z]$ we obtain \eqref{eq:Texlim_1}.\\

Now, as $\mathcal{Y}_\rho$ is independent of $\varepsilon$, there also exist $\varepsilon_n\rightarrow0$ and $\Psi_\rho\in\mathcal{Y}_\rho$ such that $\Psi^{\varepsilon_n}\wsc\Psi_\rho$ in $\mathcal{X}_\rho$. Writing then $\varepsilon$ for $\varepsilon_n$ and $\varepsilon\rightarrow0$ for $\varepsilon_n\rightarrow0$, the left hand side of \eqref{eq:Tex_2}, by definition of weak-$*$ convergence, converges to
\begin{multline}\label{eq:Tex_3}
\frac1\rho\int_{[0,\infty)}\left(x\vartheta'(x)+(2-\rho)\vartheta(x)\right)\Psi_\rho(x)\dd x\\
=\frac1\rho\int_{[0,\infty)}\big(x[x\vartheta(x)]_x-(\rho-1)x\vartheta(x)\big)\tfrac1x\Psi_\rho(x)\dd x.
\end{multline}
We next check that \eqref{eq:Texlim_1} carries over to the limit. Let thereto $\eta_\delta\in C(\mathbb{R})$, for $\delta>0$, be nondecreasing with $\eta_\delta=0$ on $(-\infty,\delta)$ and $\eta_\delta=1$ on $(2\delta,\infty)$. Using then \eqref{eq:Texlim_1}, for all $z>0$ we obtain
\begin{multline}\label{eq:Texlim_4}
\int_{(0,z]}\tfrac1x\Psi_\rho(x)\dd x=\lim_{\delta\rightarrow0}\int_{(0,z+2\delta]}\tfrac1x\Psi_\rho(x)(\eta_\delta(x)-\eta_\delta(x-z))\dd x\\
=\lim_{\delta\rightarrow0}\lim_{\varepsilon\rightarrow0}\int_{(0,z+2\delta]}\frac{\sqrt{x}\Psi_\rho(x)}{(x+\varepsilon)^{3/2}}(\eta_\delta(x)-\eta_\delta(x-z))\dd x
\leq K\cdot z^{1-\rho/2}.
\end{multline}
Now, with $\eta_\delta$ as above and using \eqref{eq:Texlim_1}, we find for any $\psi\in C([0,\infty])$ that
\begin{multline}\nonumber
\left|\int_{[0,\infty)}\psi(x)\frac{\sqrt{x}\Psi^\varepsilon(x)}{(x+\varepsilon)^{3/2}}\dd x-\int_{[0,\infty)}\frac{\eta_\delta(x)\psi(x)}{(1+\frac\varepsilon{x})^{3/2}}\frac1x\Psi^\varepsilon(x)\dd x\right|\\
\leq\int_{(0,2\delta]}\frac{\sqrt{x}\Psi^\varepsilon(x)}{(x+\varepsilon)^{3/2}}\psi(x)\dd x\leq 2^{1-\rho/2}K\|\psi\|_{C([0,\infty])}\times\delta^{1-\rho/2},
\end{multline}
so similarly using \eqref{eq:Texlim_4}, and since $\Psi_\rho\in\mathcal{X}_\rho$ implies $\Psi_\rho(\{0\})=0$, we obtain
\begin{equation}\nonumber
\lim_{\varepsilon\rightarrow0}\int_{[0,\infty)}\psi(x)\frac{\sqrt{x}\Psi^\varepsilon(x)}{(x+\varepsilon)^{3/2}}\dd x
=\int_{[0,\infty)}\psi(x)\tfrac1x\Psi_\rho(x)\dd x.
\end{equation}
We thus find that the finite measures $\frac{\sqrt{x}\Psi^\varepsilon(x)}{(x+\varepsilon)^{3/2}}\dd x$ converge with respect to the weak-$*$ topology in $(C([0,\infty]))'$ to some $\Phi_\rho\in\mathcal{X}_2$ that satisfies $\Phi_\rho(x)=\frac1x\Psi_\rho(x)$ for $x>0$. Recall now that $|\mathcal{D}_2^*[\vartheta](x,y)|\leq4\|\vartheta\|_{\mathcal{B}_0}$ for all $x,y\geq0$, and also $|\mathcal{D}_2^*[\vartheta](x,y)|\leq2\|\vartheta\|_{\mathcal{B}_1}(x\wedge y)$. Then $(xy)^{-1/2}\mathcal{D}_2^*[\vartheta](x,y)$ is bounded and continuous on $[0,\infty]^2$, and vanishes uniformly on the boundary, and we thus obtain that the right hand side of \eqref{eq:Tex_2} converges to the right hand side of
\begin{multline}\label{eq:Tex_last}
\frac1\rho\int_{[0,\infty)}\big(x[x\vartheta(x)]_x-(\rho-1)x\vartheta(x)\big)\Phi_\rho(x)\dd x\\
=\iint_{[0,\infty)^2}\frac{\Phi_\rho(x)\Phi_\rho(y)}{\sqrt{xy}}\mathcal{D}_2^*[\vartheta](x,y)\dd x\dd y,
\end{multline}
which $\Phi_\rho$ satisfies for all $\vartheta\in\mathcal{B}_1$ [cf.~\eqref{eq:Tex_3}]. To conclude, we note that by appropriate approximation it can be shown that $\Phi_\rho$ satisfies \eqref{eq:Tex_last} for all $\vartheta$ for which the mapping $z\mapsto z\vartheta(z)$ is differentiable and constant from a certain point onwards towards infinity. The proof is then completed by the observation that any $\varphi\in C_c^1([0,\infty))$ defines such a function via $\varphi(x)-\varphi(0)= x\vartheta(x)$.
\end{proof}

\section{Regularity, and a decay result}\label{sec:regularity}
We start this section with a useful integrability estimate, which can be seen as an improvement on \eqref{eq:Texlim_1}.
\begin{lm}\label{lm:integrallimitnearzero}
Given $\rho\in(1,2]$, if $\Phi_\rho\in\mathcal{X}_2$ satisfies \eqref{eq:selfsimilarprofileM+} for all $\varphi\in C_c^1([0,\infty))$, then there exists a constant $C>0$ such that
\begin{equation}\label{eq:integrallimitnearzero}
\int_{(0,R]}\Phi_\rho(x)\dd x\leq C\cdot\sqrt{R}\text{ for all }R>0.
\end{equation}
\end{lm}
\begin{proof}
Given any nonincreasing convex function $\varphi\in C_c^1([0,\infty))$, we obtain from \eqref{eq:selfsimilarprofileM+} that there holds
\begin{multline}\label{eq:ienz_1}
\tfrac{\rho-1}\rho\int_{[0,\infty)}\left(\varphi(0)-\varphi(x)\right)\Phi_\rho(x)\dd x\\
\geq\frac12\iint_{[0,\infty)^2}\frac{\Phi_\rho(x)\Phi_\rho(y)}{\sqrt{xy}}\mathcal{D}_2[\varphi](x,y)\dd x\dd y,
\end{multline}
so using appropriate arguments to approximate $\varphi(x)=(r-x)_+$, with $r>0$ arbitrary, and arguing as in the proof of \eqref{eq:Texlim_1}, we find from \eqref{eq:ienz_1} that
\begin{equation}\label{eq:ienz_2}
\tfrac{\rho-1}\rho\|\Phi_\rho\|_2\cdot r\geq\tfrac12\mathcal{I}[\Phi_\rho](r),
\end{equation}
with
\begin{equation}\label{eq:defofcalI}
\mathcal{I}[\Phi_\rho](r):=\iint_{\mathbb{R}_+^2}\frac{\Phi_\rho(x)\Phi_\rho(y)}{\sqrt{xy}}(x+y-r)_+\wedge(r-|x-y|)_+\dd x\dd y.
\end{equation}
Also as in the proof of \eqref{eq:Texlim_1}, we then restrict the domain of integration on the right hand side of \eqref{eq:ienz_2} to $(\alpha r,\alpha^2 r]^2$, with $\alpha=\frac13(\sqrt7+1)$, to obtain
\begin{equation}\nonumber
\frac14\left(\int_{(\alpha r,\alpha^2r]}\Phi_\rho(x)\dd x\right)^2\leq\tfrac{\rho-1}\rho\|\Phi_\rho\|_2\cdot r,
\end{equation}
hence there holds
\begin{equation}\nonumber
\int_{(\alpha^{-1}r,r]}\Phi_\rho(x)\dd x\leq\tfrac{2}{\alpha}\sqrt{\tfrac{\rho-1}\rho\|\Phi_\rho\|_2}\cdot\sqrt{r}\text{ for all }r>0,
\end{equation}
and \eqref{eq:integrallimitnearzero} follows by the decomposition $(0,R]=\bigcup_{j=0}^\infty(\alpha^{-j-1}R,\alpha^{-j}R]$.
\end{proof}
We now first show that self-similar profiles are H\"older continuous.
\begin{lm}\label{lm:regularity}
Given $\rho\in(1,2]$, if $\Phi_\rho\in\mathcal{X}_2$ satisfies \eqref{eq:selfsimilarprofileM+} for all $\varphi\in C_c^1([0,\infty))$, then it is absolutely continuous with respect to Lebesgue measure, its Radon-Nykodim derivative is locally $\alpha$-H\"older continuous on $(0,\infty)$ for any $\alpha<\frac12$, and it actually satisfies \eqref{eq:selfsimilarprofileL1} for all $\varphi\in C_c^1([0,\infty))$.
\end{lm}
\begin{proof}
Given $\chi\in C_c^\infty((0,\infty))$, we set $\varphi(x)=-\int_x^\infty\frac1z\chi(z)\dd z$ and use this function in \eqref{eq:selfsimilarprofileM+} to obtain
\begin{multline}\label{eq:L42_1}
\int_{(0,\infty)}\chi(x)\Phi_\rho(x)\dd x-(\rho-1)\int_{(0,\infty)}\int_0^x\tfrac1z\chi(z)\dd z\Phi_\rho(x)\dd x\\
=\frac\rho2\iint_{[0,\infty)^2}\frac{\Phi_\rho(x)\Phi_\rho(y)}{\sqrt{xy}}\left(\int_{x\vee y}^{x+y}\tfrac1z\chi(z)\dd z-\int_{|x-y|}^{x\vee y}\tfrac1z\chi(z)\dd z\right)\dd x\dd y.
\end{multline}
Writing then $\Sigma_\chi={\rm supp}(\chi)$ and $\varsigma_\chi=\frac12\min(\Sigma_\chi)$, we first of all note that
\begin{multline}\nonumber
\left|\int_{(0,\infty)}\int_0^x\tfrac1z\chi(z)\dd z\Phi_\rho(x)\dd x\right|
\leq\Phi_\rho([2\varsigma_\chi,\infty))\times\int_{\Sigma_\chi}\tfrac1z|\chi(z)|\dd z\\
\leq\left(\Phi_\rho([2\varsigma_\chi,\infty))\times\|\tfrac1z\|_{L^q(\Sigma_\chi)}\right)\|\chi\|_{L^p(\Sigma_\chi)},
\end{multline}
with $p\in[1,\infty)$ and $q=\frac{p}{p-1}$, and similarly we find that
\begin{multline}\nonumber
\left|\iint_{[\varsigma_\chi,\infty)^2}\frac{\Phi_\rho(x)\Phi_\rho(y)}{\sqrt{xy}}\left(\int_{x\vee y}^{x+y}\tfrac1z\chi(z)\dd z-\int_{|x-y|}^{x\vee y}\tfrac1z\chi(z)\dd z\right)\dd x\dd y\right|\\
\leq\tfrac1{\varsigma_\chi}(\Phi_2([\varsigma_\chi,\infty)))^2\times2\int_{\Sigma_\chi}\tfrac1z|\chi(z)|\dd z\\
\leq\left(\tfrac2{\varsigma_\chi}(\Phi_2([\varsigma_\chi,\infty)))^2\times\|\tfrac1z\|_{L^q(\Sigma_\chi)}\right)\|\chi\|_{L^p(\Sigma_\chi)}.
\end{multline}
Noticing now that the term between brackets in the double integral on the right hand side of \eqref{eq:L42_1} vanishes on $\{x+y\leq2\varsigma_\chi\}$, by symmetry it remains only to estimate the integral over $(x,y)\in[\varsigma_\chi,\infty]\times(0,\varsigma_\chi]$. We thereto let $p\in(1,\infty)$ be arbitrary, $r\in(\frac{p}{p-1},\infty)$ large, and $q\in(1,\infty]$ such that $1=\frac1p+\frac1q+\frac1r$, and we obtain that
\begin{multline}\nonumber
\left|\iint_{[\varsigma_\chi,\infty)\times(0,\varsigma_\chi]}\frac{\Phi_\rho(x)\Phi_\rho(y)}{\sqrt{xy}}\left(\int_x^{x+y}\tfrac1z\chi(z)\dd z-\int_{x-y}^x\tfrac1z\chi(z)\dd z\right)\dd(x,y)\right|\\
\leq\int_{[\varsigma_\chi,\infty)}\frac{\Phi_\rho(x)}{\sqrt{x}}\dd x\times\int_{(0,\varsigma_\chi]}2\|\tfrac1z\chi(z)\|_{L^{\frac{r}{r-1}}(\Sigma_\chi)}y^{\frac1r-\frac12}\Phi_\rho(y)\dd y,
\end{multline}
which, using for the integral with respect to $y$ a dyadic decomposition and Lemma \ref{lm:integrallimitnearzero}, can be bounded by a constant times $\|\chi\|_{L^p(\Sigma_\chi)}$. Note here that the dependence on $\chi$ of the constants in the preceding estimates is limited to dependence on $\varsigma_\chi$. Combining thus the preceding estimates and using a density argument, we then find for any $p\in(1,\infty)$ and any $K\subset(0,\infty]$ compact that
\begin{equation}\label{eq:Lqdualityestimate}
\left|\int_{K}\chi(x)\Phi_\rho(x)\dd x\right|\leq C(\Phi_\rho,\min(K),p,\rho)\|\chi\|_{L^p(K)}\text{ for all }\chi\in L^p(K),
\end{equation}
with $C(\Phi_\rho,k,p,\rho)\leq c(\Phi_\rho,p,\rho)O(k^{-\frac1p}\Phi_2([k,\infty)))$ as $k\rightarrow\infty$. By duality it now follows that $\Phi_\rho\in\bigcap_{q\in(1,\infty)}L_{\rm loc}^q((0,\infty])$, and since $\Phi_\rho\in\mathcal{X}_2$ is finite we have $\Phi_\rho\in L^1(0,\infty)$. Moreover, from the dependence on $\min(K)$ of the constant $C$ in \eqref{eq:Lqdualityestimate} we find for all $q\in[1,\infty)$ that $\|\Phi_\rho\|_{L^q(r,\infty)}\leq O(r^{\frac1q-1}\|\Phi_2\|_{L^1(r,\infty)})$ as $r\rightarrow\infty$. Note further that the contribution of integrals over lines to the double integral on the right hand side of \eqref{eq:selfsimilarprofileM+} is zero for Lebesgue integrable functions, so $\Phi_\rho$ actually satisfies \eqref{eq:selfsimilarprofileL1} for all $\varphi\in C_c^1([0,\infty))$.

For the remaining continuity claim we fix $\gamma\in(0,\frac12)$ and $[a,b]\subset(0,\infty)$ arbitrarily, and we start by showing that for any $\varphi\in C_c^\infty(\mathbb{R})$ with support in $[a,b]$ we have
\begin{equation}\label{eq:sobolevestimate}
\left|\int_\mathbb{R}\varphi'(x)\,x\Phi_\rho(x)\dd x\right|\leq C(\Phi_\rho,a,b,\gamma,\rho)\|\varphi\|_{H^\gamma(\mathbb{R})}.
\end{equation}
Indeed, by \eqref{eq:selfsimilarprofileL1} we immediately have 
\begin{multline}\nonumber
\left|\int_{(0,\infty)}\varphi'(x)\,x\Phi_\rho(x)\dd x\right|\leq(\rho-1)\|\Phi_\rho\|_{L^2(a,b)}\|\varphi\|_{L^2(\mathbb{R})}\\
+\rho\left|\iint_{S_1}\frac{\Phi_\rho(x)\Phi_\rho(y)}{\sqrt{xy}}\Delta_y^2\varphi(x)\dd x\dd y\right|
+\rho\left|\iint_{S_2}\Big[\cdots\Big]\dd x\dd y\right|,
\end{multline}
where we have split the double integral over the domains
\begin{equation}\nonumber
S_1=\{(y\vee(a-y))<x<y+b<2b\}\quad\text{and}\quad S_2=\left\{x>y>\left(\tfrac{x}{2}\vee b\right)\right\}.
\end{equation}
For the integral over $S_1$ we now first note for all $y\in[0,b]$ that
\begin{equation}\nonumber
\left|\int_{(y\vee(a-y),y+b)}\Phi_\rho(x)\Delta_y^2\varphi(x)\dd x\right|\leq2\|\Phi_\rho\|_{L^2(\frac{a}{2},2b)}\|\Delta_y^2\varphi\|_{L^2(\mathbb{R})},
\end{equation}
and since $\|\Delta_y^2\varphi\|_{L^2(\mathbb{R})}\leq C(\gamma)\|\varphi\|_{H^\gamma(\mathbb{R})}\,y^\gamma$ (cf.~\cite{S90}, \cite{T78,T83}) we obtain
\begin{equation}\nonumber
\left|\iint_{S_1}\Big[\cdots\Big]\dd x\dd y\right|\leq C(\Phi_\rho,a,b)C(\gamma)\sqrt{\frac2a}\int_{(0,b)}y^{\gamma-\frac12}\Phi_\rho(y)\dd y\times\|\varphi\|_{H^\gamma(\mathbb{R})},
\end{equation}
where the remaining integral with respect to $y$ is bounded (use Lemma \ref{lm:integrallimitnearzero} and a dyadic decomposition of the interval $(0,b)$). For the integral over $S_2$ we note that the second difference of $\varphi$ is now completely given by $\varphi(x-y)$. Applying thus H\"older's and Young's inequalities we obtain
\begin{multline}\nonumber
\left|\iint_{S_2}\Big[\cdots\Big]\dd x\dd y\right|\leq\frac1b\|\Phi_\rho\|_{L^2(b,\infty)}\left\|\int_{(b,\infty)}\Phi_\rho(y)|\varphi(z-y)|\dd y\right\|_{L^2(b,\infty)}\\
\leq C(\Phi_\rho,b)\|\Phi_\rho\|_{L^1(0,\infty)}\|\varphi\|_{L^2(\mathbb{R})},
\end{multline}
and it follows that \eqref{eq:sobolevestimate} holds. We lastly fix any $\zeta\in C_c^\infty((a,b))$ with $\zeta(x)=\frac1x$ for $x\in I_{a,b}:=[\frac{2a+b}{3},\frac{a+2b}{3}]$, and we set $\Theta(x):=\zeta(x)x\Phi_\rho(x)$. Given then any $\varphi\in C^\infty(\mathbb{R})$, we use \eqref{eq:sobolevestimate} to get
\begin{multline}\nonumber
\left|\int_\mathbb{R}\varphi'(x)\Theta(x)\dd x\right|
\leq\left|\int_\mathbb{R}(\zeta\varphi)'(x)\,x\Phi_\rho(x)\dd x\right|+\left|\int_\mathbb{R}\zeta'(x)\varphi(x)\,x\Phi_\rho(x)\dd x\right|\\
\leq C(\Phi_\rho,a,b,\gamma,\rho)\|\zeta\varphi\|_{H^\gamma(\mathbb{R})}+\|z\zeta'(z)\|_\infty\|\Phi_\rho\|_{L^2(a,b)}\|\varphi\|_{L^2(\mathbb{R})},
\end{multline}
and since $C^\infty(\mathbb{R})$ is dense in $H^\gamma(\mathbb{R})$ we obtain
\begin{equation}\nonumber
\left|\int_\mathbb{R}\varphi'(x)\Theta(x)\dd x\right|\leq C(\Phi_\rho,a,b,\gamma,\zeta,\rho)\|\varphi\|_{H^\gamma(\mathbb{R})}\text{ for all }\varphi\in H^\gamma(\mathbb{R}),
\end{equation}
hence $\Theta'\in H^{-\gamma}(\mathbb{R})=(H^\gamma(\mathbb{R}))'$. Therefore $\Theta\in H^{1-\gamma}(\mathbb{R})\subset C^{0,\frac12-\gamma}(\mathbb{R})$ (cf.~\cite[Thm.~19.6(b)]{DK10} and \cite[Sec.~2.7.1 Rk.~2]{T78} resp.) and since $\Theta=\Phi_\rho$ on $I_{a,b}$ we have $\Phi_\rho\in C^{0,\frac12-\gamma}(I_{a,b})$. We then complete the proof by observing that $\gamma\in(0,\frac12)$ was chosen arbitrarily, and that for any $K\subset(0,\infty)$ compact there are $a,b\in(0,\infty)$ such that $K\subset I_{a,b}$.
\end{proof}
We are now able to prove Propositions \ref{pr:regularity} and \ref{pr:rescaling}.
\begin{proof}[Proof of Proposition \ref{pr:regularity}]
By Lemma \ref{lm:regularity} we have that $\Phi_\rho\in\bigcap_{\alpha<\frac12}C^{0,\alpha}((0,\infty))$, and $\Phi_\rho$ satisfies \eqref{eq:selfsimilarprofileL1} for all $\varphi\in C_c^1([0,\infty))$. We will prove smoothness via a bootstrap argument, for which we need to rewrite \eqref{eq:selfsimilarprofileL1}. For any $\delta>0$, let $\eta_\delta\in C([0,\infty])$ be a nondecreasing function with $\eta_\delta=0$ on $[0,\delta)$ and $\eta_\delta=1$ on $(2\delta,\infty]$. For fixed $\varphi\in C_c^\infty((0,\infty))$ it then holds by dominated convergence that
\begin{multline}\label{eq:reg_1}
\int_{(0,\infty)}\left(\tfrac1\rho[x\varphi(x)]_x-\varphi(x)\right)\Phi_\rho(x)\dd x\\
=\lim_{\delta\rightarrow0}\iint_{\{x>y>0\}}\frac{\eta_\delta(x)\Phi_\rho(x)\eta_\delta(y)\Phi_\rho(y)}{\sqrt{xy}}\Delta_y^2\varphi(x)\dd x\dd y,
\end{multline}
and if $\delta<\frac14\min({\rm supp}(\varphi))$, then
\begin{multline}\label{eq:reg_2}
\iint_{\{x>y>0\}}\frac{\eta_\delta(x)\Phi_\rho(x)\eta_\delta(y)\Phi_\rho(y)}{\sqrt{xy}}\Delta_y^2\varphi(x)\dd x\dd y\\
=\int_{(0,\infty)}\left(\int_0^{x/2}\frac{\eta_\delta(y)\Phi_\rho(y)}{\sqrt{y}}\left[\frac{\Phi_\rho(x+y)}{\sqrt{x+y}}+\frac{\Phi_\rho(x-y)}{\sqrt{x-y}}-2\frac{\Phi_\rho(x)}{\sqrt{x}}\right]\dd y\right.\\
\left.+\int_{x/2}^\infty\frac{\Phi_\rho(y)\Phi_\rho(x+y)}{\sqrt{y(x+y)}}\dd y-2\frac{\Phi_\rho(x)}{\sqrt{x}}\int_{x/2}^x\frac{\Phi_\rho(y)}{\sqrt{y}}\dd y\right)\varphi(x)\dd x.
\end{multline}
Using then the local H\"older regularity of $\Phi_\rho$ and the integral estimate from Lemma \ref{lm:integrallimitnearzero}, we find that $\frac{\Phi_\rho(y)}{\sqrt{y}}[\frac{\Phi_\rho(x+y)}{\sqrt{x+y}}+\frac{\Phi_\rho(x-y)}{\sqrt{x-y}}-2\frac{\Phi_\rho(x)}{\sqrt{x}}]$ is integrable with respect to $y$ near zero, and we are able to take the limit $\delta\rightarrow0$ in the right hand side of \eqref{eq:reg_2}. Combining \eqref{eq:reg_1} and \eqref{eq:reg_2}, we thus obtain that $\Phi_\rho$ satisfies \eqref{eq:reg_3}, where the derivative on the left hand is still taken in the distributional sense.

Suppose now that $\Phi_\rho\in C^{k,\alpha}((0,\infty))$ for some $k\in\mathbb{N}_0$ and $\alpha\in(0,1)$. To show that $\Phi_\rho\in C^{k+1,\alpha-\epsilon}((0,\infty))$ for some arbitrarily small $\epsilon>0$, it then suffices to check that the right hand side of \eqref{eq:reg_3} is in $C^{k,\alpha-\epsilon}((0,\infty))$, and since the second and third terms are actually even more regular it is enough to check this for the first. Moreover, writing $f(x)=\frac{\Phi_\rho(x)}{\sqrt{x}}$ we observe that
\begin{equation}\nonumber
f(\tfrac12x)\left[f^{(\ell)}(\tfrac32x)+f^{(\ell)}(\tfrac12x)-2f^{(\ell)}(x)\right]\in C^{k-\ell,\alpha}((0,\infty))\text{ for }\ell=0,1,\ldots,k,
\end{equation}
and we can restrict ourselves to proving that $f\in C^{0,\alpha}((0,\infty))$ implies
\begin{equation}\label{eq:reg_4}
\int_0^{x/2}\frac{\Phi_\rho(y)}{\sqrt{y}}\big[f(x+y)+f(x-y)-2f(x)\big]\dd y=:F(x)\in C^{0,\alpha-\epsilon}((0,\infty).
\end{equation}
Let thereto $K\subset[k_1,k_2]\subset(0,\infty)$ be compact, and let $\kappa>0$ be a constant such that $|f(x)-f(y)|\leq\kappa|x-y|^\alpha$ for all $x,y\in[\frac12k_1,2k_2]$. For $x_1,x_2\in K$ with $x_1\leq x_2$ there then holds
\begin{multline}\nonumber
\left|F(x_1)-F(x_2)\right|
\leq\left|\int_{x_1/2}^{x_2/2}\frac{\Phi_\rho(y)}{\sqrt{y}}\Delta_y^2f(x_2)\dd y\right|\\
+\int_0^{x_1/2}\frac{\Phi_\rho(y)}{\sqrt{y}}\left|\Delta_y^2f(x_1)-\Delta_y^2f(x_2)\right|\dd y,
\end{multline}
where the first term on the right hand side is bounded by a constant times $|x_1-x_2|$. Writing further $\xi=\min\{\frac12x_1,x_2-x_1\}$, we find that
\begin{multline}\nonumber
\int_0^{x_1/2}\frac{\Phi_\rho(y)}{\sqrt{y}}\left|\Delta_y^2f(x_1)-\Delta_y^2f(x_2)\right|\dd y\\
\leq4\kappa\left(\int_0^{\xi}\frac{\Phi_\rho(y)}{\sqrt{y}}y^\alpha\dd y+|x_1-x_2|^\alpha\int_\xi^{x_1/2}\frac{\Phi_\rho(y)}{\sqrt{y}}\dd y\right),
\end{multline}
which, using dyadic decompositions of the domains of integration and the estimate from Lemma \ref{lm:integrallimitnearzero}, can be bounded by a constant times
\begin{equation}\nonumber
|x_1-x_2|^\alpha(1+\log|x_1-x_2|)\leq|x_1-x_2|^{\alpha-\epsilon}\text{ as }|x_1-x_2|\rightarrow0,
\end{equation}
with $\epsilon>0$ arbitrarily small, and we have shown \eqref{eq:reg_4}. By induction it then follows that $\Phi_\rho\in C^\infty((0,\infty))$.\\

To lastly prove our positivity claim we suppose that there is some $x\in(0,\infty)$ such that $\Phi_\rho(x)=0$. As we have $\Phi_\rho\geq0$ on $(0,\infty)$ there then holds $\Phi_\rho'(x)=0$, and it follows from \eqref{eq:reg_3} that $\Phi_\rho(y)\Phi_\rho(x+y)=0$ for all $y\in(0,\infty)$, and that $\Phi_\rho(y)\Phi_\rho(x-y)=0$ for all $y\in(0,x)$. As a consequence of the latter identity we find that $\Phi_\rho(\frac{x}{2})=0$, hence $\Phi_\rho(2^{-n}x)=0$ for all $n\in\mathbb{N}$ by induction. From the iterated first identity we therefore have
\begin{equation}\nonumber
\Phi_\rho(y)\Phi_\rho(2^{-n}x+y)=0\text{ for all }y\in(0,\infty)\text{ and all }n\in\mathbb{N},
\end{equation}
so local uniform continuity implies $(\Phi_\rho(y))^2=0$ for all $y\in K$ with $K\subset(0,\infty)$ compact, hence $\Phi_\rho=0$ on $(0,\infty)$, and we conclude that $\Phi_\rho$ is indeed either strictly positive or identically zero on $(0,\infty)$.
\end{proof}

\begin{proof}[Proof of Proposition \ref{pr:rescaling}.]
Since the rescaling statement is an easy exercise, we restrict ourselves to proving that $x\Phi_\rho(x)\in\mathcal{X}_\rho$. Moreover, in view of Remark \ref{rk:restriction} we restrict ourselves to the case $\rho\in(1,2)$, but see also Lemma \ref{lm:momentestimateforrhoequaltwo}.

Now, we first show that
\begin{equation}\label{eq:P1_19_1}
R^{\rho-2}\int_{[0,\infty)}\left(1\wedge\tfrac{R}{x}\right)x\Phi_\rho(x)\dd x=\frac\rho2\int_0^R\mathcal{I}[\Phi_\rho](r)r^{\rho-3}\dd r\text{ for all }R>0,
\end{equation}
with $\mathcal{I}$ given by \eqref{eq:defofcalI}. To that end we note that by continuity of $\Phi_\rho$ (cf.~Proposition \ref{pr:regularity}), after an approximation argument we can use $\varphi(x)=(r-x)_+$, with $r>0$ arbitrary, directly in \eqref{eq:selfsimilarprofileL1} to obtain
\begin{equation}\label{eq:P1_19_2}
\frac1\rho\left((\rho-2)\int_0^rx\Phi_\rho(x)\dd x+(\rho-1)r\int_r^\infty\Phi_\rho(x)\dd x\right)=\tfrac12\mathcal{I}[\Phi_\rho](r).
\end{equation}
Using then the estimate from Lemma \ref{lm:integrallimitnearzero} in a dyadic decomposition for the first integral on the left hand side of \eqref{eq:P1_19_2}, we find that for small $r>0$ the second term is dominant and of order $O(r)$ as $r\rightarrow0$. As a consequence the product of \eqref{eq:P1_19_2} and $r^{\rho-3}$ is integrable near zero, and for any $R>0$ we find that the right hand side of \eqref{eq:P1_19_1} equals
\begin{equation}\nonumber
\int_0^R\left[r^{\rho-2}\int_0^rx\Phi_\rho(x)\dd x+r^{\rho-1}\int_r^\infty\Phi_\rho(x)\dd x\right]_r\dd r,
\end{equation}
hence \eqref{eq:P1_19_1} holds.

We next note that $\mathcal{I}[\Phi_\rho]\geq0$, so the right hand side of \eqref{eq:P1_19_1} is nondecreasing as a function of $R$, hence the supremum over $R>0$ is given by the limit $R\rightarrow\infty$. Observing lastly that $(xy)^{-1/2}(x+y-r)_+\wedge(r-|x-y|)_+\leq1$ for all $x,y\geq0$ and $r>0$, we obtain the uniform bound $\mathcal{I}[\Phi_\rho]\leq\|\Phi_\rho\|_2^2$, hence the right hand side of \eqref{eq:P1_19_1} is bounded as a function of $R$, and we conclude that $\|x\Phi_\rho(x)\|_\rho<\infty$.
\end{proof}

\section{Power law asymptotics - Theorem \ref{tm:fattails}}\label{sec:fattails}
The proof of Theorem \ref{tm:fattails} is given after the following useful result.
\begin{lm}\label{lm:integrallimitatinfinity}
Given $\rho\in(1,2)$, if $\Phi_\rho\in\mathcal{X}_2$ satisfies \eqref{eq:selfsimilarprofileM+} for all $\varphi\in C_c^1([0,\infty))$, then the limits
\begin{equation}\nonumber
\lim_{R\rightarrow\infty}\frac{R^{\rho-1}}{2-\rho}\int_{(R,\infty)}\Phi_\rho(x)\dd x\quad\text{and}\quad\lim_{R\rightarrow\infty}\frac{R^{\rho-2}}{\rho-1}\int_{(0,R)}x\Phi_\rho(x)\dd x
\end{equation}
exist, and both equal $\|x\Phi_\rho(x)\|_\rho$.
\end{lm}
\begin{proof}
Recalling from the proof of Proposition \ref{pr:rescaling} that $\Phi_\rho$ satisfies \eqref{eq:P1_19_2} for all $r>0$, we find by a rearrangement of terms that
\begin{multline}\label{eq:L4_1_1}
R\int_{(R,\infty)}\Phi_\rho(x)\dd x-(2-\rho)\int_{[0,\infty)}\left(1\wedge\tfrac{R}{x}\right)x\Phi_\rho(x)\dd x=\tfrac\rho2\,\mathcal{I}[\Phi_\rho](R)\\
=(\rho-1)\int_{[0,\infty)}\left(1\wedge\tfrac{R}{x}\right)x\Phi_\rho(x)\dd x-\int_{(0,R)}x\Phi_\rho(x)\dd x\text{ for all }R>0.
\end{multline}
Also from the proof of Proposition \ref{pr:rescaling} [cf.~\eqref{eq:P1_19_1}], we know that
\begin{equation}\nonumber
\int_{[0,\infty)}\left(1\wedge\tfrac{R}{x}\right)x\Phi_\rho(x)\dd x=R^{2-\rho}\left(\|x\Phi_\rho(x)\|_\rho-\frac\rho2\int_R^\infty\mathcal{I}[\Phi_\rho](r)r^{\rho-3}\dd r\right),
\end{equation}
so we can rewrite the first equality in \eqref{eq:L4_1_1} as
\begin{multline}\nonumber
\frac{R^{\rho-1}}{2-\rho}\int_{(R,\infty)}\Phi_\rho(x)\dd x-\|x\Phi_\rho(x)\|_\rho\\
=\frac\rho2\left(\frac{\mathcal{I}[\Phi_\rho](R)R^{\rho-2}}{2-\rho}-\int_R^\infty\mathcal{I}[\Phi_\rho](r)r^{\rho-3}\dd r\right),
\end{multline}
where the right hand side tends to zero as $R\rightarrow\infty$. Doing the same for the second equality in \eqref{eq:L4_1_1} completes the proof.
\end{proof}
\begin{proof}[Proof of Theorem \ref{tm:fattails}.]
Recall that $\Phi_\rho\in C^\infty((0,\infty))$ (cf.~Proposition \ref{pr:regularity}).

We first remark that for any $r>0$ we have
\begin{multline}\label{eq:T2_1}
\left|\frac{\Phi_\rho(r)}{(2-\rho)(\rho-1)r^{-\rho}}-1\right|\leq\left|\frac{r\Phi_\rho(r)}{(\rho-1)\int_{(r,\infty)}\Phi_\rho(x)\dd x}-1\right|\\
+\frac{r\Phi_\rho(r)}{(\rho-1)\int_{(r,\infty)}\Phi_\rho(x)\dd x}\left|\frac{r^{\rho-1}}{2-\rho}\int_{(r,\infty)}\Phi_\rho(x)\dd x-1\right|,
\end{multline}
which by Lemma \ref{lm:integrallimitatinfinity} reduces the problem to showing that the first term on the right hand side of \eqref{eq:T2_1} vanishes as $r\rightarrow\infty$. Recall now from the proof of Proposition \ref{pr:rescaling} that $\Phi_\rho$ satisfies \eqref{eq:P1_19_2} for all $r>0$, and note that we may differentiate this equation with respect to $r$ to obtain
\begin{multline}\label{eq:T2_2}
(\rho-1)\int_{(r,\infty)}\Phi_\rho(x)\dd x-r\Phi_\rho(r)
=\tfrac\rho2\,\lim_{h\rightarrow0}\tfrac1h\big(\mathcal{I}[\Phi_\rho](r+h)-\mathcal{I}[\Phi_\rho](r)\big)\\
=\frac\rho2\left(\iint_{\substack{\{|x-y|<r,\\(x\vee y)>r\}}}\frac{\Phi_\rho(x)\Phi_\rho(y)}{\sqrt{xy}}\dd x\dd y-\iint_{\substack{\{x+y>r,\\(x\vee y)<r\}}}\frac{\Phi_\rho(x)\Phi_\rho(y)}{\sqrt{xy}}\dd x\dd y\right),
\end{multline}
Then, since $\int_{(r,\infty)}\Phi_\rho(x)\dd x\sim(2-\rho)r^{1-\rho}$ as $r\rightarrow\infty$ (cf.~Lemma \ref{lm:integrallimitatinfinity}), we find by \eqref{eq:T2_2} that the first term on the right hand side of \eqref{eq:T2_1} vanishes as $r\rightarrow\infty$ if
\begin{equation}\label{eq:T2_3}
\iint_{\{x+y>r,|x-y|<r\}}\frac{\Phi_\rho(x)\Phi_\rho(y)}{\sqrt{xy}}\dd x\dd y=o(r^{1-\rho})\text{ as }r\rightarrow\infty,
\end{equation}
so by proving \eqref{eq:T2_3} we prove the theorem. Now, the left hand side of \eqref{eq:T2_3} can be estimated by
\begin{equation}\label{eq:T2_4}
\tfrac4r\|\Phi_\rho\|_{L^1(\frac14r,\infty)}^2+2\int_{\frac34r}^{\frac54r}\frac{\Phi_\rho(x)}{\sqrt{x}}\left(\int_{|r-x|}^{\frac12r}\frac{\Phi_\rho(y)}{\sqrt{y}}\dd y\right)\dd x,
\end{equation}
where the first term is $O(r^{1-2\rho})$, which decays sufficiently fast. For the second term of \eqref{eq:T2_4} we use a dyadic decomposition of the interval $(|r-x|,\frac12r)$, and Lemma \ref{lm:integrallimitnearzero}, to find that the term between brackets can be bounded up to a constant by $\log(\tfrac{r}{|r-x|})$. H\"older's inequality further gives us the estimate
\begin{equation}\nonumber
\int_{\frac34r}^{\frac54r}\frac{\Phi_\rho(x)}{\sqrt{x}}\log\big|\tfrac{r}{r-x}\big|\dd x
\leq\frac{\|\Phi_\rho\|_{L^q(\frac34r,\infty)}\big\|\log\big|\frac{r}{r-z}\big|\big\|_{L^p(\frac34r,\frac54r)}}{\sqrt{\tfrac34r}},
\end{equation}
so recalling for $q\in(1,\infty)$ that $\|\Phi_\rho\|_{L^q(r,\infty)}\leq O(r^{\frac1q-1}\|\Phi_\rho\|_{L^1(r,\infty)})=O(r^{\frac1q-\rho})$ as $r\rightarrow\infty$ (cf.~proof of Lemma \ref{lm:regularity}), we find that the second term in \eqref{eq:T2_4} is bounded by a term of order $O(r^{\frac12-\rho})$ as $r\rightarrow\infty$, hence \eqref{eq:T2_3} holds.
\end{proof}

\section{Exponential bounds - Theorem \ref{tm:exponentialtails}}\label{sec:exponentialtail}
\subsection{A pointwise exponential upper bound}
Our first result gives an explicit upper bound to the moments of self-similar profiles, and can be seen as an improvement on \cite[Lm.~4.22]{KV15}.
\begin{lm}\label{lm:momentestimateforrhoequaltwo}
If $\Phi_2\in\mathcal{X}_2$ satisfies \eqref{eq:selfsimilarprofileM+} with $\rho=2$ for all $\varphi\in C_c^1([0,\infty))$, then there exists a finite constant $A>0$ such that
\begin{equation}\label{eq:momentestimatesforrohequaltotwo}
\int_{(0,\infty)}x^\gamma\Phi_2(x)\dd x\leq \gamma^\gamma A^{\gamma+1}\text{ for all }\gamma>0.
\end{equation}
\end{lm}
\begin{proof}
Let $r>0$ be fixed arbitrarily, and let $m_\gamma=\int_{(0,r)}x^\gamma\Phi_2(x)\dd x$ (for $\gamma\geq0$). To prove the result it suffices to show that there exists a finite constant $A>0$, independent of $r$, such that $m_\gamma\leq\gamma^\gamma A^{\gamma+1}$ for all $\gamma>0$.

We first recall from the proof of Proposition \ref{pr:rescaling} that $\Phi_2$ satisfies
\begin{equation}\label{eq:L6_1_1}
\int_{(r,\infty)}\Phi_2(x)\dd x=r^{-1}\mathcal{I}[\Phi_2](r),
\end{equation}
with $\mathcal{I}$ as defined in \eqref{eq:defofcalI}. Similar to the derivation of \eqref{eq:P1_19_2}, for any $\gamma>1$ we now approximate $\varphi_{r,\gamma}(x)=(r^\gamma-x^\gamma)_+$ by functions in $C_c^1([0,\infty))$ to obtain
\begin{multline}\label{eq:L6_1_2}
(1-\gamma)\int_{(0,r)}x^\gamma\Phi_2(x)\dd x+r^\gamma\int_{(r,\infty)}\Phi_2(x)\dd x\\
=2\iint_{\{x>y>0\}}\frac{\Phi_2(x)\Phi_2(y)}{\sqrt{xy}}\Delta_y^2\varphi_{r,\gamma}(x)\dd x\dd y.
\end{multline}
Introducing the notation $\phi_\gamma(x)=x^\gamma$, we note for $x\geq y\geq0$ that we have
\begin{equation}\nonumber
\Delta_y^2\varphi_{r,\gamma}(x)=
\begin{cases}
((x+y)^\gamma-r^\gamma)_+-\Delta_y^2\phi_\gamma(x)&\text{if }x\leq r,\\
(r^\gamma-(x-y)^\gamma)_+&\text{if }x>r.
\end{cases}
\end{equation}
Noticing further that
\begin{align}\nonumber
((x+y)^\gamma-r^\gamma)_+&\geq r^{\gamma-1}((x+y)-r)_+,\\\nonumber
(r^\gamma-(x-y)^\gamma)_+&\geq r^{\gamma-1}(r-(x-y))_+,
\end{align}
we find from \eqref{eq:L6_1_2}, where we use \eqref{eq:L6_1_1} in the left hand side, that
\begin{multline}\nonumber
(1-\gamma)\int_{(0,r)}x^\gamma\Phi_2(x)\dd x+r^{\gamma-1}\mathcal{I}[\Phi_2](r)\\
\geq r^{\gamma-1}\mathcal{I}[\Phi_2](r)-2\iint_{\{0<y<x<r\}}\frac{\Phi_2(x)\Phi_2(y)}{\sqrt{xy}}\Delta_y^2\phi_\gamma(x)\dd x\dd y,
\end{multline}
hence
\begin{equation}\label{eq:L6_1_3}
(\gamma-1)\int_{(0,r)}x^\gamma\Phi_2(x)\dd x\leq2\iint_{\{0<y<x<r\}}\frac{\Phi_2(x)\Phi_2(y)}{\sqrt{xy}}\Delta_y^2\phi_\gamma(x)\dd x\dd y.
\end{equation}

Since for $x\geq y\geq0$ there holds $\Delta_y^2\phi_2(x)=2y^2\leq2y\sqrt{xy}$, \eqref{eq:L6_1_3} now yields
\begin{multline}\label{eq:L6_1_4}
m_2\leq2\iint_{\{0<y<x<r\}}2y\,\Phi_2(x)\Phi_2(y)\dd x\dd y\\
=2\iint_{(0,r)^2}(x\wedge y)\Phi_2(x)\Phi_2(y)\dd x\dd y\leq2m_0m_1,
\end{multline}
so using H\"older's inequality, then \eqref{eq:L6_1_4}, and then again H\"older's inequality, we obtain for all $\gamma\in[0,2]$ that
\begin{multline}\nonumber
m_\gamma\leq m_0^{1-\frac\gamma2}m_2^{\frac\gamma2}=m_0^{1-\frac\gamma2}\left(m_2\right)^\gamma m_2^{-\frac\gamma2}\leq m_0^{1-\frac\gamma2}\left(2m_0m_1\right)^{\gamma}m_2^{-\frac\gamma2}\\=2^\gamma m_0^{1+\frac\gamma2}\left(m_1\right)^{\gamma}m_2^{-\frac\gamma2}\leq2^\gamma m_0^{1+\frac\gamma2}\left(m_0^{\frac12}m_2^{\frac12}\right)^{\gamma}m_2^{-\frac\gamma2}=\tfrac12\left(2m_0\right)^{\gamma+1},
\end{multline}
hence $m_\gamma\leq\gamma^\gamma A^{\gamma+1}$ for all $\gamma\in(0,2]$ if $A\geq2\|\Phi_2\|_2\geq2m_0$ (since $\gamma^\gamma>\frac12$).

For $n\in\mathbb{N}\cap(2,\infty)$ we use the binomial formula to note for $x\geq y\geq0$ that
\begin{equation}\nonumber
\Delta_y^2\phi_n(x)=
\sum_{j=2}^{n}(1+(-1)^j)\times{n\choose{j}}x^{n-j}y^j\leq2\sum_{j=2}^{n}{n\choose{j}}x^{n-j}y^{j-1}\times\sqrt{xy},
\end{equation}
which we then use in \eqref{eq:L6_1_3} to obtain
\begin{equation}\label{eq:binomial}
m_n\leq\frac{4}{n-1}\sum_{j=2}^{n}{n\choose{j}}m_{n-j}m_{j-1}.
\end{equation}
Supposing now that $m_\gamma\leq\gamma^\gamma A^{\gamma+1}$ for all $\gamma\in\mathbb{N}\cap(0,n)$ with some $A\geq2\|\Phi_2\|_2$, and since in particular $m_0\leq\|\Phi_2\|_2\leq A$, we use \eqref{eq:binomial} to find
\begin{equation}\label{eq:binomial2}
m_n\leq\left(\frac{4}{n-1}\sum_{j=2}^{n}{n\choose{j}}(n-j)^{n-j}(j-1)^{j-1}\right)A^{n+1},
\end{equation}
where we suppose that $0^0=1$. Also by the binomial formula, we note that
\begin{equation}\nonumber
{n\choose j}(n-j)^{n-j}j^j\leq\sum_{l=0}^n{n\choose l}(n-j)^{n-l}j^l= n^n\text{ for }j\in\mathbb{N}\cap(0,n],
\end{equation}
hence
\begin{equation}\label{eq:binomial3}
\frac{4}{n-1}\sum_{j=2}^{n}{n\choose{j}}(n-j)^{n-j}(j-1)^{j-1}\leq n^n\times4\left(\frac{1}{n-1}\sum_{j=2}^{n}\frac{(j-1)^{j-1}}{j^j}\right).
\end{equation}
Noticing then that the term between brackets on the right hand side of \eqref{eq:binomial3} is actually the average of terms that are all bounded by $\frac14$, it follows that the right hand side of \eqref{eq:binomial2} is bounded by $n^nA^{n+1}$, hence by induction we have that $m_\gamma\leq\gamma^\gamma A^{\gamma+1}$ for all $\gamma\in(0,2]\cup\mathbb{N}$ with any $A\geq2\|\Phi_2\|_2$.

Finally, suppose that $\gamma\in(2,\infty)\setminus\mathbb{N}$ and let $n$ be the smallest integer larger than $\gamma$. By H\"older's inequality we then have $m_\gamma\leq m_0^{1-\frac{\gamma}{n}}m_n^{\frac\gamma{n}}$, so with the above estimates on $m_n$ for $n\in\mathbb{N}$ we find that
\begin{equation}\nonumber
m_\gamma\leq A^{1-\frac{\gamma}{n}}\left(n^nA^{n+1}\right)^{\frac\gamma{n}}
=\gamma^\gamma\big(\tfrac{n}{\gamma}\big)^\gamma A^{\gamma+1}\leq\gamma^\gamma\left(\tfrac32A\right)^{\gamma+1}\text{ with any }A\geq2\|\Phi_2\|_2,
\end{equation}
whereby $m_\gamma\leq\gamma^\gamma(3\|\Phi_2\|_2)^{\gamma+1}$ for all $\gamma>0$, so \eqref{eq:momentestimatesforrohequaltotwo} holds with $A\geq3\|\Phi_2\|_2$.
\end{proof}
Mimicking the proof in \cite{NV14}, we are now able to prove the pointwise exponential upper bound.
\begin{pr}\label{pr:pointwiseupperbound}
If $\Phi_2\in\mathcal{X}_2$ satisfies \eqref{eq:selfsimilarprofileM+} with $\rho=2$ for all $\varphi\in C_c^1([0,\infty))$, then there exists a constant $a\in(0,1)$ such that $\|e^{ar}\Phi_2(r)\|_{L^\infty(1,\infty)}<\infty$.
\end{pr}
\begin{proof}
Recall first of all that $\Phi_2\in C^\infty((0,\infty))$ (cf.~Proposition \ref{pr:regularity}).

Now, let $A>\frac1{2e}$ be a constant such that \eqref{eq:momentestimatesforrohequaltotwo} holds, which exists by the proof of Lemma \ref{lm:momentestimateforrhoequaltwo}. For any $r>0$ there then holds
\begin{equation}\nonumber
\int_{(r,\infty)}\Phi_2(x)\dd x\leq r^{-\gamma}\int_{(0,\infty)}x^\gamma\Phi_2(x)\dd x\leq A\exp\left(\gamma\log\left(\tfrac{\gamma A}{r}\right)\right)\text{ for all }\gamma>0,
\end{equation}
where the right hand side is minimal if $\gamma=\frac{r}{eA}$, so
\begin{equation}\label{eq:Ppub_1}
\int_{(r,\infty)}\Phi_2(x)\dd x\leq A\exp\left(-\tfrac{r}{eA}\right)\text{ for all }r>0.
\end{equation}

We next differentiate \eqref{eq:P1_19_2} with respect to $r$, i.e.~we set $\rho=2$ in \eqref{eq:T2_2}, and we drop the double integral over $\{|x-y|<r,(x\vee y)>r\}$ to find for $r>0$ that
\begin{equation}\label{eq:Ppub_2}
r\Phi_2(r)\leq\int_{(r,\infty)}\Phi_2(x)\dd x+\iint_{\{x+y>r,(x\vee y)<r\}}\frac{\Phi_2(x)\Phi_2(y)}{\sqrt{xy}}\dd x\dd y.
\end{equation}
Since the integrand in the double integral on the right hand side of \eqref{eq:Ppub_2} is symmetric, for $r>1$ we can now estimate that term by
\begin{equation}\label{eq:Ppub_3}
2\iint_{\{x>\frac{r}{2},y>\frac12\}}\frac{\Phi_2(x)\Phi_2(y)}{\sqrt{xy}}\dd x\dd y+2\int_{r-\frac12}^r\frac{\Phi_2(x)}{\sqrt{x}}\left(\int_{r-x}^{\frac12}\frac{\Phi_2(y)}{\sqrt{y}}\dd y\right)\dd x.
\end{equation}
For $\xi\in(0,\frac12)$ we then let $n$ be the smallest integer such that $2^{-n-1}<\xi$, and we use Lemma \ref{lm:integrallimitnearzero} to obtain
\begin{equation}\label{eq:Ppub_4}
\int_{\xi}^\frac12\frac{\Phi_2(y)}{\sqrt{y}}\dd y\leq\sum_{j=1}^n\int_{2^{-j-1}}^{2^{-j}}\frac{\Phi_2(y)}{\sqrt{y}}\dd y\leq\sum_{j=1}^n\frac{C\cdot\sqrt{2^{-j}}}{\sqrt{2^{-j-1}}}\leq\tfrac{C\sqrt2}{\log2}|\log\xi|,
\end{equation}
so combining \eqref{eq:Ppub_2}, \eqref{eq:Ppub_3} and \eqref{eq:Ppub_4}, and using \eqref{eq:Ppub_1}, we find for $r>1$ that
\begin{equation}\label{eq:Ppub_5}
r\Phi_2(r)\leq e^{-\frac12\frac{r}{eA}}\left(Ae^{-\frac12\frac{r}{eA}}+\tfrac4{\sqrt{r}}A^2\right)+\frac{\frac{4C}{\log2}}{\sqrt{r}}\int_{r-\frac12}^r\Phi_2(x)|\log(r-x)|\dd x.
\end{equation}
Multiplying \eqref{eq:Ppub_5} by $\frac1re^{\frac12\frac{r}{eA}}$, and choosing $R\gg1$ sufficiently large, then yields
\begin{multline}\label{eq:Ppub_6}
e^{\frac12\frac{r}{eA}}\Phi_2(r)
\leq\frac1r\left(1+\frac{\frac{4C}{\log2}}{\sqrt{r}}\int_0^{\frac12}e^{\frac12\frac{x}{eA}}|\log x|\dd x\times\left\|e^{\frac12\frac{z}{eA}}\Phi_2(z)\right\|_{L^\infty(r-\frac12,r)}\right)\\
\leq\frac1r\left(1+\left\|e^{\frac12\frac{z}{eA}}\Phi_2(z)\right\|_{L^\infty(1,r)}\right)\text{ for all }r>R,
\end{multline}
so setting $a=\frac12\frac1{eA}\in(0,1)$, using \eqref{eq:Ppub_6}, and iterating, we obtain
\begin{multline}\label{eq:Ppub_7}
\|e^{az}\Phi_2(z)\|_{L^\infty(1,r)}
\leq\|e^{az}\Phi_2(z)\|_{L^\infty(1,R)}+\tfrac1R\left(1+\|e^{az}\Phi_2(z)\|_{L^\infty(1,r)}\right)\\
\leq\tfrac{R}{R-1}\left(\|e^{az}\Phi_2(z)\|_{L^\infty(1,R)}+\tfrac1R\right)\text{ for all }r>R.
\end{multline}
The claim now follows since the right hand side of \eqref{eq:Ppub_7} is independent of $r$.
\end{proof}

\subsection{An exponential lower bound in integral form}
We will prove the following result, of which the lower bound is a corollary.
\begin{pr}\label{pr:averagedlowerbound}
If $\Phi_2\in\mathcal{X}_2$ satisfies \eqref{eq:selfsimilarprofileM+} with $\rho=2$ for all $\varphi\in C_c^1([0,\infty))$, and if $\Phi_2$ is not identically zero on $(0,\infty)$, then there exists a finite constant $B>1$ such that
\begin{equation}\nonumber
\inf_{R\geq0}\left\{\int_{(R,R+1)}e^{Bx}\Phi_2(x)\dd x\right\}>0.
\end{equation}
\end{pr}
Supposing that $\Phi_2\in\mathcal{X}_2$ is as assumed in Proposition \ref{pr:averagedlowerbound}, then $\Phi_2$ is smooth and strictly positive on $(0,\infty)$, and it satisfies \eqref{eq:selfsimilarprofileL1} in particular for all $\varphi\in C_c^1((0,\infty))$, which we can rewrite as
\begin{equation}\label{eq:aelb_dual}
\int_0^\infty\left[\tfrac12(x\varphi_x(x)-\varphi(x))-\int_0^x\frac{\Phi_2(y)}{\sqrt{xy}}\Delta_y^2\varphi(x)\dd y\right]\Phi_2(x)\dd x=0.
\end{equation}
Note that \eqref{eq:aelb_dual} also holds for $\varphi\in W_0^{1,\infty}((0,\infty))$, so given a function $\varphi:(0,t)\rightarrow W_0^{1,\infty}((0,\infty))$ with $\psi_s(s,\cdot)\in L^\infty(0,\infty)$ that satisfies
\begin{equation}\label{eq:aelb_dual_2}
\varphi_s(s,x)\leq-\tfrac12(x\varphi_x(s,x)-\varphi(s,x))+\int_0^x\frac{\Phi_2(y)}{\sqrt{xy}}\Delta_y^2[\varphi(s,\cdot)](x)\dd y
\end{equation}
for almost all $s\in(0,t)$ and $x\in(0,\infty)$, there holds
\begin{equation}\nonumber
\int_0^\infty\varphi(0,x)\Phi_2(x)\dd x\geq\int_0^\infty\varphi(s,x)\Phi_2(x)\dd x\text{ for all }s\in[0,t].
\end{equation}
We will construct such a function, for which we will need the following definition.
\begin{df}\label{df:subsolution}
Let $u$ be the solution to \eqref{eq:ide} with $\alpha=\frac12$ that satisfies $u(0,x)={\rm sgn}(x)$ for all $x\in\mathbb{R}$, i.e.~(cf.~Lemma \ref{lm:fractionalheatequation})
\begin{equation}\nonumber
u(s,x)=2\int_0^{\frac{x}{s^2}}v_{\frac12}(z)\dd z=:w(\tfrac{x}{s^2})\quad[\,s>0,\,x\in\mathbb{R}\,]
\end{equation}
where $w:\mathbb{R}\rightarrow(-1,1)$ is a bijection that is smooth, increasing, odd, and concave on $\mathbb{R}_+$, since $v_\frac12$ is nonincreasing there. Let further $c_2>c_1>0$ be such that $w(c_1)=\frac12$ and $w(c_2)=\frac34$. For any $s>0$ the convex mappings $x\mapsto\frac12(\frac{x}{c_1s^2}\wedge1)$ and $x\mapsto\frac34(\frac{x}{c_2s^2}\wedge1)$ then lie below $u(s,\cdot)$.
\end{df}
\begin{figure}
\center
\begin{tikzpicture}[scale=4]
\draw(-.375,0)--(1.675,0);
\draw[dotted](0,0)--(0,1.125);
\draw[dotted](-.375,1)--(1.675,1);
\draw[dotted](-.375,.75)--(1.675,.75);
\draw[dotted](-.375,.5)--(1.675,.5);
\draw
(-.375,1) node[anchor=east] {1}
(-.375,.75) node[anchor=east] {$\frac34$}
(-.375,.5) node[anchor=east] {$\frac12$}
(-.375,.25) node[anchor=east] {$\frac14$};
\draw[domain=-0.25:1.5,smooth,variable=\x,dashed] plot (\x,{exp(\x)/4});
\draw[domain=0:1.5,smooth,variable=\x] plot (\x,1);
\draw[domain=0:1.5,smooth,variable=\x,thick] plot (\x,{2*rad(atan(4*\x))/pi});
\draw[domain=0:((1-cot(deg(8/pi)))/4),smooth,variable=\x] plot (\x,{.75*\x/((1-cot(deg(8/pi)))/4)});
\draw[domain=((1-cot(deg(8/pi)))/4):1.5,smooth,variable=\x] plot (\x,.75);
\draw[domain=0:.25,smooth,variable=\x] plot (\x,{.5*\x*4});
\draw[domain=.25:1.5,smooth,variable=\x] plot (\x,.5);
\draw[dotted](.1183128232890882,0)node[anchor=north east]{$\ell_1(s)$}--(.1183128232890882,.281399042248582);
\draw[dotted](.1183128232890882,0)--(.1183128232890882,-.0625);
\draw[dotted](.1183128232890882-.015625,-.0625)--(.1183128232890882,-.0625);
\draw[dotted](.25,0)node[anchor=north]{$c_1s^2$}--(.25,.5);
\draw[dotted]({(1-cot(deg(8/pi)))/4},0)node[anchor=north]{$c_2s^2$}--({(1-cot(deg(8/pi)))/4},.75);
\draw[dotted](1.252253608443016,0)node[anchor=north]{$\ell_2(s)$}--(1.252253608443016,.87455442344328);
\draw[dotted](1.386294361119891,0)node[anchor=north west]{$\frac1b\log4$}--(1.386294361119891,1);
\draw[dotted](1.386294361119891,0)--(1.386294361119891,-.0625);
\draw[dotted](1.386294361119891+.015625,-.0625)--(1.386294361119891,-.0625);
\end{tikzpicture}
\caption{\label{fig:1} For $s>0$ small and $b>1$ large, we have ${\rm sgn}(x)$, $u(s,x)$ (thick) and $\frac14e^{bx}$ (dashed), as well as $\frac12(\frac{x}{c_1s^2}\wedge1)$ and $\frac34(\frac{x}{c_2s^2}\wedge1)$.}
\end{figure}
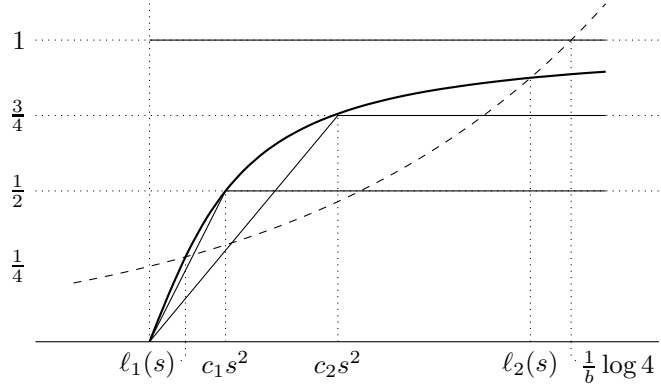
\begin{figure}
\center
\begin{tikzpicture}[xscale=5,yscale=2]
\draw(-.125,0)--(1.25,0);
\draw[dotted](1,0)--(1,3);
\draw[dotted](-.125,e)--(1.25,e)node[anchor=west]{$e^b$};
\draw[dotted](-.125,e*.75)--(1.25,e*.75)node[anchor=west]{$\frac34e^b$};
\draw[dotted](-.125,e*.5)--(1.25,e*.5)node[anchor=west]{$\frac12e^b$};
\draw[dotted](-.125,e*.25)--(1.25,e*.25)node[anchor=west]{$\frac14e^b$};
\draw[domain=0:1.125,smooth,variable=\x,dotted] plot (\x,{e*\x});
\draw[domain=.25:1.125,smooth,variable=\x,dotted] plot (\x,{e*(\x-.25)});
\draw[domain=.5:1.125,smooth,variable=\x,dotted] plot (\x,{e*(\x-.5)});
\draw[domain=0:1.125,smooth,variable=\x,dotted] plot (\x,{exp(\x)});
\draw[domain=0:1.125,smooth,variable=\x,dashed] plot (\x,{exp(\x)/4});
\draw[domain=.075:1,smooth,variable=\x] plot (\x,{e*\x});
\draw[domain=.125:1,smooth,variable=\x,thick] plot (\x,{e*(2*rad(atan(10*(1-\x)))/pi+(\x-1))});
\draw[domain=(1-(cot(deg(pi/8)))/10):1,smooth,variable=\x] plot (\x,{e*(\x-1)*(1-.75/((cot(deg(pi/8)))/10))});
\draw[domain=.325:(1-(cot(deg(pi/8)))/10),smooth,variable=\x] plot (\x,{e*(\x-.25)});
\draw[domain=.9:1,smooth,variable=\x] plot (\x,{e*(4-4*\x)});
\draw[domain=.575:.9,smooth,variable=\x] plot (\x,{e*(\x-.5)});
\draw[dotted](.951686022628324,0)node[anchor=north west]{$r_1(s)$}--(.951686022628324,.6475182253851713);
\draw[dotted](.951686022628324,0)--(.951686022628324,-.125);
\draw[dotted](.951686022628324+.0125,-.125)--(.951686022628324,-.125);
\draw[dotted](.9,-.2)node[anchor=north]{$1-c_1s^2$}--(.9,{e*.4});
\draw[dotted]({(1-(cot(deg(pi/8)))/10)},0)node[anchor=north]{$1-c_2s^2$}--({(1-(cot(deg(pi/8)))/10)},{e*((-(cot(deg(pi/8)))/10))*(1-.75/((cot(deg(pi/8)))/10))});
\draw[dotted](.1892611704921516,0)node[anchor=north]{$r_2(s)$}--(.1892611704921516,0.302089124584210);
\draw[dotted](.1018284310941419,0)node[anchor=north east]{$O(\frac1b)$}--(.1018284310941419,.276798373863700);
\draw[dotted](.1018284310941419,0)--(.1018284310941419,-.125);
\draw[dotted](.1018284310941419-.0125,-.125)--(.1018284310941419,-.125);
\end{tikzpicture}
\caption{\label{fig:2} For $s>0$ small and $b>1$ large, we have $e^b({\rm sgn}(1-x)+b(x-1))$, $e^b(u(s,1-x)+b(x-1))$ (thick), $\frac14e^{bx}$ (dashed) and $e^{bx}$ (dotted), as well as $e^b(\frac12(\frac{1-x}{c_1s^2}\wedge1)+b(x-1))$ and $e^b(\frac34(\frac{1-x}{c_2s^2}\wedge1)+b(x-1))$.}
\end{figure}
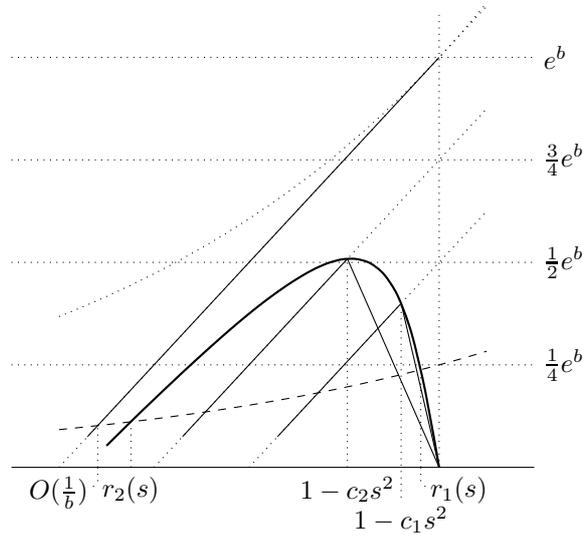
For fixed $R,b>1$ we first compare the functions $f^1(s,x)=e^{bR}u(s,x-R)$ and $f^2(s,x)=\frac14e^{bx}$. In particular we are interested in the solutions $x(s)$ to $f^1(s,x(s))=f^2(s,x(s))$, which for small $s>0$ are given by $x(s)=R+\ell_i(s)$, $i=1,2$, with $\ell_2(s)\geq \ell_1(s)>0$ the solutions to $u(s,\ell_i(s))=\frac14e^{b\ell_i(s)}$. From Figure \ref{fig:1} we then find that as long as $u(s,c_1s^2)=\frac12>\frac14e^{bc_1s^2}$ and $u(s,c_2s^2)=\frac34>\frac14e^{bc_2s^2}$ hold, i.e.~as long as $s^2<\frac1b\min\{\frac1{c_1}\log2,\frac1{c_2}\log3\}$, then there exist two different solutions $\ell_i(s)$, and there holds $\ell_2(s)-\ell_1(s)>(c_2-c_1)s^2>0$.

We then compare $f^3(s,x)=e^{b(R+1)}(u(s,R+1-x)+b(x-(R+1)))$ with $f^2$. We are interested in the solutions $x(s)$ to $f^3(s,x(s))=f^2(s,x(s))$, which for small $s>0$ are given by $x(s)=R+r_i(s)$, $i=1,2$, where $r_2(s)\leq r_1(s)<1$ are the solutions to $e^b(u(s,1-r_i(s))+b(r_i(s)-1))=\frac14e^{br_i(s)}$. Figure \ref{fig:2} now shows that as long as both $e^b(u(s,c_1s^2)-bc_1s^2)=\frac12e^b-be^bc_1s^2>\frac14e^{bc_1s^2}$ and $e^b(u(s,c_2s^2)-bc_2s^2)=\frac34e^b-be^bc_2s^2>\frac14e^{bc_2s^2}$ hold, i.e.~as long as $s^2<\frac1b\min\{\frac1{c_1}q(b,\frac12),\frac1{c_2}q(b,\frac34)\}$ with $q(b,\alpha)$ such that $e^b(\alpha-q(b,\alpha))=\frac14e^{q(b,\alpha)}$, then there exist two different solutions $r_i(s)$, and $r_1(s)-r_2(s)>(c_2-c_1)s^2>0$. Note further that indeed $r_2(s)\geq\frac1bq(b,1)=O(\frac1b)$ as $b\rightarrow\infty$.

For $b\gg1$ sufficiently large, we then define $f:(0,\frac1b)\rightarrow W_0^{1,\infty}((0,\infty))$ by
\begin{equation}\label{eq:f}
f(s,\cdot)=
\begin{cases}
f^1(s,\cdot)&\text{on }[R,R+\ell_2(s)),\\
f^2(s,\cdot)&\text{on }[R+\ell_2(s),R+r_2(s)],\\
f^3(s,\cdot)&\text{on }(R+r_2(s),R+1],\\
0&\text{else.}
\end{cases}
\end{equation}
The following lemma will be useful.
\begin{lm}\label{lm:L64}
Let $\Phi_2\in\mathcal{X}_2$ be as in the statement of Proposition \ref{pr:averagedlowerbound}. Then there exists a finite constant $C>0$ such that
\begin{equation}\nonumber
\int_z^\infty\frac{\Phi_2(e^{s/2}y)}{\sqrt{y}}\dd y\leq C\left((1+|\log z|)\wedge \frac1{\sqrt{z}}\right)\text{ for all }z>0\text{ and all }s\geq0.
\end{equation}
\end{lm}
\begin{proof}
We first note the trivial estimate that
\begin{equation}\label{eq:L64_1}
\int_z^\infty\frac{\Phi_2(e^{s/2}y)}{\sqrt{y}}\dd y\leq\frac{e^{-s/2}}{\sqrt{z}}\int_{ze^{s/2}}^\infty\Phi_2(y)\dd y\leq\frac{\|\Phi_2\|_2}{\sqrt{z}}\text{ for $z>0$ and $s\geq0$.}
\end{equation}
For $z\in(0,\frac12)$, let now $n$ be the smallest integer such that $2^{-n-1}<z$, so that using Lemma \ref{lm:integrallimitnearzero} we find
\begin{multline}\label{eq:L64_2}
\int_z^\frac12\frac{\Phi_2(e^{s/2}y)}{\sqrt{y}}\dd y
\leq\sum_{j=1}^n\int_{2^{-j-1}}^{2^{-j}}\frac{\Phi_2(e^{s/2}y)}{\sqrt{y}}\dd y
\leq\sum_{j=1}^n\frac{e^{-s/2}\int_0^{e^{s/2}2^{-j}}\Phi_2(y)\dd y}{\sqrt{2^{-j-1}}}\\
\leq\sum_{j=1}^n\frac{e^{-s/2}C\cdot\sqrt{e^{s/2}2^{-j}}}{\sqrt{2^{-j-1}}}\leq C\sqrt{2}\,n
\leq\tfrac{C\sqrt2}{\log2}|\log z|\text{ for all }s\geq0.
\end{multline}
The claim then follows by combining \eqref{eq:L64_1} and \eqref{eq:L64_2}.
\end{proof}
\begin{lm}\label{lm:L65}
Let $\Phi_2\in\mathcal{X}_2$ be as in the statement of Proposition \ref{pr:averagedlowerbound}. There then exist large constants $R_0,b_0\gg1$ such that if for $R\geq R_0$ and $b\geq b_0$ the function $f:(0,\frac1b)\rightarrow W_0^{1,\infty}((0,\infty))$ is given by \eqref{eq:f}, then for all $s\in(0,\frac1b)$ the function $\psi(s,x):=e^{-\log(b)s}f(s,x)$ satisfies
\begin{equation}\label{eq:L65_0}
\psi_s(s,x)\leq\int_0^x\frac{\Phi_2(e^{s/2}y)}{\sqrt{xy}}\Delta_y^2[\psi(s,\cdot)](x)\dd y\text{ for almost all }x\geq0.
\end{equation}
\end{lm}
\begin{proof}
Clearly, by nonnegativity of $\Phi_2$ and $\psi$, the right hand side of \eqref{eq:L65_0} is nonnegative if $x\in[0,R)\cup(R+1,\infty)$, while the left hand side is identically equal to zero. For $x\in[R,R+1]$ we now set $\bar{c}=c_2-c_1>0$, and we estimate the right hand side of \eqref{eq:L65_0} from below by
\begin{equation}\label{eq:L65_1}
\int_0^{\bar{c}/b^2}\frac{\Phi_2(e^{s/2}y)}{\sqrt{xy}}\Delta_y^2[\psi(s,\cdot)](x)\dd y-2\psi(s,x)\int_{\bar{c}/b^2}^\infty\frac{\Phi_2(e^{s/2}y)}{\sqrt{xy}}\dd y,
\end{equation}
where, for sufficiently large $R_0\gg1$, the second term is bounded from below, uniformly for all $R\geq R_0$, by $-\log b\times\psi(s,x)$ (use Lemma \ref{lm:L64} and $x\geq R_0$). By the smallness of the domain of integration in the first term of \eqref{eq:L65_1}, we further find that we can bound this term by
\begin{equation}\label{eq:L65_2}
e^{-\log(b)s}\int_0^{\bar{c}/b^2}\frac{\Phi_2(e^{s/2}y)}{\sqrt{xy}}\Delta_y^2[f^i(s,\cdot)](x)\dd y,
\end{equation}
where $i=1$ if $x\in[R,R+\ell_2(s))$, $i=2$ if $x\in[R+\ell_2(s),R+r_2(s)]$, and $i=3$ if $x\in(R+r_2(s),R+1]$. By the semigroup property of the exponential function this then means that, for $x\in(R+\ell_2(s),R+r_2(s))$, \eqref{eq:L65_1} can be bounded from below by
\begin{multline}\nonumber
\left(\int_0^{\bar{c}/b^2}\frac{\Phi_2(e^{s/2}y)}{\sqrt{xy}}\left(e^{by}+e^{-by}-2\right)\dd y-\log b\right)e^{-\log(b)s}\tfrac14e^{bx}\\
\geq-\log b\times\psi(s,x)=\psi_s(s,x).
\end{multline}
Note now, for $x\in(R,R+\ell_2(s))$, that $\Delta_y^2[f^1(s,\cdot)](x)=e^{bR}\Delta_y^2[u(s,\cdot)](x-R)\leq0$ for all $y\in\mathbb{R}$ (cf.~Lemma \ref{lm:oddfunctionsunderide}), so we can estimate the integral in \eqref{eq:L65_2} with $i=1$ from below by
\begin{multline}\label{eq:L65_4}
\frac1{\sqrt{R}}\int_{\mathbb{R}_+}\frac{\Phi_2(e^{s/2}y)}{\sqrt{y}}\Delta_y^2[f^1(s,\cdot)](x)\dd y\\
=\frac1{\sqrt{R}}\int_{\mathbb{R}_+}\left(\int_y^\infty\frac{\Phi_2(e^{s/2}z)}{\sqrt{z}}\dd z\right)\left(\int_{x-y}^{x+y}f^1_{ww}(s,w)\dd w\right)\dd y,
\end{multline}
where the equality follows by integration by parts and \eqref{eq:rewriteofsecdifprime} for $\partial_y[\Delta_y^2[f^1(s,\cdot)](x)]$, and where the integral with respect to $w$ in the right hand side of \eqref{eq:L65_4} is nonpositive for all $y\in\mathbb{R}_+$ (cf.~proof of Lemma \ref{lm:lowerboundfirstlemma}). By Lemma \ref{lm:L64}, and choosing $R_0\gg1$ sufficiently large, we then bound the right hand side of \eqref{eq:L65_4} from below, uniformly for all $R\geq R_0$, by
\begin{equation}\label{eq:L65_5}
\int_{\mathbb{R}_+}\frac{2}{\sqrt{y}}\left(\int_{x-y}^{x+y}f^1_{ww}(s,w)\dd w\right)\dd y=e^{bR}\int_{\mathbb{R}_+}y^{-\frac32}\Delta_y^2[u(s,\cdot)](x-R)\dd y,
\end{equation}
where for the equality we have integrated by parts back again, and noting that the right hand side of \eqref{eq:L65_5} by construction equals $e^{bR}u_s(s,x-R)=f_s^1(s,x)$ it follows that \eqref{eq:L65_1} can be estimated from below by
\begin{equation}\nonumber
e^{-\log(b)s}f_s^1(s,x)-\log(b)e^{-\log(b)s}f^1(s,x)=\psi_s(s,x).
\end{equation}
Recalling lastly that the second difference of an affine function is zero, similar arguments show that the inequality in \eqref{eq:L65_0} also holds for $x\in(R+r_2(s),R+1)$, which completes the proof.
\end{proof}
Now, for $R\gg1$ sufficiently large, let $\eta_R\in C^\infty(\mathbb{R})$ be such that ${\rm supp}(\eta_R)=[\frac35R+1,\frac45R]$, such that $\eta_R=1$ on $[\frac35R+2,\frac45R-1]$, and such that $\eta_R$ is increasing on $(\frac35R+1,\frac35R+2)$, and decreasing on $(\frac45R-1,\frac45R)$. For $b\gg1$ and $t\in(0,\frac1b)$ we then define $f^0\in C^\infty((0,\infty))$ by
\begin{equation}\label{eq:f0}
f^0(x)=\eta_R(x)\times\frac{e^{bx}}{8R}\inf_{r\in[0,\frac45R]}\left\{\int_r^{r+1}e^{be^{-t/2}y}\Phi_2(y)\dd y\right\}.
\end{equation}
\begin{lm}\label{lm:L66}
Let $\Phi_2\in\mathcal{X}_2$ be as in the statement of Proposition \ref{pr:averagedlowerbound}. There then exist large constants $R_0,b_0\gg1$ such that if for $R\geq R_0$ and $b\geq b_0$ the functions $f:(0,t)\rightarrow W_0^{1,\infty}((0,\infty))$ and $f^0\in W_0^{1,\infty}((0,\infty))$, with $t\in(0,\frac1b)$, are given by \eqref{eq:f} and \eqref{eq:f0}, then the function
\begin{equation}\label{eq:L66_0}
\psi(s,x):=e^{-\log(b)s}f(s,x)+\tfrac12se^{-\log(b)s}f^0(x)
\end{equation}
satisfies \eqref{eq:L65_0} for all $s\in(0,t)$.
\end{lm}
\begin{proof}
By Lemma \ref{lm:L65} we can restrict ourselves to $x\in(\frac35R+1,\frac45R)$ in checking that $\psi$ satisfies \eqref{eq:L65_0} for all $s\in(0,t)$. Similar to the first estimate in that lemma, we now note that
\begin{multline}\label{eq:L67_2_1}
\int_0^x\frac{\Phi_2(e^{s/2}y)}{\sqrt{xy}}\Delta_y^2[\psi(s,\cdot)](x)\dd y
\geq\int_{R+\ell_1(s)-x}^{R+r_1(s)-x}\frac{\Phi_2(e^{s/2}y)}{\sqrt{xy}}\psi(s,x+y)\dd y\\
+\int_0^{\bar{c}/b^2}\frac{\Phi_2(e^{s/2}y)}{\sqrt{xy}}\Delta_y^2[\psi(s,\cdot)](x)\dd y-2\psi(s,x)\int_{\bar{c}/b^2}^\infty\frac{\Phi_2(e^{t/2}y)}{\sqrt{xy}}\dd y,
\end{multline}
where the last term on the right hand side is bounded from below by $-\log b\times\psi(s,x)$ if $R_0\gg1$ is sufficiently large (cf.~proof of Lemma \ref{lm:L65}). We then estimate the first term on the right hand side of \eqref{eq:L67_2_1} from below by
\begin{multline}\nonumber
e^{-\log(b)s}\int_{R+\ell_1(s)-x}^{R+r_1(s)-x}\frac{\Phi_2(e^{s/2}y)}{\sqrt{xy}}\tfrac14e^{b(x+y)}\dd y\\
\geq e^{-\log(b)s}\times\frac{e^{bx}}{4R}\inf_{r\in[0,\frac25R]}\left\{\int_r^{r+(1-2c_1s^2)}e^{by}\Phi_2(e^{s/2}y)\dd y\right\}\\
\geq e^{-\log(b)s}\times\frac{e^{bx}}{8R}\inf_{r\in[0,\frac45R]}\left\{\int_r^{r+1}e^{be^{t/2}y}\Phi_2(y)\dd y\right\},
\end{multline}
where we have used that $(1-2c_1s^2)e^{s/2}=1+\frac12s+O(s^2)\geq1$ as $s\rightarrow0$. Note next that the second term on the right hand side of \eqref{eq:L67_2_1} can be written as
\begin{multline}\label{eq:L67_2_3}
\tfrac12se^{-\log(b)s}\int_0^{\bar{c}/b^2}\frac{\Phi_2(e^{s/2}y)}{\sqrt{xy}}\left(\eta_R(x+y)e^{by}+\eta_R(x-y)e^{-by}-2\eta_R(x)\right)\dd y\\
\times\frac{e^{bx}}{8R}\inf_{r\in[0,\frac45R]}\left\{\int_r^{r+1}e^{be^{t/2}y}\Phi_2(y)\dd y\right\}.
\end{multline}
Recalling now \eqref{eq:rewriteofsecdif}, the integral over the second difference in \eqref{eq:L67_2_3} can be bounded from below by
\begin{multline}\nonumber
-\frac1{\sqrt{\frac35R}}\int_0^{\bar{c}/b^2}\frac{\Phi_2(e^{s/2}y)}{\sqrt{y}}\left|\int_\mathbb{R}\left(y-|w|\right)_+\left[\eta_R(x+w)e^{bw}\right]_{ww}\dd w\right|\dd y\\
\geq-\frac1{\sqrt{\frac35R}}\int_0^{\bar{c}/b^2}y^{\frac32}\Phi_2(e^{s/2}y)\dd y\times\sup_{|w|<\frac{\bar{c}}{b^2}}\left|\left[\eta_R(x+w)e^{bw}\right]_{ww}\right|,
\end{multline}
which is $O(R^{-\frac12}b^{-2})$ as $R,b\rightarrow\infty$ (use Lemma \ref{lm:integrallimitnearzero} and a dyadic decomposition), hence bounded from below by $-1$ if $R_0,b_0\gg1$ are sufficiently large. Recalling lastly that $1-\frac12s\geq\frac12$ for $s\in(0,\frac1{b_0})$, we conclude with the above that the right hand side of \eqref{eq:L67_2_1} is bounded from below by
\begin{equation}\nonumber
\left[\tfrac12se^{-\log(b)s}\right]_s\times\eta_R(x)\times\frac{e^{bx}}{8R}\inf_{r\in[0,\frac45R]}\left\{\int_r^{r+1}e^{be^{t/2}y}\Phi_2(y)\dd y\right\}=\psi_s(s,x),
\end{equation}
and the proof is complete.
\end{proof}
\begin{lm}\label{lm:L67}
Let $\Phi_2\in\mathcal{X}_2$ be as in the statement of Proposition \ref{pr:averagedlowerbound}. There then exist constants $R_0,b_0\gg1$ and $c>0$ such that $\log(I(b_0,R_0))\geq\log(b_0)+1$, and such that for all $R\geq R_0$ and all $b\geq b_0$ there holds
\begin{equation}\label{eq:L67_0}
I(b,R)\geq t\left(I(be^{-t/2},\tfrac45R)\right)^2\text{ for all }t\in(0,\tfrac1b),
\end{equation}
where
\begin{equation}\label{eq:L67_00}
I(b,R):=c\times\inf_{r\in[0,R]}\left\{\int_r^{r+1}e^{bx}\Phi_2(x)\dd x\right\}.
\end{equation}
\end{lm}
\begin{proof}
Let $R_0,b_0\gg1$ be as obtained in Lemma \ref{lm:L66}, and let $R\geq R_0$, $b\geq b_0$ and $t\in(0,\frac1b)$ be fixed arbitrarily. We then define $\varphi:(0,t)\rightarrow W_0^{1,\infty}((0,\infty))$ as $\varphi(s,x):=e^{s/2}\psi(s,xe^{-s/2})$, with $\psi$ given by \eqref{eq:L66_0}, and we note that $\varphi$ satisfies \eqref{eq:aelb_dual_2} for almost all $(s,x)\in(0,t)\times(0,\infty)$ and $\varphi(0,x)\leq e^{bx}\mathbf{1}_{(R,R+1)}(x)$, hence
\begin{equation}\nonumber
\int_R^{R+1}e^{bx}\Phi_2(x)\dd x\geq\int_0^\infty\psi(t,xe^{-t/2})\Phi_2(x)\dd x,
\end{equation}
where the right hand side can be bounded from below by
\begin{multline}\nonumber
\int_{\frac23Re^{t/2}}^{\frac{11}{15}Re^{t/2}}e^{be^{-t/2}x}\Phi_2(x)\dd x\times\frac{\frac12t}{8R}\inf_{r\in[0,\frac45R]}\left\{\int_r^{r+1}e^{be^{-t/2}y}\Phi_2(y)\dd y\right\}\\
\geq\frac{R}{16}\frac{t}{16R}\left(\inf_{r\in[0,\frac45R]}\left\{\int_r^{r+1}e^{be^{-t/2}y}\Phi_2(y)\dd y\right\}\right)^2.
\end{multline}
Taking then the infimum, \eqref{eq:L67_0} follows with $c=\frac1{256}$. We lastly note that
\begin{equation}\nonumber
I(b_0,R_0)\geq ce^{\frac12b_0}\times\inf_{r\in[0,R_0]}\left\{\int_{r+\frac12}^{r+1}\Phi_2(x)\dd x\right\},
\end{equation}
where the logarithm of the right hand side is linear as a function of $b_0$, and it follows that indeed $\log(I(b_0,R_0))\geq\log(b_0)+1$ if $b_0\gg1$ is sufficiently large.
\end{proof}

\begin{proof}[Proof of Proposition \ref{pr:averagedlowerbound}]
Let $R_0,b_0\gg1$ and $c>0$ be as obtained in Lemma \ref{lm:L67}, and let $I$ be given by \eqref{eq:L67_00}. Now, set $B=b_0e^{\pi^2/12}$, and for all $n\in\mathbb{N}$ define $t_n=n^{-2}\wedge B^{-1}$, $b_n=b_{n-1}e^{t_n/2}$ and $R_n=\frac54R_{n-1}$. For every $n\in\mathbb{N}$ there then holds $b_n\leq b_0\exp(\frac12\sum_{j=1}^nj^{-2})<B$, hence $t_n\in(0,\frac1{b_n})$, so it follows from \eqref{eq:L67_0} that $I(b_n,R_n)\geq t_n(I(b_{n-1},R_{n-1}))^2$ for all $n\in\mathbb{N}$. Taking the logarithm and multiplying by $2^{-n}$, we now find for all $n\in\mathbb{N}$ by iteration that
\begin{multline}\nonumber
2^{-n}\log(I(b_n,R_n))
\geq2^{-(n-1)}\log(I(b_{n-1},R_{n-1}))+2^{-n}\log(t_n)\\
\geq\log(I(b_0,R_0))+\sum_{j=1}^n2^{-j}\log(t_j),
\end{multline}
from which we obtain that $\log(I(b_n,R_n))>0$ for all $n\in\mathbb{N}$, since $\log(I(b_0,R_0))\geq\log(b_0)+1$ (cf.~Lemma \ref{lm:L67}), and since
\begin{multline}\nonumber
\sum_{j=1}^n2^{-j}\log(t_j)\geq\sum_{j=1}^n2^{-j-1}\left(\log(j^{-2})+\log(B^{-1})\right)\\
\geq-\tfrac12\left(\log(b_0)+\tfrac{\pi^2}{12}\right)-\sum_{j=1}^\infty2^{-j}\log(j)>-\log(b_0)-1.
\end{multline}
The proof is then completed by the observation that
\begin{equation}\nonumber
\inf_{R\geq0}\left\{\int_{(R,R+1)}e^{Bx}\Phi_2(x)\dd x\right\}=\frac1c\times\inf_{n\in\mathbb{N}}I(B,R_n)\geq\frac1c\times\inf_{n\in\mathbb{N}}I(b_n,R_n),
\end{equation}
where the right hand side is strictly positive, as the infimum is taken over terms that are strictly larger than $1$.
\end{proof}
\begin{proof}[Proof of Theorem \ref{tm:exponentialtails}]
Corollary of Propositions \ref{pr:pointwiseupperbound} and \ref{pr:averagedlowerbound}.
\end{proof}
\begin{rk}
To deduce a pointwise exponential lower bound on $\Phi_2$, we seem to require existence of a solution $\varphi$ on $(s,x)\in[0,T)\times[0,\infty)$ to
\begin{equation}\label{eq:difficultequation}
\left\{\begin{array}{rl}
\varphi_s(s,x)\hspace{-8pt}&=-\tfrac12(x\varphi_x(s,x)-\varphi(s,x))+\frac1{\sqrt{x}}\int_0^x\frac{\Phi_2(y)}{\sqrt{y}}\Delta_y^2[\varphi(s,\cdot)](x)\dd y\\
\varphi(0,x)\hspace{-8pt}&=\delta_0(x-r)
\end{array}\right.
\end{equation}
for all $r\geq R_0$ with $R_0\gg1$ finite. However, existence of such solutions is nontrivial due to the possibly divergent behaviour of $\Phi_2$ near zero. A better understanding of well-posedness of \eqref{eq:difficultequation} requires a more detailed analysis of the asymptotics of $\Phi_2(z)$ as $z\rightarrow0$.
\end{rk}

\begin{ack}
The authors acknowledge support through the {\it CRC 1060 The mathematics of emergent effects} at the University of Bonn, that is funded through the German Science Foundation (DFG).
\end{ack}

\end{document}